%% file: main.tex
\documentclass[10.5pt]{article}
\usepackage[utf8]{inputenc}
\usepackage{authblk} 
\usepackage{graphicx}
\usepackage{amsmath,amsthm,amssymb} 
\usepackage{ying}
\usepackage[dvipsnames]{xcolor} 
\usepackage{cancel}

\providecommand{\keywords}[1]{\textbf{\textit{Keywords:}} #1}

\usepackage{mathtools}
\mathtoolsset{showonlyrefs}
\definecolor{ForestGreen}{RGB}{34,139,34}

\newcommand\rrevise[1]{{\color{black} #1}}
\newcommand\revise[1]{{\color{black} #1}}

\usepackage{geometry}
\usepackage{parskip}
\usepackage{tikz,tabularx}
\usepackage{subcaption}
\usetikzlibrary{patterns}
\usetikzlibrary{arrows}
\usetikzlibrary{tikzmark, positioning, fit, shapes.misc}

\usepackage[colorlinks,
            linkcolor=red,
            anchorcolor=blue,
            citecolor=blue
            ]{hyperref}

\def \test {\textrm{test}}
\def \mod {\textrm{mod}}
\def \ols {\textrm{ols}}
\def \ml {\textrm{ml}}
\def \full {\textrm{full}}
\def \Lasso {\textrm{Lasso}}
\def \partial {{\textrm{part}}}
\def \miss {\textrm{miss}}

\def \lm {\textrm{lm}}

\definecolor{ForestGreen}{RGB}{34,139,34}

\title{Modular Regression: Improving Linear Models \\ by Incorporating Auxiliary Data} 
\author{Ying Jin}  
\author{Dominik Rothenh\"ausler}
\affil{Department of Statistics, Stanford University}

\begin{document}

\maketitle

\begin{abstract}
This paper develops a new framework, called modular regression, 
to utilize auxiliary information --
such as variables other than the 
original features or 
additional data sets -- in the training process of linear models. 
At a high level, our method follows the routine: 
(i) decomposing the regression task into several  sub-tasks, (ii) fitting the sub-task models,
and (iii) using the sub-task models to provide an improved estimate for the original regression problem.
This routine applies  to 
widely-used low-dimensional (generalized) linear models 
and high-dimensional regularized linear regression.  
It also naturally extends to missing-data 
settings where only partial observations are available.
By incorporating auxiliary information, 
our approach improves the 
estimation efficiency and prediction accuracy  
upon linear regression or 
the Lasso under a conditional independence assumption 
for predicting the outcome. 
For high-dimensional settings, we develop an extension of our 
procedure  that is robust to violations of the conditional independence assumption, in the sense that it improves efficiency  
if this assumption holds and coincides with the Lasso otherwise. 
We demonstrate the efficacy of our 
methods with simulated and real data sets. 
\end{abstract}


\section{Introduction}
\label{sec:intro}

\input{intro}

\section{Modular linear regression in low dimensions}
\label{sec:lowd}
\input{sec2_lowd}

\section{Modular linear regression in high dimensions}
\label{sec:highd}
\input{sec3_highd}

\section{Extension to missing data}
\label{sec:partial}
\input{sec4_app.tex}

\section{Simulation studies}
\label{sec:simu}
\input{simu}

\section{Real data analysis}
\label{sec:real}
\input{real}

\section{Discussion}

In this work, we propose the 
modular regression framework and show that 
conditional independence structures between variables can be used to decompose statistical tasks into sub-tasks. 
We develop decomposition techniques for linear models 
in both low and high dimensional settings. 
We show that such decomposition can 
improve efficiency and 
allow to combine different datasets 
for a single estimation 
or prediction task with rigorous statistical guarantees. 
In practice, the conditional independence conditions for 
decomposition may be violated, leading to a 
bias-variance trade-off. 
We also develop a robust implementation of our method 
to adapt to potentially more complicated dependence structures. 

Looking ahead, statistical tasks that allow for 
decomposition may go well beyond the cases studied in this work, 
and the assumptions for decomposition may vary with the nature 
of the tasks. For instance, 
in high dimensional graphical models, 
the edges between variables that indicate independence may be sparse,  
and conditional independence may not hold exactly for 
two disjoint sets (e.g., our $X\in \RR^{p_x}$ and $Z\in \RR^{p_z}$). 
In biological applications, 
there may exist several paths from $X$ to $Y$ rather than 
being fully
mediated by $Z$. 
The dependence among the features and the response 
may still be sparse, but 
additional efforts are needed in order to  
leverage potential independence structures.
In addition, 
the data fusion technique may be further 
extended:  
In practice, one may have access to many auxiliary datasets 
that cover  different sets of features. Developing a framework 
to systematically combine multiple datasets 
may also be an interesting direction. 
Finally, our theoretical results in high dimensions 
only cover the sparse setting 
($p\gg n\gg \log p$ with $s\leq \sqrt{n}$). 
In this regime, both the LASSO and the modular estimator 
are consistent, and 
the role 
of conditional independence is to lower the variance 
of residuals. 
Its role in modern asymptotic regimes such as 
$p/n\to \rho$ for some fixed $\rho>0$ remains 
an interesting  question. 

\bibliographystyle{apalike}
\bibliography{references}

\newpage 
\appendix 
\input{appendix.tex}

\end{document}

%% file: intro.tex

Suppose for a patient  subject to 
a surgical procedure, we are interested in predicting their future health outcome $Y\in \RR$ 
in two years 
using some features $X\in \RR^{p_x}$ 
collected at the time of the surgery. 
The standard approach  is supervised learning:
The training data containing $n$
observations $\{(X_i,Y_i)\}_{i=1}^n$ 
from previous patients 
is used to pick a predictor $\hat{f}\colon \RR^{p_x}\to \RR$ that 
minimizes the average prediction error 
$\frac{1}{n}\sum_{i=1}^n \ell(f(X_i),Y_i)$ 
over $f\in \cF$ for some loss function $\ell\colon \RR\times\RR\to \RR$ 
and some function class $\cF$. 

However, the long duration 
of the study can pose various practical challenges to 
this paradigm. 
The number of joint observations of $(X,Y)$ may be very limited because they
take at least two years to collect.
The training data 
may also contain intermediate measurements $Z$, such as 
the health outcome one year after their surgery.  
In addition, there may be some patients for whom we 
only observe $(X,Z)$ which are easier to collect.  
It is also probable that 
another data set provides 
several $(Z,Y)$ samples 
from  patients who had surgery more than two years ago, 
but their features $X$ were missing 
due to limited technology. 
In this case, the standard  
approach of empirical risk minimization only over $(X,Y)$ samples
falls short as it cannot make use of 
auxiliary variables 
or additional partial observations. 

In the missing data literature, 
intermediate outcomes $Z$ are often 
utilized to estimate 
treatment effects
under a
conditional independence assumption (i.e.,
$X\indep Y\given Z$ in  our setting) 
when the outcome $Y$ is difficult to measure~\citep{prentice1989surrogate,begg2000use,athey2016estimating,athey2019surrogate}.  
Incorporating intermediate outcomes 
improves the estimation efficiency when joint observations 
of $(X,Y,Z)$ are available. It 
also offers more flexibility in data usage because it allows to 
incorporate 
$(X,Z)$ and $(Z,Y)$ samples. 
However, the conditional independence assumption may not 
hold in practice, and the existing methods are highly specialized
to estimating causal effects.

In this work, we aim to build a framework 
for flexibly incorporating 
auxiliary information into generic 
estimation and prediction procedures 
while maintaining rigorous guarantees. 
Such information may come 
from other variables $Z$ in addition to 
the original features and responses in the training process; 
it may also be additional data sets 
that only cover a subset of variables. 
We will extend the ideas of leveraging independence  
from the missing data 
literature 
to generalized and high-dimensional linear models, 
and develop more robust methods for general dependence patterns.
Let us begin with the tasks we consider.

\paragraph{A common algorithmic structure.}
We focus on 
statistical learning algorithms that take the form
\begin{equation}\label{eq:erm}
  \hat \theta = \argmin_{\theta \in \mathbb{R}^{p_x}} ~ \frac{1}{n}  \sum_{i=1}^n\big\{  h(X_i,\theta) - Y_i X_i^\top \theta \big\},
\end{equation}
where $h\colon \RR^{p_x} \times \RR^{p_x}\to \RR$  is a known function that is convex in $\theta$, and ${p_x}\in \NN$ is the 
dimension of the prediction features. 
The simplest example is ordinary least squares (OLS):
\begin{equation*}
    \hat \theta = \argmin_{\theta \in \mathbb{R}^{p_x}} ~ \frac{1}{n} \sum_{i=1}^n \Big\{  \frac{1}{2} \theta^\top X_i X_i^\top \theta - Y_i X_i^\top \theta \Big\}.
\end{equation*}  
Logistic regression with 
$Y\in\{0,1\}$ also satisfies~\eqref{eq:general}: 
\begin{equation*}\label{eq:general}
    \hat \theta = \argmin_{\theta \in \mathbb{R}^{p_x}}  ~ \frac{1}{n} \sum_{i=1}^n \big\{ \log(1+ \exp(X_i^\top \theta)) -  Y_i X_i^\top \theta \big\}. 
\end{equation*}  
We will later extend to other generalized linear models as 
we develop our method.  
The final example we consider is the 
$\ell_1$-regularized 
linear regression~\citep{tibshirani1996regression}: 
\#\label{eq:lasso}
\hat\theta  
= \argmin_{\theta \in \RR^{p_x}} 
\big\{ -
 {\textstyle \frac{1}{n}\sum_{i=1}^n } Y_iX_i^\top \theta + \theta^\top \big({\textstyle \frac{1}{2n}\sum_{i=1}^n }X_iX_i^\top\big) \theta + \lambda  \|\theta\|_1 \big\}.
\#

In the situation we discuss 
at the beginning, 
the supervised learning estimators of the form~\eqref{eq:general} 
may not be able 
to utilize auxiliary data. 
Instead of minimizing equation~\eqref{eq:erm},
we propose an alternative risk minimization criterion: 
 
\begin{equation}\label{eq:ours}
   \hat \theta =  \argmin_{\theta\in\RR^{p_x}} ~ \bigg\{ \frac{1}{n} \sum_{i=1}^n h(X_i,\theta) - \hat C^\top \theta \bigg\},
\end{equation}
where $\hat C$ is a proxy term for $C:= \mathbb{E}[XY]$
computed from the data.
We will see that the performance 
of $\hat\theta$ is closely related to the 
estimation accuracy of $\hat{C}$. 
At a high level, our idea is to 
replace $\frac{1}{n}\sum_{i=1}^n X_iY_i$ by 
$\hat{C}$ with
negligible bias 
and a lower variance, which translates 
to the improved performance of $\hat\theta$. 
We achieve this by developing a 
``modular'' estimator $\hat{C}$ -- whose meaning 
will be made precise soon -- 
that naturally allows for flexible incorporation of 
auxiliary information to improve the learning performance.
 
We emphasize that we still aim to minimize the same population risk (e.g., the population version of~\eqref{eq:erm}) as when using $\frac{1}{n}\sum_{i=1}^n X_iY_i $. 
As a result, the estimator still 
converges to the same limit $\theta^*$ that minimizes the population risk, that is, $ \theta^* =  \arg \min  \mathbb{E}[h(X,\theta) - Y X^\top \theta]$. Relatedly, our methods can be used for both estimation and prediction.

\subsection{Leveraging the dependence structure}
\label{subsec:ci}

The driving force of our approach 
is  to leverage the conditional 
independence structure among variables.
Conditional independencies are often used to improve performance over saturated models. 
As a classical example, suppose 
we have i.i.d.\ copies of a vector $(X,Y)$ obeying $X_{-S} \indep Y \given X_S$, 
that is, $X_{-S}$ is conditionally independent of $Y$ given $X_S$,  where  $X_S$ denotes a subset $S$ of features, and $X_{-S}$ stores the remaining ones. 
It is known that running a regression of $Y$ on $X_S$ can have lower mean-squared error than regressing $Y$ on $X$, since the latter may have very large variance (see, e.g., \citet{hastie2009elements} for more intuition). However, the former estimate may be biased if the conditional independence relation is violated. Furthermore, even if the conditional independence holds, the set $S$ is generally unknown. It is common to use model selection methods such as the best subset selection, Lasso, AIC, or BIC to navigate the bias-variance tradeoff. 

Perhaps less widely known is that  structures of the type 
$Y \indep X_S \given Z$  for auxiliary variables $Z$
can also be leveraged 
to improve estimation and prediction.
Analogous to classical model selection strategies, 
we want to derive a data-driven strategy to learn and exploit such structures. 
To fix ideas, let us consider 
an extreme case where $X\indep Y\given Z$.  

\paragraph{A naive approach.}
If $X \indep Y \given Z$, 
we can re-write 
\begin{equation}\label{eq:ci_naive}
    \mathbb{E}[X Y] = \mathbb{E}[\mathbb{E}[X\given Z] \mathbb{E}[Y \given Z]].
\end{equation}
Assuming access to i.i.d.~copies  $\{(X_i,Y_i,Z_i)\}_{i=1}^n$,
in view of~\eqref{eq:ci_naive}, 
we can (i) estimate $\mathbb{E}[Y|Z=z]$ via $\hat \mu_y(z)$, (ii) estimate $\mathbb{E}[X|Z=z]$ via $\hat \mu_x(z)$ and (iii) combine these two estimates 
to compute $\hat C = \frac{1}{n} \sum_{i=1}^n \hat \mu_y(Z_i) \hat \mu_x(Z_i)$ for $C= \mathbb{E}[XY]$. If $\hat \mu_x$ and $\hat \mu_y$ are accurate, 
then this estimate has smaller variance than 
$\frac{1}{n}\sum_{i=1}^nX_iY_i$, as 
\begin{equation}\label{eq:naive}
    \text{Var}( \hat \mu_y(Z) \hat \mu_x(Z) ) \approx \text{Var}( \mathbb{E}[Y \given Z] \mathbb{E}[X \given Z] ) = \text{Var}( \mathbb{E}[XY \given Z] ) \le \text{Var}(  XY ).
\end{equation}
However, this naive approach has a pressing issue 
that may hinder its performance: even if \eqref{eq:ci_naive} holds, 
there can be considerable bias 
if the estimation of $\hat\mu_y$ and $\hat\mu_x$ 
has slow convergence rates. 
In this case, $\hat{C}$ is not a good estimator for $C$ 
because the bias $\mathbb{E}[\hat C] - C$ can be comparable with, or even larger in large samples than,  the variance of $\frac{1}{n} \sum_{i=1}^n X_i Y_i$. 

Figure~\ref{fig:motivate} illustrates this point 
via a simple numerical example, where $Z\in \RR^{20}$, $X,Y\in \RR$, and the conditional expectation functions 
$\EE[Y\given Z]$ and $\EE[X\given Z]$ only involve the 
first three entries in $Z$. 
The goal is to compute the best linear predictor for $Y$, 
and the default choice is OLS.
We apply the above naive approach with $\hat\mu_y$ and $\hat\mu_x$ fitted by random forests in R, and show the bias and standard deviation scaled by $\sqrt{n}$  for various sample sizes $n$.

\begin{figure}[ht]
    \centering
    \includegraphics[width=4.5in]{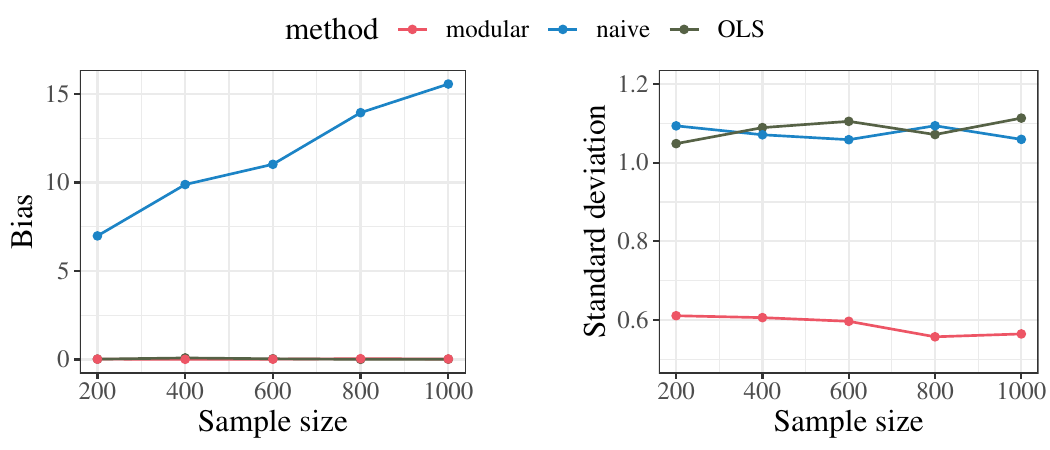}
    \caption{{Scaled by $\sqrt{n}$, the bias (on the left) and the standard deviation (on the right) are shown for plain OLS, the \texttt{naive} approach, and our to-be-introduced \texttt{modular} approach, across various sample sizes $n$.}} 
    
    \label{fig:motivate}
\end{figure}
Even in this simple example  where random forests should excel, the naive approach exhibits a significantly larger bias compared to plain OLS. This bias is larger than $n^{-1/2}$, and is due to the slow convergence of random forests. 
In addition, contrary to our prediction, 
the standard deviation of the naive approach 
is close to, 
instead of smaller than, the OLS, which is 
also due to the variability in training the
random forests.

\paragraph{A better approach.} The performance of our new 
approach is depicted in red in Figure~\ref{fig:motivate}: 
It has negligible bias and
reduces the standard deviation by more than $40\%$  compared with OLS at all sample sizes. 
\rrevise{Our new approach uses the same intuitions as mentioned earlier but additionally addresses bias using well-known semi-parametric techniques based on the mixed-bias property~\citep{rotnitzky2021characterization,robins2008higher}.} 
We further rewrite~\eqref{eq:ci_naive} as
\begin{equation} 
\mathbb{E}[X Y] = \mathbb{E}[X \mathbb{E}[Y \given Z]] + \mathbb{E}[\mathbb{E}[X|Z] Y] - \mathbb{E}[\mathbb{E}[X\given Z] \mathbb{E}[Y \given Z]].
\end{equation}
Again, an estimator can be obtained via
three sub-tasks: (i) estimating $\mathbb{E}[Y|Z=z]$ via $\hat \mu_y(z)$, (ii) estimating $\mathbb{E}[X|Z=z]$ via $\hat \mu_x(z)$, and (iii) combining these two estimates via
\begin{equation}\label{eq:indep}
\hat C = \frac{1}{n} \sum_{i=1}^n \big\{ \hat \mu_y(Z_i) X_i + Y_i \hat \mu_x(Z_i) - \hat \mu_y(Z_i)  \hat \mu_x(Z_i) \big\}.
\end{equation}
We will show that a slight variation of this approach has very favorable properties. Most importantly, 
even if $\hat \mu_y$ and $\hat \mu_x$ converge 
to the ground truth at slow speed,  
the bias of 
$\hat C$ in~\eqref{eq:indep} can still be negligible, 
and it has a lower variance than $\frac{1}{n} \sum_{i=1}^n X_i Y_i$:
\$
n\cdot \Var(\hat{C}) \approx \Var(XY) - \Var\big( (X-\mu_x(Z))\cdot(Y-\mu_y(Z))  \big) \leq \Var(XY).
\$
Indeed, it leads to the most efficient estimator 
if $X\indep Y\given Z$  (see Section~\ref{subsec:dr_semi} for details).
This approach also opens 
a door for data fusion:~\eqref{eq:indep} 
only involve pairwise observations 
of $(X,Z)$ or $(Z,Y)$. 
Thus, it can  flexibly use
additional data to improve
estimation accuracy, and works even when 
no joint observations of $(X,Z,Y)$ are available.

\subsection{A modular regression framework}

We propose a modular regression framework 
that concretize the above ideas by carefully 
dealing with the nuances in obtaining estimators 
$\hat\mu_x$, $\hat\mu_y$ to eliminate bias, 
navigating bias-variance tradeoffs under
potential violation of the conditional independence 
assumption, and developing principled algorithms 
for general estimators 
and data fusion scenarios.

As the name suggests, we 
decompose the estimation of $\mathbb{E}[XY]$
into smaller modules -- each being a sub-task 
that involves a subset of variables 
such as $(Y,X)$, $(Y,Z)$ or $(Z,X)$ -- and  
then re-assemble  the modules 
to construct $\hat C$. 
This idea is visualized
in Figure~\ref{fig:mod_ill}.  


 
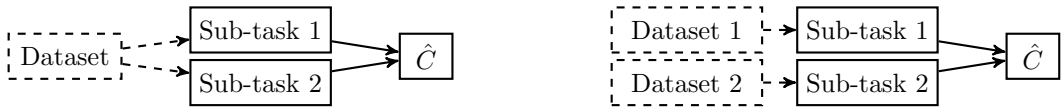
\begin{figure}[h]
\centering
\begin{subfigure}[t]{0.48\linewidth}
\centering
\begin{tikzpicture}[->,>=stealth', thick, main node/.style={circle,draw}]
\node[rectangle,draw, minimum width=1.5cm, minimum height=0.6cm] (3) at (0.8,1.25) {Sub-task 1};
\node[rectangle,draw, minimum width=1.5cm, minimum height=0.6cm] (4) at (0.8,0.55) {Sub-task 2};
\node[rectangle,draw, minimum width=0.7cm, minimum height=0.6cm] (5) at (3,0.9) { $\hat C$};
\node[rectangle,draw, dashed, minimum width=1.5cm, minimum height=0.6cm] (1) at (-1.8,0.9) {Dataset};
\draw[->] (1) edge [draw=black, dashed] (3);
\draw[->] (3) edge [draw=black] (5); 
\draw[->] (1) edge [draw=black, dashed] (4); 
\draw[->] (4) edge [draw=black] (5); 
\end{tikzpicture}  
\caption{Conditional independencies can be leveraged to decompose statistical tasks into subtasks.}
\end{subfigure}
\begin{subfigure}[t]{0.48\linewidth}
\centering
\begin{tikzpicture}[->,>=stealth', thick, main node/.style={circle,draw}]
\node[rectangle,draw, minimum width=1.5cm, minimum height=0.6cm] (3) at (0.8,1.25) {Sub-task 1};
\node[rectangle,draw, minimum width=1.5cm, minimum height=0.6cm] (4) at (0.8,0.55) {Sub-task  2};
\node[rectangle,draw, minimum width=0.7cm, minimum height=0.6cm] (5) at (3,0.9) {$\hat C$};
\node[rectangle,draw, dashed, minimum width=2cm, minimum height=0.6cm] (6) at (-1.6,1.25) {Dataset 1};
\node[rectangle,draw, dashed, minimum width=2cm, minimum height=0.6cm] (2) at (-1.6,0.55) {Dataset 2};
\draw[->] (6) edge [draw=black, dashed] (3);
\draw[->] (3) edge [draw=black] (5); 
\draw[->] (2) edge [draw=black, dashed] (4);
\draw[->] (4) edge [draw=black] (5); 
\end{tikzpicture}  
\caption{Modularity allows combining multiple datasets.}
\end{subfigure}
\caption{
Decomposing the estimation of $C=\mathbb{E}[YX]$ 
tackles two settings that may seem different at first sight. In (a), one can leverage conditional independencies to reduce statistical uncertainty. In (b), one can use additional data to increase  effective sample size.
} 
\label{fig:mod_ill}
\end{figure}
 
As shown in panel (a), 
our method adapts to specific structures, such as conditional independence between variables, 
so as to ensure 
a properly chosen decomposition.  
Such structure can also 
be learned from data.  
The decomposition also
allows data fusion; see panel (b).  
As long as 
a data set covers the variables in a sub-task, 
it can be incorporated to improve efficiency. 
We summarize our main results below.

\begin{itemize}
    \item \emph{Modular linear regression.}  
    We develop a generic approach for \revise{improving the estimation and prediction 
    in (generalized) linear models.  
    This is achieved by rewriting the estimation equation 
    to utilize auxiliary variables.} 
    Leveraging semi-parametric theory, 
    we show \revise{the optimality of our approach}
    under \revise{a conditional independence condition}. 
    
    \item \emph{High-dimensional modular regression.}
    By adding $\ell_1$-regularization, 
    our tools extend 
    to the high-dimensional setting. 
    We develop a proxy empirical risk that has negligible bias and lower variance compared with the Lasso. 
    We show that \revise{it improves the upper bounds for prediction accuracy.} 
    
    \item \emph{Robustness under conditional dependence.}  
    Since the conditional independence condition may be 
    violated in practice, 
    we develop an extension of our method that adapts to the dependency structures among variables, and we allow  
    such structures to be learned from data. 
    The proposed procedure interpolates between linear models (non-modular) 
    and the fully-modular approach (which is efficient under conditional independence). 
    The proposed procedure, bridging linear (non-modular) and fully-modular models, enhances robustness while still allowing for efficiency gains.
    
    \item \emph{Data fusion.} 
    We further extends
    method to combine multiple data sets. 
    This formulation allows using additional data on sub-tasks to improve overall statistical efficiency.
\end{itemize}

The rest of the paper is organized as follows. 
Section~\ref{sec:lowd} develops modular 
(generalized) linear regression 
in the fixed-$p$ setting. 
Section~\ref{sec:highd} studies modular regression 
for penalized linear regression in high dimensions, 
and develops a robust variant that can adapt to different dependence structures. 
Applications to data fusion and partial observations are discussed
in Section~\ref{sec:partial}. 
Finally, we evaluate our methods on simulated datasets in Section~\ref{sec:simu} 
and apply them to real datasets in Section~\ref{sec:real}.

\subsection{Related work}
\label{subsec:rel_work}

Modular regression combines evidence from sub-tasks, leveraging 
a ``modular structure'' provided by 
the conditional independence. 
%
One related strand is to use 
surrogates when the outcome of interest is expensive to measure~\citep{fleming1994surrogate,post2010analysis}. 
In particular, in   causal inference,  
a series of work~\citep{prentice1989surrogate,lauritzen2004discussion,chen2007criteria,vanderweele2013surrogate,athey2016estimating,athey2019surrogate,kallus2020role}  
have advocated using 
intermediate (short-term) outcome(s) as ``surrogates'' 
for long-term outcomes, and assumes various 
surrogate conditions  such as  
conditional independence of the long-term outcome and the treatment given 
the surrogates. 
Similar to our approach, 
the surrogate condition allows to decompose 
treatment effect estimation  into 
sub-tasks of estimating the effect of the treatment on the intermediate outcome 
and  estimating the effect of the intermediate outcome on the long-term outcome. 
We 
work under a generic conditional independence condition and 
mainly focus on  regression and prediction tasks. 
In addition, we propose a method that is robust to violations of  conditional independence.

Our framework is related to missing data and 
in particular data fusion, which combines data sets  that 
cover different subsets of variables. 
The preceding paragraph is also related; 
more generally,  there has been a surge of interest in combining different data sets, 
such as  combining experimental and observational data 
for treatment effect estimation~\citep{rosenman2022propensity,rosenman2020combining,colnet2020causal}, generalizing inference to 
different populations~\citep{dahabreh2019generalizing,hartman2015sample,jin2021attribute}, etc. 
Many of these works require identifiability assumptions for the target estimand, 
which are similar to the conditional independence condition. 
Compared to these works,
we study general tasks and enjoy robustness to 
failure of conditional independence.

In missing data scenarios, some 
other works study regression problem in settings that differ from ours. 
Those for high-dimensional data include
\cite{loh2011high} on
penalized regression,~\cite{lounici2014high,cai2016minimax} on the
estimation of covariance matrices, 
and~\cite{elsener2019sparse,zhu2019high} 
on sparse principal component analysis, which 
consider estimation when some 
covariates are missing at random. 
In contrast, we focus more on the data fusion 
aspect; 
we draw ideas from semiparametric statistics, 
and account for the dependence 
structure. A recent work of~\cite{cannings2022correlation} on U-statistics 
with low-dimensional data is related, which 
devises an estimator with 
smaller prediction MSE  by incorporating 
a partially missing data set.  
Instead, 
our approach leverages conditional 
independence structures among variables 
instead of correlation between estimators; 
we focus on regression 
and consider the high dimensional regime.

Our estimator  
is inspired by recent progress in 
efficient regression adjustment~\citep{henckel2019graphical,rotnitzky2019efficient} 
in low-dimensional graphical models. 
Our framework is more general as it applies to (generalized) 
linear models,  
works for high dimensional settings, does not require  the graphical structure, 
and imposes minimal model assumptions; 
it is also more restricted in that we do not 
consider choosing the optimal covariate sets. 

The modular estimator is  optimal for linear regression 
under  conditional independence, 
in the sense of semiparametric efficiency~\citep{bickel1993efficient,tsiatis2006semiparametric}. 
Thus, this work is also connected to the literature of 
 missing data.
Furthermore, the double robustness of our estimator is
similar to the AIPW estimator~\citep{robins1994estimation} and leverages the ``mixed-bias'' property~\citep{rotnitzky2021characterization,robins2008higher}.

The conditional independence structure has been used 
in a series of early works~\citep{sargan1958estimation,hansen1982large} 
to improve estimation efficiency; 
also related is~\cite{causeur2003linear} on 
linear regression under conditional independence, 
which is closely related to our 
first modular regression algorithm 
in low dimensions. However,  
these works assume strong parametric models such as a joint Gaussian distribution and independent Gaussian noise, 
while our method does not require stringent model assumptions.

\revise{Broadly speaking, our variance reduction idea may also apply to other statistical problems where the `modular' structure and conditional independence are present. For instance, sufficient dimensional reduction (SDR)~\citep{li2018sufficient} imposes $Y\indep X\given X^\top \Gamma$ for some unknown matrix $\Gamma$, whose estimation often relies on inverse regression algorithms~\citep{li1991sliced,adragni2009sufficient,ma2013efficient} that demonstrate a modular structure. Our idea may be adapted to enable more efficient inverse regression given an approximate solution for $\Gamma$.}

\paragraph{Notations.}

We use the standard $O_P(\cdot)$ and $o_P(\cdot)$ 
notation to denote smaller order 
in probability. 
For any two sequences $\{a_n\}$ and $\{b_n\}$ of positive numbers, 
we denote  $a_n=\Omega(b_n)$ if $\lim_{n\to \infty}a_n/b_n<\infty$ 
and $\lim_{n\to \infty} b_n /a_n<\infty$; 
we denote $a_n = O(b_n)$ if $a_n/b_n \leq C$ for some constant $C>0$.  
We use $\PP_{X,Y,Z}$ to denote the joint distribution of $(X,Y,Z)$. 
For any functions $f,g\colon \cZ\to \RR^d$, 
we denote their $L_2$-distance as $\|f-g\|_{L_2(\PP_Z)}^2 = \EE[\|f(Z)-g(Z))\|_2^2]$, 
where $\|\cdot\|_2$ is the Euclidean norm and 
the expectation is with respect to the distribution $\PP_Z$ of $Z$. For any vector $v \in \mathbb{R}^l$ and positive definite matrix $\Sigma \in \mathbb{R}^{l \times l} $ we set $\| v \|_\Sigma^2 = v^\top \Sigma v$.

%% file: sec2_lowd.tex

We first introduce modular regression in low-dimensional 
linear models, where
$X\in \RR^{p_x}$, $Z\in \RR^{p_z}$, and
$p_x$, $p_z$ do not grow with $n$ asymptotically.  

\begin{assumption}
\label{assump:ci}
The joint distribution of $(X,Z,Y)$ satisfies $X\indep Y\given Z$. 
\end{assumption}

In this section, we use this working assumption to show how conditional independencies can be leveraged to improve estimation accuracy. 
It will be relaxed later 
in Section~\ref{sec:robustness}, where 
we develop a method that is robust to violations of Assumption~\ref{assump:ci}.

Intuitively, Assumption~\ref{assump:ci} ensures that
$Z$ contains all the information 
in $Y$ that is relevant to $X$.  
While this assumption seems strong, it
is reasonable in many situations and 
has been widely used in a variety of applications. 
For example, $Z$ could be an intermediate outcome in the 
training data that is unavailable 
for the test samples, while $Y$ is a long-term outcome of interest. 
Intermediate outcomes are widely used as a proxy 
in estimating 
long-term treatment effects~\citep{athey2016estimating,athey2019surrogate}, 
while we focus on regression problems.

\subsection{Modular linear regression}
\label{subsec:modular_linear}

For simplicity of exposition, we start with linear (OLS) regression to show 
the benefits of 
a modular structure. 
Our goal is to predict $Y\in \RR$ by 
$\hat{Y} = X^\top \theta$ for some $\theta\in \RR^{p_x}$, 
but we do not necessarily assume a well-posed linear model. 
In the training process,   
we have observations of 
i.i.d.~triples $\{(X_i,Y_i,Z_i)\}_{i=1}^n$ for some 
auxiliary variables $Z_i \in \RR^{p_z}$, 
and $X_i\in \RR^{p_x}$ are random vectors (we view all vectors 
as column vectors throughout). 
The training and the target distributions 
are induced by the same distribution $(X,Y,Z)\sim \PP$.

Modular regression proceeds in three steps to leverage  conditional 
independence. 
\begin{enumerate}
    \item \emph{Decompose into sub-tasks.} First, the conditional independence structure in Assumption~\ref{assump:ci} allows us to decompose the estimation of $\EE[XY]$
    into sub-tasks that only involve 
    observations of $(X,Z)$ or $(Y,Z)$.
    \item \emph{Solve the sub-tasks.} We then learn the 
    conditional mean functions.
    We use cross-fitting~\citep{chernozhukov2018double}  
    to ensure desirable statistical properties: we 
randomly split $\cI=\{1,\dots,n\}$
into two disjoint folds $\cI_1$, $\cI_2$. Then 
for each fold $k=1,2$, 
we fit models $\hat\mu_x^{(i)}\colon \cZ\to \RR^{p_x}$ 
and $\hat\mu_y^{(k)} \colon \cZ\to \RR^{p_y}$ for  $\mu_x(z)=\EE[X\given Z=z]$ and 
$\mu_y(z) = \EE[Y\given Z=z]$ using data in 
$\cI\backslash\cI_k$. 
With slight abuse of notation, we write $\hat\mu_y(Z_i)=\hat\mu_y^{(k)}(Z_i)$ 
and $\hat\mu_x(Z_i)=\hat\mu_x^{(k)}(Z_i)$, $i\in \cI_k$.
    \item \emph{Assemble the sub-tasks.}  
Finally, we solve a modular 
least squares regression
\#\label{eq:modular_default}
\hat\theta_{n}^\mod = \argmin_{\theta\in \RR^{p_x}}~
\bigg\{\frac{1}{n}\sum_{i=1}^n 
  \theta^\top X_iX_i^\top \theta /2 - \hat{C}_{\lm}^\top \theta \bigg\}, 
\#
where we use a  proxy cross-term 
\#\label{eq:C_default}
\hat{C}_{\lm} = \frac{1}{n}\sum_{i=1}^n \big[ X_i \hat\mu_y (Z_i) + \hat\mu_x (Z_i)Y_i - 
\hat\mu_x (Z_i)\hat\mu_y (Z_i)\big].
\#
\end{enumerate}  
The procedure is summarized in Algorithm~\ref{alg:lowd}. 
\revise{When $X$ contains an intercept, 
the corresponding entry of $\hat{C}_{\lm}$ can be simply set as $\frac{1}{n}\sum_{i=1}^nY_i$.}
\begin{algorithm}[h]
\caption{Modular linear regression}\label{alg:lowd}
\begin{algorithmic}[1]
\REQUIRE Dataset $\cI =\{(Y_i,X_i,Z_i\}_{i=1}^{n}$.
\STATE Randomly split $\cI$ into equally-sized folds $\cI_k$, $k=1,2$.
\FOR{$k=1,2$} 
\STATE Fit models $\hat\mu_x^{(k)}(\cdot)$ for $\EE[X\given Z=\cdot]$ using observations $(X_i,Z_i)$, $i\notin \cI_k$; 
\STATE Fit models $\hat\mu_y^{(k)}(\cdot)$ for $\EE[Y\given Z=\cdot]$ using observations $(Y_i,Z_i)$, $i\notin \cI_k$; 
\STATE Compute $\hat\mu_x(Z_i) = \hat\mu_x^{(k)}(Z_i)$ and  $\hat\mu_y(Z_i) = \hat\mu_y^{(k)}(Z_i)$ for all $i\in \cI_k$.
\ENDFOR 
\STATE Compute the proxy cross-term  $\hat{C}_\lm$ as in~\eqref{eq:C_default}. 
\ENSURE The modular 
linear regression $\hat\theta_n^\mod$ as in~\eqref{eq:modular_default}. 
\end{algorithmic}
\end{algorithm}

\subsection{Double robustness and semiparametric efficiency}
\label{subsec:dr_semi}

The modular estimator $\hat\theta_n^\mod$ is 
doubly robust with respect to the estimation error of 
$\hat\mu_{x }$ and $\hat\mu_{y }$. 
To be more precise, the estimator 
converges at $n^{-1/2}$ rate to $\theta^*$ 
and is semiparametrically efficient even 
when $\hat\mu_{x }$ and $\hat\mu_{y }$
are consistent with slow nonparametric rates, that is if $\|\hat\mu_{x}^{(k)}-\mu_x\|_{L_2(\PP_Z)} = o_P(n^{-1/4})$ and $\|\hat\mu_{y}^{(k)}-\mu_y\|_{L_2(\PP_Z)} = o_P(n^{-1/4})$. 
The proof of the next theorem is in Appendix~\ref{app:subsec_thm_lowd}. 

\begin{theorem}[Double robustness and efficiency] \label{thm:lowd}
Suppose $\|\hat\mu_{x}^{(k)}-\mu_x\|_{L_2(\PP_Z)}\cdot\|\hat\mu_{y}^{(k)}-\mu_y\|_{L_2(\PP_Z)}=o_P(1/\sqrt{n})$ 
for $k=1,2$, 
$\EE[XX^\top]\succ 0$ is finite, 
and $X(Y-X^\top\theta^*)$ has finite 
second moment. Then 
$\sqrt{n}(\hat\theta_n^\mod - \theta^*)\stackrel{d}{\to}N(0,\Sigma^\mod)$ as $n\to \infty$, where $\Sigma^\mod=\Cov(\phi^\mod(X_i,Y_i,Z_i))$ 
for $\phi^\mod(x,y,z) = \EE[XX^\top]^{-1} \big(x\mu_y(z) + \mu_x(z) y - \mu_x(z)\mu_y(z) - xx^\top \theta^*\big)$. 
\revise{Further, suppose the  distribution of $(X,Y,Z)$ 
has density with respect to a base measure $\mu$. Then} $\phi^\mod$ is the (semiparametrically) 
efficient influence function
for estimating $\theta^*$.
\end{theorem}

 A short discussion of the implications of Theorem~\ref{thm:lowd} is in order.  
When $\hat\mu_x$ and $\hat\mu_y$ are estimated by nonparametric regression which typically comes with slow rates, Assumption~\ref{assump:ci} (conditional independence) is needed to ensure fast convergence of $\hat\theta_n^\mod$. 
However, when they are estimated by parametric regression (such as OLS)
and converge at $\sqrt{n}$-rate to any $\mu_x^*$ and $\mu_y^*$ (not necessarily equal $\mu_x$ or $\mu_y$), 
our estimator is consistent and asymptotically normal with $\sqrt{n}$-rate as long as $\EE [ (Y-\mu_y^*(Z)) (X-\mu_x^*(Z)) ] = 0$, i.e., the residuals of $Y$ and $X$ are uncorrelated.  This effect will also be visible in our simulations \rrevise{in Section~\ref{sec:simu}.}

Theorem~\ref{thm:lowd} shows that modular regression 
has the lowest asymptotic variance among all regular estimators~\citep{bickel1993efficient}. 
We now compare it with the OLS estimator  
$\hat\theta_n^\ols = \argmin_{\theta\in \RR^d} ~\sum_{i=1}^n ( Y_i - X_i^\top \theta )^2$, 
which only uses the information in $(X_i,Y_i)$ pairs.  
We know 
$
\hat\theta_n^\ols - \theta^*  
= \frac{1}{n}\sum_{i=1}^n \phi^\ols(X_i,Y_i) + o_P(1/\sqrt{n})  
$
with $\phi^\ols(x,y) = \EE[XX^\top]^{-1} (xy-xx^\top \theta^*)$
as $n\to\infty$. 
Under Assumption~\ref{assump:ci}, the 
asymptotic variance 
of $\phi^\mod(X_i,Y_i,Z_i)$ is dominated by that of $\phi^\ols(X_i,Y_i)$. 
Indeed, one can check that
\$
\Cov\big(\phi^\mod(X_i,Y_i,Z_i)\big) = \Cov\big(\phi^\ols(X_i,Y_i)\big) - 
\Sigma^{-1}\Cov\big( (Y-\mu_y(Z))(X-\mu_x(Z) \big)\Sigma^{-1},
\$
where $\Cov(A)$ denotes the covariance matrix 
for a random vector $A$, so that 
$\Cov\big( (Y-\mu_y(Z))(X-\mu_x(Z) \big)\succ 0$.
More efficient parameter estimation also translates to 
more accurate prediction: 
The prediction mean squared error (MSE) 
for an independent test sample $(X,Y)$ is  
$\EE[(Y-X^\top\hat\theta)^2\given \hat\theta]  
= \EE[ (Y-X^\top\theta^*)^2] + \|\hat\theta-\theta^*\|_{\Sigma}^2$, 
where the second term is smaller when $\hat\theta$ is 
the modular estimator instead of the OLS estimator.  
In our numerical experiments, 
we mostly focus on improving prediction in very noisy settings  with low sample size because 
the 
improvement in $\|\hat\theta-\theta^*\|_\Sigma^2$ 
is more pronounced in those cases.

\revise{
\begin{remark}[Alternative estimators]\label{rem:alternative}
    There are several natural alternatives to our modular estimator. The first is outcome regression, i.e., 
    running OLS of $\hat\mu_y(Z_i)$ over $X_i$, where
    $\hat\mu_y(\cdot)$ is an estimator for $\mu_y(\cdot)=\EE[Y\given Z=\cdot]$. 
    The second is orthogonal regression, i.e., 
    extracting the $X$-coefficients in OLS of $Y_i$ over $(X_i,Z_i^{\rm resid})$, where $Z^{\rm resid}:=Z_i-\mu_z(X_i)$, and 
    $\hat\mu_z(\cdot)$ is an estimator for $\mu_z(\cdot)=\EE[Z\given X=\cdot]$.
    However, both options  can be less efficient than our proposals. When $\hat\mu_y$ and $\hat\mu_z$ are estimated by flexible nonparametric  algorithms, both methods suffer from substantial bias that is comparable to or larger than the OLS variance (this is similar to the issue with our naive approach in Section 1); in contrast, our estimator achieves fast (product) convergence rate, and is the most efficient among all regular and asymptotically linear estimators. An extended discussion is in Appendix~\ref{app:discuss_alternative}. 
\end{remark}
}

Practitioners may be interested in quantifying the uncertainty of modular regression estimates. 
Wald-type confidence intervals can be derived 
using the asymptotic distribution in Theorem~\ref{thm:lowd}.  
\rrevise{The next corollary formalizes this result for the one-dimensional case; the multi-dimensional case follows similar ideas.}

\begin{corollary}[Confidence intervals]
    When $p_x=1$, set $\hat{\rm CI}:=\hat\theta_n^\mod \pm z_{1-\alpha/2} \hat\sigma_n^\mod /\sqrt{n}$, where 
$(\hat\sigma_n^\mod)^2 = \frac{1}{n}\sum_{i=1}^n (d_i-\bar{d} )^2$ for
$\bar{d} =\frac{1}{n}\sum_{i=1}^n  d_i $, and $d_i:= \hat\EE[XX^\top]^{-1} \big(X_i\hat\mu_y(Z_i) + \hat\mu_x(Z_i) Y_i - \hat\mu_x(Z_i)\hat\mu_y(Z_i) - X_iX_i^\top\hat\theta_n^\mod\big)$ from Algorithm~\ref{alg:lowd}. Here, $z_{1-\alpha/2}$ denotes the $1-\alpha/2$ quantile of a standard Gaussian random variable.
Then, under the conditions of Theorem~\ref{thm:lowd}, $(\hat\sigma_n^\mod)^2 \stackrel{P}{\to}  \Var(\phi^\mod(X,Y,Z))$, and thus $\PP(\theta^*\in \hat{\rm CI})\to 1-\alpha$ as $n\to \infty$.
\end{corollary}
 
One may also learn and adapt to the  
dependence structure from data  (see Section~\ref{sec:robustness} for a related discussion). In such cases, confidence intervals have to be adjusted to account for the variation induced by estimation of the structure. In practice, one can deal with this issue by data-splitting (i.e., use one fold of data for model selection and the second for estimation) and cross-fitting~\citep{chernozhukov2018double}.

\begin{remark}[Efficiency gain over OLS]\label{rem:linear}
To gain more intuition 
on the improvement in parameter estimation, 
we consider a one-dimensional special case 
where $Z=\alpha X+ \epsilon_z$,
$Y=\beta Z + \epsilon_y$  
for $\alpha,\beta\in \RR$. 
Here $X\sim N(0,\sigma_x^2)$, and $\epsilon_z\sim N(0,\sigma_z^2)$, 
$\epsilon_y\sim N(0,\sigma_y^2)$ 
are independent random noise. 
The true parameter for OLS regression is $\theta^*=\alpha \beta$. 
The OLS estimator $\hat\theta_n^\ols$ has asymptotic variance 
$
\sigma_{\ols}^2 
= \Var(X)^{-2} \Var(X\epsilon_y+X\beta\epsilon_z) 
= \Var(X)^{-1}(\sigma_y^2+\beta^2\sigma_z^2)
$, 
and our modular estimator $\hat\theta_n^\mod$ 
has asymptotic variance 
$
\sigma_\mod^2 
= \Var(X)^{-2}\Var(\mu_x(Z)\epsilon_y + X\beta\epsilon_z)
= \Var(X)^{-1}(\sigma_y^2 \frac{\Var(\mu_x(Z))}{\Var(X)}+\beta^2\sigma_z^2)$. 
That is, the absolute improvement in asymptotic variance is 
$\sigma_{\ols}^2-\sigma_{\mod}^2 = 
\sigma_y^2\frac{\Var(X-\mu_x(Z))}{\sigma_x^4}$; this quantity is large 
if $\sigma_y^2$ is large or $Z$ explains
a small proportion of the variance in $X$, i.e., $\Var(X-\mu_x(Z))$ is large compared to $\sigma_x^2$. The relative improvement in asymptotic variance is 
$1-\frac{\sigma_\mod^2}{\sigma_\ols^2} = \frac{1}{1+\alpha^2\sigma_x^2/\sigma_z^2}\frac{1}{1+\beta^2\sigma_z^2/\sigma_y^2}$; 
this quantity is large if $\alpha^2\frac{\sigma_x^2}{\sigma_z^2}$ 
and $\beta^2 \frac{\sigma_z^2}{\sigma_y^2}$ are small, which  
corresponds to (i) weak signal $\alpha,\beta$, or 
(ii) large noise $\sigma_z^2$ compared to $\sigma_x^2$, and 
$\sigma_y^2$ compared to $\sigma_z^2$. 
In the most extreme case where $\alpha=\beta=0$, 
the best prediction is $0$, and we achieve the asymptotic variance $\sigma_{\mod}^2=0$.  
\end{remark}

Our discussion so far is mainly asymptotic, 
so that the bias incurred by the estimation error of $\hat\mu_x$ and $\hat\mu_y$ is negligible compared to the $O(1/\sqrt{n})$ statistical error. 
Next, we use a simple example to provide 
insights on the finite sample behavior of 
$\hat{\theta}_n^{\mod}$ and $\hat{\theta}_n^{\ols}$. 

\begin{remark}[Finite-sample efficiency]
\label{rem:finite_sample}
We assume $(X,Y,Z)\in \RR^3$ are jointly Gaussian, and $\hat\mu_x$ and $\hat\mu_y$ 
are estimated by OLS.  Given OLS's $O_P(1/\sqrt{n})$ consistency, we forego cross-fitting to prevent cumbersome calculations.  As a result,  
$\hat\mu_x(z)=\hat{\EE}_n[Z^2]^{-1}\hat\EE_n[XZ]z$, 
and $\hat\mu_y(z)=\hat{\EE}_n[Z^2]^{-1}\hat\EE_n[YZ]z$, 
where $\hat\EE_n[\cdot]$ denotes empirical mean over all 
the data. One can then check that $\hat{\theta}_n^{\mod}  =\hat\EE_n[X^2]^{-1} \hat{\EE}_n[Z^2]^{-1}\hat\EE_n[YZ] \hat\EE_n[XZ]$ 
and $\hat{\theta}_n^{\ols}=\hat\EE_n[X^2]^{-1} \hat\EE_n[YX]$. 
By joint Gaussianity, there exists some $\alpha,\beta\in \RR$ such that $X=\alpha Z+\epsilon_x$ and 
$Y=\beta Z+ \epsilon_y$, where $\epsilon_x\sim N(0,\sigma_x^2)$ and $\epsilon_y\sim N(0,\sigma_y^2)$ conditional on $Z$. 
Both $\hat{\theta}_n^{\mod}$ and $\hat{\theta}_n^{\ols}$ 
are unbiased for $\theta^*$, while 
$\Var(\hat\theta_n^\ols)-\Var(\hat\theta_n^\mod) = \sigma_y^2 \EE\big[ \frac{  \hat{\EE}_n[Z^2]\cdot \hat{\EE}_n[\epsilon_x^2] - \hat{\EE}_n[Z\epsilon_x] ^2 }{n(\hat{\EE}_n[X^2])^2 \cdot \hat{\EE}_n[Z^2]}\big] \geq 0$. That is, in finite sample, $\hat{\theta}_n^{\mod}$ is always more 
efficient than $\hat{\theta}_n^{\ols}$. See detailed calculation and discussion in Appendix~\ref{app:detail_finite_sample}.
\end{remark} 

\subsection{Modular regression in generalized linear models}
\label{subsec:glm}

The ideas outlined in the previous section also apply to generalized linear models
(GLMs) as long as the estimation equation  has a modular structure. 

 Following the general setup in Section~\ref{sec:intro},  
we suppose 
the true parameter $\theta^*\in\Theta\subseteq\RR^{p_x}$ 
is the unique minimizer of 
$ 
\EE\big[ \ell(X_i,Y_i,\theta^*) \big]
$, where
$\ell(x,y,\theta)=YX^\top \theta  + h(x,\theta)  
$ 
for some function $h\colon \cX\times\Theta\to \RR^p$ 
that is convex in $\theta$. 
A default estimator $\hat\theta_n^{\ml}$ 
is the unique minimizer of the empirical loss~\eqref{eq:erm}. 
Maximum likelihood estimation for 
GLMs with $\log$ links 
satisfies this condition in general, such as  
logistic regression with
$\ell(x,y,\theta)=-x^\top \theta y + \log(1+\exp(x^\top \theta))$
for $y\in\{0,1\}$, and 
Poisson regression with
$\ell(x,y,\theta)=-x^\top \theta y + \exp(x^\top\theta) $
for $y\in \NN$. 
More generally, for  
exponential regression with 
$\ell(x,y,\theta)= 2\log y \cdot \theta^\top x - (\theta^\top x)^2$ 
for $y\in \RR^+$, one can use the transformed outcome  $\tilde{Y}:=2\log(Y)$  to make it a special case of equation~\eqref{eq:erm}.

To obtain a
more accurate estimator for $C=\EE[XY]$, 
the first two steps are exactly the same 
as in Section~\ref{subsec:modular_linear}.  
In the third step, 
we use  $\hat{C}_{\lm}$ 
defined in~\eqref{eq:C_default}, 
and compute 
\#\label{eq:def_mod_glm}
\hat\theta_n^\mod = \argmin_{\theta\in\Theta}~
\bigg\{\frac{1}{n}\sum_{i=1}^n h(X_i,\theta) + \hat{C}_{\lm}^\top \theta \bigg\}.
\# 
This  estimator 
is again doubly-robust to the error of $\hat\mu_x$ 
and $\hat\mu_y$, 
and has smaller asymptotic variance 
than $\hat\theta_n^\ml$ 
under mild conditions.
Its theoretical justification follows similar ideas 
as before with slightly more involved  technical conditions;
we defer all the results to Appendix~\ref{app:glm}.

%% file: sec3_highd.tex

Many prediction problems involve a huge number of predictive variables,
entering the high-dimensional regime where 
$p_x$ grows with, and can even be large than, the sample size $n$. 
A popular approach to 
dealing with high-dimensional data is penalized regression 
such as the Lasso~\citep{tibshirani1996regression}. 
By assuming a sparse linear regression function, 
high-dimensional regression has been shown 
to be consistent 
under various well-conditioning assumptions~\citep{candes2005decoding,candes2007dantzig,van2007deterministic,zhang2008sparsity}.  

Our modular regression idea extends naturally to the high-dimensional setting. 
In this section, we show that by including an $\ell_1$ penalty, modular regression 
improves upon the Lasso  
by seeking for a more efficient estimation equation.
Throughout this section,  
we assume 
a sparse linear model 
to ensure the task is tractable. 

\begin{assumption}
There exists some $\theta^*\in \RR^{p_x}$ 
with $\|\theta^*\|_0 = s$, such that 
$\EE[Y\given X]=X^\top \theta^*$. 
\end{assumption}

\subsection{Regularized modular regression}
\label{subsec:highd_method}

We start with cross-fitting, 
and let $\hat\mu_x^{(k)}$, $\hat\mu_y^{(k)}$ 
be estimators for $\mu_x(\cdot)=\EE[X\given Z=\cdot]$ 
and $\mu_y(\cdot)=\EE[Y\given Z=\cdot]$ 
obtained from $\cI\backslash \cI_k$, and define 
the cross-terms as in~\eqref{eq:C_default}. 
The only distinction from Section~\ref{subsec:modular_linear} is that we  encourage sparsity by $\ell_1$-regularization. 
We left
\#\label{eq:def_modular_highd}
\hat\theta_n^\mod = \argmin_{\theta\in \RR^{p_x}}~ \bigg\{ \frac{1}{n}\sum_{i=1}^n 
 \theta^\top X_iX_i^\top \theta/2  - \hat{C}_\lm^\top \theta  + \lambda \|\theta\|_1 \bigg\}
\#
for some regularization parameter $\lambda>0$, 
where $\|\theta\|_1=\sum_{j=1}^{p_x}|\theta_j|$. 
As~\eqref{eq:def_modular_highd} is convex in $\theta$, the optimization  can 
be efficiently solved similarly to  
the Lasso (e.g.~coordinate  descent). 
In practice, $\lambda$ can 
be determined by cross validation.

\subsection{Theoretical guarantee}

With proper choice of $\lambda$, 
modular regression
leads to a sharper upper bound 
on the estimation error. 
We assume a $\mu$-RSC property for the 
design matrix $X\in \RR^{n\times p_x}$ which is standard in the literature for the consistency of 
Lasso-type methods~\citep{bickel2009simultaneous,negahban2012unified}. 
Our theoretical analysis may be generalized 
to other conditions, which is 
beyond the scope of this work.

\begin{assumption}\label{assump:RSC}
$X\in \RR^{n\times p_x}$ obeys $\zeta$-Restricted Strong Convexity ($\zeta$-RSC), i.e., 
$\|X\Delta\|_2^2/n \geq \zeta\|\Delta\|_2^2$ 
for any $\Delta\in \CC_3$, where 
$\CC_3 := \big\{x\in \RR^p\colon \|x_{S^c}\|_1 \leq 3\|x_{S}\|_1\big\}$, and $S=\{j\colon \theta_j^*\neq 0\}$.
\end{assumption}

We assume entries of $X$ and the response $Y$ 
are bounded by constants. 
This condition is mild since 
the Lasso is often implemented after normalization. 

\begin{assumption}\label{assump:bounded}
$|X_j|\leq 1$ and 
$\|\hat\mu_{x,j}^{(k)}(\cdot)\|_{\infty}\leq 1$  almost surely for $k=1,2$ and all $1\leq j\leq p$. 
Also, $|Y_i|\leq c_0$ and 
$\|\hat\mu_y^{(k)}(\cdot)\|_\infty \leq c_0$ 
almost surely for $k=1,2$ for some constant $c_0>0$. 
\end{assumption}

Finally, we assume 
$\hat\mu_y^{(k)}$ and $\hat\mu_{x}^{(k)}$ are 
consistent with $o(n^{-1/4})$ convergence rates.

\begin{assumption}\label{assump:consistency}
Suppose $\|\hat\mu_y^{(k)}(\cdot)-\mu_y(\cdot)\|_{L_2(\PP_Z)} = o_P(n^{-1/4})$ for $k=1,2$. 
Also, there exists a sequence of constants $c_n \to 0$ such that 
for sufficiently large $n$,  
\#\label{eq:assump_consist}
\max_{1\leq j\leq p_x,k=1,2}\PP\bigg( \big\|\hat\mu_{x,j}^{(k)}(\cdot)-\mu_{x,j}(\cdot) \big\|_{L_2(\PP_Z)} \geq \frac{c_n\log(1/\delta)}{n^{1/4}}\bigg)  \leq \delta.
\#
Here $\hat{\mu}_{x,j}^{(k)}(\cdot)$ is the $j$-th entry 
of $\hat\mu_x^{(k)}$, an 
estimator for $\mu_{x,j}(\cdot):=\EE[X_j\given Z=\cdot]$. 
\end{assumption}

Assumption~\ref{assump:consistency} slightly 
differs from the commonly adopted consistency
conditions,
which is often of the form $\|\hat\mu_{x,j}^{(k)}-\mu_{x,j}\|_{L_2(\PP_Z)}=o_P(n^{-1/4})$, 
on the estimation of nuisance components. 
This is because 
we need an exponential tail bound to control 
all $2\times p_x$ estimated functions $\{\hat\mu_{x,j}^{(k)}\}_{j=1,k=1}^{p_x,2}$ simultaneously. 
Running  many regressions may 
incur considerable computational cost 
in high dimensions; we  provide a 
computational shortcut 
in Section~\ref{subsec:proj}.

Many estimators in the literature obey 
Assumption~\ref{assump:consistency}. 
When $Z$ is high-dimensional (i.e., $p_z$ 
may grow with $n$ or be larger than $n$), 
if $\EE[X_j\given Z]$ is an $s'$-sparse linear function of $Z$, 
then~\eqref{eq:assump_consist} 
holds for the Lasso estimator with 
$c_n \asymp \sqrt{s'\log(p_z)}/n^{1/4}$ 
when the regularization parameter is properly chosen.  
When $Z$ is low-dimensional (i.e., $p_z$ 
does not grow with $n$), 
under the standard assumption 
that $\EE[X_j\given Z]$ is sufficiently smooth, 
the well-established convergence results of 
sieve estimation can be turned into 
such bounds 
by exponential tail bounds 
for the concentration of empirical loss functions~\citep{chen1998sieve,chen2007large}.

We show an improved bound for the estimation error of $\hat\theta_n^\mod$ compared to the Lasso estimator using $(X_i,Y_i)$. 
The proof of Theorem~\ref{thm:highd} 
is in Appendix~\ref{app:subsec_thm_highd}. 

\begin{theorem}[Finite-sample bound]\label{thm:highd}
Suppose Assumptions~\ref{assump:ci},~\ref{assump:RSC},~\ref{assump:bounded}
and~\ref{assump:consistency} hold. 
Then there exists a sequence of constants $\{\bar c_n\}$ with $\bar c_n\to 0$ as $n\to \infty$, such that 
for any fixed $\delta\in(0,1)$ 
and 
any regularization parameter $\lambda$ obeying
\#\label{eq:thm_highd}
\lambda \geq \frac{2}{n} \big\| 
\mu_y^\top X +  Y^\top \mu_x -  \mu_y^\top \mu_x - (X\theta^*)^\top X
\big\|_\infty +  \frac{\bar{c}_n (\log(3p_x/\delta))^2}{\sqrt{n}},
\#
it holds with probability at least $1-\delta$ that 
\$
\|\hat\theta_n^{\textnormal \mod} - \theta^*\|_2 \leq \frac{3\lambda \sqrt{s}}{2\zeta}.
\$  
Here we denote $\|z\|_\infty=\sup_j|z_j|$, $Y,\mu_y \in \RR^n$ are vectors 
whose $i$-th entries are $Y_i$ and $\mu_y(Z_i)$, 
and $X,\mu_x\in \RR^{n\times p_x}$ are matrices 
whose $(i,j)$-th entries are $X_{ij}$ and $\mu_{x,j}(Z_i)$. 
\end{theorem}

The second term in~\eqref{eq:thm_highd} arises from 
the estimation error in $\hat\mu_y$ and $\hat\mu_x$. 
It enjoys a double robustness property that is 
similar to the low-dimensional case: 
under the slow convergence rate in Assumption~\ref{assump:consistency},
this error is negligible compared to the first term in~\eqref{eq:thm_highd} that is typically 
$\Omega(\sqrt{(\log p_x)/n})$~\citep{vershynin2018high}.  
The estimation error of $\hat\theta_n^\mod$ is then 
characterized by the first term  in~\eqref{eq:thm_highd}.

When $n$ is sufficiently large such that the 
second term in~\eqref{eq:thm_highd} is negligible, 
the deviation of the first term can be  
decided by the variance of each entry. 
We define the random vector 
$\epsilon^{\mod}:=X^\top\theta^* X-\mu_y(Z)X-Y\mu_x(Z)+\mu_y(Z)\mu_x(Z)\in\RR^{p_x}$. Then, in Theorem~\ref{thm:highd}, 
choosing 
$\lambda \asymp  L\sqrt{\max_j\Var(\epsilon_j^\mod)\cdot \log(p_x) /n}$ yields the estimation error 
\#\label{eq:bd_mod_simple}
\|\hat\theta_n^{\textnormal \mod} - \theta^*\|_2 \leq \frac{L\cdot \sqrt{s \log(p_x)}}{\sqrt{n}\mu} \sqrt{\max_j\Var(\epsilon_j^\mod)} ,
\#
where $L>0$ is a universal constant.  
On the other hand, 
we let $\hat\theta_n^\Lasso$ be the 
Lasso estimator~\eqref{eq:lasso}.   
Under similar conditions like Assumption~\ref{assump:RSC},
existing results in the literature~\citep{negahban2012unified} 
show that 
$\|\hat\theta_n^\Lasso - \theta^*\|_2 
\leq \frac{3\lambda \sqrt{s}}{2\mu}$ for any 
regularization parameter obeying
$ 
\lambda \geq  \big\| \frac{2}{n}\sum_{i=1}^n (Y_i-X_i^\top\theta^*)X_i \big\|_\infty.
$
We define $\epsilon^\Lasso = X(Y-X^\top\theta^*)$. 
Similarly, choosing 
$\lambda\asymp L\sqrt{\max_j\Var(\epsilon_j^\Lasso)\cdot \log(p_x) /n}$ yields 
\#\label{eq:bd_Lasso_simple}
\|\hat\theta_n^{\textnormal \Lasso} - \theta^*\|_2 \leq \frac{L\cdot \sqrt{s \log(p_x)}}{\sqrt{n}\mu} \sqrt{\max_j\Var(\epsilon_j^\Lasso)} ,
\#
where $L>0$ is the same universal scaling   as above. 
The bounds in~\eqref{eq:bd_mod_simple} 
and~\eqref{eq:bd_Lasso_simple} distinguish
our modular estimator 
from the Lasso, 
as one can check that $\Var(\epsilon_j^\mod) \leq \Var(\epsilon_j^\Lasso)$ for all $j$. 
That is, modular regression 
reduces the uncertainty  
by a constant order.

Though this is only an upper bound analysis, 
our numerical experiments later on confirm the 
improvement in estimation accuracy. In particular, 
we will see that 
the regularization parameter $\lambda$ 
chosen by cross-validation is smaller in 
modular regression. 
Intuitively, the reduction in $\text{Var}(\epsilon_j^\text{mod})$  allows cross-validation to choose a smaller $\lambda$ than for the Lasso. 

\revise{Interested readers may wonder whether 
the modular idea can be leveraged 
to construct more efficient confidence intervals 
following the  `debiasing' Lasso ideas~\citep{zhang2014confidence,javanmard2014confidence,van2014asymptotically}. 
While this may be feasible\footnote{Note the `modular' structure 
in the debiased estimator $\hat\beta^{\rm debias} = \hat\beta +\hat\Theta (X^\top Y-X^\top X\hat\beta)/n$, where $\hat\beta$ is the Lasso estimator, and $\hat\Theta$ is the estimated precision matrix. A natural and heuristic idea is to apply modular regression to improve the estimation of $X^\top Y$ and $\hat\beta$.}, this problem warrants a careful and rigorous investigation that is beyond the scope of this work. In addition, structure learning (which we introduce in the next part) may lead to irregular estimators that deserve extra care. In the current paper, we mainly focus on the 
estimation and prediction aspects, and leave this point for future work.}

\subsection{Robustness to the conditional independence condition}\label{sec:robustness}

In practice, Assumption~\ref{assump:ci} 
may be violated,  and  
the true dependence structure among the variables 
may be completely unknown.  
In this part, we generalize 
our modular regression framework  
to estimate and adapt to the 
conditional independence structure.  

To be precise, Assumption~\ref{assump:ci} 
posits that  
$X_j\indep Y\given Z$ for every $j$. 
That is, $Z$ captures all the predictive power 
of every $X_j$ for $Y$. 
This condition can be violated in many ways. 
For instance, 
$Z$ may capture all the information 
for a subset of variables in $X$, while 
others in $X$ still 
have direct effects for $Y$.
In this subsection, we 
assume that 
$X \indep Y \given (Z,X_{\cJ })$ for some subset $\cJ \subseteq \{1,\dots,p_x\}$,
and  
$X_{\cJ}$ 
is the vector containing $X_j$ for all $j\in \cJ $. 
The choice of $\cJ $ can be from 
prior knowledge, or   
by estimating (consistently) 
the dependence structure 
among all the variables when joint 
observations of $(X,Z,Y)$ are available.  
Learning the conditional independence structure is 
beyond our focus;  
popular methods in the literature include constraint-based~\citep{spirtes2000causation,margaritis1999bayesian,tsamardinos2006max}, 
score-based~\citep{lam1994learning,jordan1999learning,friedman2013learning}, 
and regression-based~\citep{lee2006efficient,meinshausen2006high,roth2004generalized,banerjee2006convex} 
ones assuming a high-dimensional graphical model, to name a few.

\begin{remark}[Recommendation in practice]\label{remark:learn-from-data}
For high-dimensional linear regression, 
a heuristic idea is to use the Lasso for 
structure learning. For example, 
one can regress $Y$ on $(X,Z)$ via the Lasso, 
and use all features in $X$ with 
nonzero estimated coefficient as $X_{\cJ }$. 
In our numerical experiments in Section~\ref{sec:simu},
We find that this heuristic approach works 
well in improving estimation accuracy.
\end{remark}

In the following, we outline how to use the output from  structure learning  such as in Remark~\ref{remark:learn-from-data} in conjunction with modular regression. 
The condition $X \indep Y \given (Z,X_{\cJ })$ can be seen as a special case of Assumption~\ref{assump:ci} 
for a different choice of 
the conditioning set ``$Z$'':  
\#\label{eq:general_ci}
X \indep Y \given Z^{\full},\quad \textrm{where} \quad Z^{\full} = (Z,X_{\cJ}).
\#
This again allows us to 
break the problem into sub-tasks. 
In learning   $\mu_y$, 
after data splitting, for each $k=1,2$, 
we aim $\hat\mu_y^{(k)}$ for $\EE[Y\given Z^\full]$ 
using data in $\cI\backslash \cI_k$,
instead of $\EE[Y\given Z]$. 
We only learn  $(\mu_x)_j$ for $j\notin \cJ $:  
we let 
$\hat\mu_{x,j}^{(k)}(x)$ be an estimate for $\EE[X_j\given Z^\full]$ 
using the data in $\cI\backslash \cI_k$ for $k=1,2$. 
Then for each $i \in \cI_k$, we compute 
the cross-term $C_i$ via
\#\label{eq:C_struc}
(C_i)_j = 
\begin{cases}
     (X_i)_j \hat\mu_y^{(k)}(Z_i^\full) + \hat\mu_{x,j}^{(k)}(Z_i^\full)Y_i - 
    \hat\mu_{x,j}^{(k)}(Z_i^\full)\hat\mu_y^{(k)}(Z_i^\full), \quad \text{if }j\notin \cJ , \\[1.5ex]
    (X_i)_j  Y_i , \quad \text{if }j\in \cJ  .
\end{cases}
\#
Finally, we solve the same modular least squares~\eqref{eq:modular_default} or penalized least squares~\eqref{eq:def_modular_highd} with 
$
\hat{C}_{\lm} := \frac{1}{n}\sum_{i=1}^n C_i
$
for the above $C_i$. 
When Assumption~\ref{assump:ci} is violated, 
the original $\hat{C}_{\lm}$ defined in~\eqref{eq:C_default}
might be biased even if the estimators for conditional 
expectations are correct. 
In contrast, the new cross-terms we derive in~\eqref{eq:C_struc} 
are unbiased for $(X_{i})_jY_i$ 
with potentially smaller variance 
under the generalized condition~\eqref{eq:general_ci}. 

Let $\hat{J}$ be the output of the structure learning step, and set $Z^{\hat{\full}}=(X_{\hat{J}},Z)$. We note that $Y\indep X\given Z^{\hat{\full}}$ holds with probability tending to 1 if $\hat{J}$ converges 
in probability to some superset $\tilde{J}\supseteq\cJ$, our asymptotic expansion in Theorems~\ref{thm:lowd} 
and~\ref{thm:highd}  holds  on an event with probability tending to 1, and theoretical guarantees for this approach  
can be directly implied in view of~\eqref{eq:general_ci} 
and the fact that the definition of $C_i$ in~\eqref{eq:C_struc}
is equivalent to taking $\hat\mu_{x,j}^{(k)}(Z_{i}^\full):=(X_i)_j$, which 
we omit here. 
\rrevise{The convergence of $\hat{J}$ can be obtained under, e.g.,~irrepresentability-type conditions~\citep{zhao2006model}, and is typically compatible with Assumption~\ref{assump:consistency}; see Appendix~\ref{app:subsec_conditions} for a more detailed discussion. 
}

This variant can be viewed as 
a data-driven interpolation between the 
fully modular regression we introduce 
in the preceding part
and the standard OLS or Lasso. 
If the conditional independence condition holds for $X_j$, 
then utilizing the information in $Z^\full$ 
reduces the estimation error in $[\hat\theta_n^\mod]_j$; 
otherwise, it reduces to 
$[\hat\theta_n^\mod]_j = [\hat\theta_n^\ols]_j$ in 
the low-dimensional case, 
and yields a similar bound as $\hat\theta_n^\Lasso$ for the 
high-dimensional regression.

\rrevise{In general, there is a trade-off between the robustness  (related to  accuracy of $\hat{J}$) and efficiency gain of $\hat\theta_n^\mod$. When $\hat{J}=J$, $\hat\theta_n^\mod$ achieves the most efficiency gain. When $\hat{J}=\{1,\dots,p\}$, then $\hat\theta_n^{\mod}$ reduces to $\hat\theta_n^{\mod}$ without any improvement.}

Our simulation studies in Section~\ref{sec:simu} 
show that this approach robustly improves the  
estimation and prediction accuracy  in cases where 
Assumption~\ref{assump:ci} is violated 
but the conditional independence structure~\eqref{eq:general_ci} 
may be learned from the data.



\subsection{Practical implementation via linear transformation}
\label{subsec:proj}

We now discuss 
a computational shortcut 
for high-dimensional modular regression. 
By using 
linear transformations for 
estimating the conditional mean functions, 
it reduces the computational costs 
and is readily compatible with standard implementation 
of the Lasso (e.g., \texttt{glmnet} R-package~\citep{glmnet}).  

First, let us discuss why the estimator as defined in equation~\eqref{eq:def_modular_highd} is computationally demanding. 
Note that the
modular estimator~\eqref{eq:def_modular_highd} minimizes
\#\label{eq:general_highd}
\frac{1}{n}\big\{ \theta^\top X^\top X\theta/2 - \hat\mu_y^\top X\theta - Y^\top \hat\mu_x\theta + \hat\mu_y^\top \hat\mu_x \theta \big\} + \lambda \|\theta\|_1
\#
where 
$\hat\mu_x \in \RR^{n\times p_x}$ is a matrix 
whose $(i,j)$-th entry stores 
an estimator for  $\mu_{x,j}(Z_i)$, 
and $\hat\mu_y \in \RR^n$ is a vector 
whose $i$-th entry is an 
estimator for $\mu_{y}(Z_i)$. 
In the cross-fitting approach we outline 
in Section~\ref{subsec:highd_method}, 
the estimators are specified as $[\hat\mu_x]_{i,j}=\hat\mu_{x,j}^{(k)}(Z_i)$ and 
$[\hat\mu_y]_i = \hat\mu_y^{(k)}(Z_i)$ for $i\in \cI_k$. 
That is, we need to run $\Omega(p_x)$ times of 
regression to obtain $\hat\mu_x^{(k)}$ and $\hat\mu_y^{(k)}$. 

Here, 
we take a different approach to estimating $\mu_x$ and $\mu_y$: 
\$
\hat\mu_x = \bPi_x X,\quad \hat\mu_y = \bPi_y Y,
\$
where $\bPi_x, \bPi_y\in \RR^{n\times n}$ 
are symmetric matrices. 
Examples include 
OLS regression for 
$\Pi_x=\Pi_y = Z(Z^\top Z)^{-1}Z^\top$ and 
ridge regression (with $\ell_2$-penalty parameter $\eta$) for
$\Pi_x= Z(Z^\top Z+\eta \mathbf{I})^{-1}Z^\top$, 
where $Z\in \RR^{n\times p_z}$ is the data matrix, 
$\mathbf{I}\in \RR^{n\times n}$ is the identity matrix, 
{and we call $\eta$ the ridge penalty for clarity.} 
We then compute the modular estimator by minimizing~\eqref{eq:general_highd}, or equivalently,
\#\label{eq:proj_highd}
\frac{1}{n}\big\{ \theta^\top X^\top X\theta/2 - Y^\top (\bPi_y + \bPi_x - \bPi_y\bPi_x) X\theta \big\} + \lambda \|\theta\|_1.
\# 
The objective~\eqref{eq:proj_highd} is equivalent to
the Lasso estimator~\eqref{eq:lasso}
applied to the design matrix $X$ and 
the response vector $(\bPi_y + \bPi_x - \bPi_x\bPi_y)Y$. 
Our modular estimator could then be computed 
with standard libraries or packages for the 
Lasso~\citep{glmnet}. 
The parameters in $\Pi_x,\Pi_y$ (such as 
the ridge penalty) 
can also be chosen with cross-validation. 
In our real data experiments, 
this computation shortcut 
using ridge regression and cross-validated 
ridge penalty $\eta$ achieves a
prediction accuracy that is comparable
to the fully modular approach (with entry-wise regression) and better than the Lasso. 

As this shortcut combines ridge regression 
and the Lasso,  
it is related to the LAVA estimator~\citep{chernozhukov2017lava} 
that is designed for recovering sums of dense 
and sparse signals. 
We develop a different estimator than theirs, 
which also serves a distinct goal of  
improving efficiency by 
exploiting the conditional independence 
structure.

%% file: sec4_app.tex


Modular regression allows for 
flexible combination of data sets 
in various missing data settings. 
In this part, we discuss a general scenario 
where we may have access to a collection of 
pairwise observations
$(X_i,Z_i)$ or $(Z_i,Y_i)$ and/or  some 
tuples $(X_i,Z_i,Y_i)$.

\revise{
While our work finds deep connections to the missing data literature (see Section~\ref{subsec:rel_work} for more discussion), 
a major distinction is that 
we mainly focus on linear models and their extensions. 
We leverage 
specific algorithmic structure (the product structure of expectation) and general independence patterns among variables, 
which differs from the missing data literature that usually
draws upon independence between variables and the missing indicator. Furthermore, the modular algorithmic structure allows us to deal with various scenarios of data availability with a unified approach.}  






\revise{Throughout, we  assume a  `completely-missing-at-random' (MCAR) 
mechanism, such that the probability of missing any variable is independent of $(X,Y,Z)$. 
Formally, suppose we have access to 
data sets $\cI^{xyz} =\{(Y_i,X_i,Z_i\}_{i=1}^{n_{xyz}}$, 
$\cI^{xz} = \{( X_i,Z_i)\}_{i=1}^{n_{xz}}$, 
$\cI^{yz}=\{(Y_i,Z_i)\}_{i=1}^{n_{yz}}$. For our exposition, it suffices to impose the following as a consequence of MCAR.}

\begin{assumption}\label{assump:MAR}
There exists a super-population $\PP_{X,Y,Z}$, such that 
the joint observations obey $\{(X_i,Y_i,Z_i\}_{i=1}^{n_{xyz}}\iid \PP_{X,Y,Z}$, 
and the pairwise observations obey 
$\{( X_i,Z_i)\}_{i=1}^{n_{xz}}\iid \PP_{XZ}$ and $\{(Y_i,Z_i)\}_{i=1}^{n_{yz}} \iid \PP_{YZ}$, i.e., 
the marginal  distributions are consistent with $\PP_{X,Y,Z}$. 
\end{assumption}
 
We will see that the decomposition 
of the regression task
allows us to flexibly modify 
our methods in Sections~\ref{sec:lowd}  
and~\ref{sec:highd} 
according 
to the availability of data. 
We provide a general procedure in Algorithm~\ref{alg:missing} 
for low-dimensional linear regression. Extension to GLMs and high-dimensional setting can be similarly obtained by replacing $\theta^\top X_iX_i^\top\theta$ by $h(X_i,\theta)$ as in Section~\ref{subsec:glm} or adding an $\ell_1$-regularizer to the estimation equation as in Section~\ref{sec:highd}.

\begin{algorithm} 
\caption{Modular linear regression with missing data}\label{alg:missing}
\begin{algorithmic}[1]
\REQUIRE Datasets $\cI^{xyz} =\{(Y_i,X_i,Z_i\}_{i=1}^{n_{xyz}}$, 
$\cI^{xz} = \{( X_i,Z_i)\}_{i=1}^{n_{xz}}$, 
$\cI^{yz}=\{(Y_i,Z_i)\}_{i=1}^{n_{yz}}$. 
\STATE Randomly split $\cI^{xyz}$, $\cI^{xz}$ and $\cI^{yz}$ into two equal-sized folds $\cI_k^{xyz}$, $\cI_k^{xz}$, $\cI_k^{yz}$, $k=1,2$.
\FOR{$k=1,2$} 
\STATE Fit models $\hat\mu_x^{(k)}(\cdot)$ for $\EE[X\given Z=\cdot]$ using $\{(X_i,Z_i)\colon i \in(\cI^{xyz}\backslash \cI^{xyz}_k) \cup(\cI^{xz}\backslash \cI_k^{xz})\}$; 
\STATE Fit models $\hat\mu_y^{(k)}(\cdot)$ for $\EE[Y\given Z=\cdot]$ using  $\{(Y_i,Z_i)\colon i\in (\cI^{xyz}\backslash \cI^{xyz}_k) \cup(\cI^{yz}\backslash \cI_k^{yz})\}$; 
\STATE Compute $\hat\mu_x(Z_i) = \hat\mu_x^{(k)}(Z_i)$ and  $\hat\mu_y(Z_i) = \hat\mu_y^{(k)}(Z_i)$ for all $i\in \cI_k$.
\ENDFOR 
\STATE Compute $\hat{C}_{\miss}^{zz}=\frac{1}{n_{xz}+n_{yz}+n_{xyz}}\sum_{i\in \cI^{xz}\cup \cI^{yz}\cup\cI^{xyz}} \hat\mu_y (Z_i) \hat\mu_x (Z_i)$, \\$\hat{C}_{\miss}^{xz}=\frac{1}{n_{xz}+n_{xyz}}\sum_{i\in \cI^{xz}\cup\cI^{xyz}} X_i \hat\mu_y (Z_i)$,  and 
$\hat{C}_{\miss}^{yz}=\frac{1}{n_{yz}+n_{xyz}}\sum_{i\in \cI^{yz}\cup\cI^{xyz}} Y_i \hat\mu_x (Z_i)$.
\STATE Compute $\hat{C}_{\miss} = \hat{C}_{\miss}^{xz} + \hat{C}_{\miss}^{yz} - \hat{C}_{\miss}^{zz}$, $\hat\Sigma_{\miss} = \frac{1}{n_{xz}+n_{xyz}}\sum_{i\in \cI^{xz}+\cI^{xyz}}X_iX_i^\top$.
\ENSURE The modular estimator $\hat\theta_n^{\miss} = \argmin_{\theta\in \RR^{p_x}}~
\big\{ \theta^\top \hat\Sigma_{\miss} \theta /2 - \hat{C}_{\miss}^\top \theta \big\}$.
\end{algorithmic}
\end{algorithm}

\revise{Algorithm~\ref{alg:missing} 
is operable even when no joint observations of $(X,Z,Y)$ 
are available, i.e., $\cI^{xyz}=\varnothing$ and $n_{xyz}=0$.
This is inspired by the crucial fact that 
\eqref{eq:C_default} 
only involves pairwise observations of 
$(X,Z)$ and $(Z,Y)$, 
and the same for 
learning $\mu_x(z)$ 
and $\mu_y(z)$. When $n_{xz}=n_{yz}=0$, 
it reduces to Algorithm~\ref{alg:lowd} that uses joint observations of $(X,Y,Z)$.} 

\revise{The next proposition generalizes Theorem~\ref{thm:lowd}, whose proof is in Appendix~\ref{app:subsec_proof_miss}.  Similar results can be derived for GLMs and the high-dimensional case, which we omit for brevity.}

\begin{proposition}
\label{prop:miss}
    Suppose Assumptions~\ref{assump:ci} and~\ref{assump:MAR} hold. Suppose $n_{xz}/(n_{xz}+n_{yz}+n_{xyz})\to \rho_{xz}$, 
and $n_{yz}/(n_{xz}+n_{yz}+n_{xyz})\to \rho_{yz}$ 
for some $\rho_{xz},\rho_{yz}\in (0,1)$. 
Denote $N_{xz}=n_{xz}+n_{xyz}$ and $N_{yz}=n_{yz}+n_{xyz}$.
Assume $\|\hat\mu_{x}^{(k)}-\mu_x\|_{L_2(\PP_Z)}\cdot\|\hat\mu_{y}^{(k)}-\mu_y\|_{L_2(\PP_Z)}=o_P(1/\sqrt{N_{xz} }) +o_P(1/\sqrt{N_{yz} })$ 
for $k=1,2$,  
$\EE[XX^\top]\succ 0$ is finite, 
and $X(Y-X^\top\theta^*)$ has finite 
second moments.  
Then  $
\hat\theta_n^\mod -\theta^* = \EE[XX^\top]^{-1}\big\{\frac{1}{n_{xz}}\sum_{i\in\cI^{xz}}\phi_{xz}(X_i,Z_i) 
+ \frac{1}{n_{yz}}\sum_{i\in\cI^{yz}}\phi_{yz}(Y_i,Z_i)
+ \frac{1}{n_{xyz}}\sum_{i\in\cI^{xyz}}\phi_{xyz}(X_i,Y_i,Z_i)\big\}+o_P(1/\sqrt{N_{xz} }+1/\sqrt{N_{yz}})
$, where 
$\phi_{xz}(X,Z)=\frac{\rho_{xz}}{1-\rho_{yz}} X\mu_y(Z)  - \rho_{xz}\mu_{y}(Z )\mu_x(Z)-\frac{\rho_{xz}}{1-\rho_{yz}}XX^\top\theta^* $,
$\phi_{yz}(Y,Z)= \frac{\rho_{yz}}{1-\rho_{xz}} Y \mu_x(Z ) - \rho_{yz}\mu_{y}(Z )\mu_x(Z ) $, and 
$\phi_{xyz}(Y,Z)= \frac{ X \mu_y(Z )}{1-\rho_{yz}} + \frac{1}{1-\rho_{xz}} Y \mu_x(Z ) -  \mu_{y}(Z )\mu_x(Z )-\frac{1}{1-\rho_{yz}}XX^\top\theta^* $. When $n_{xyz}=0$, in the formulas above one should interpret $0/0$ as $0$.
\end{proposition}

\subsection{Pairwise observations}
\label{subsec:pair}

\revise{We discuss a few consequences 
when only pairwise observations of $(X,Z)$ and $(Y,Z)$ 
are available.  
In general, 
predicting $Y$ with $X$  is impossible without identification assumptions (Assumption~\ref{assump:ci}).}
Also, the structure learning approach 
in Section~\ref{sec:robustness} is no longer feasible. 

Since $n_{xyz}=0$ and $\cI^{xyz}=\varnothing$, 
we know $\rho_{xz}+\rho_{yz}=1$. 
The asymptotic linear expansion in Proposition~\ref{prop:miss} reduces to 
$
\hat\theta_n^\mod -\theta^* = \frac{1}{n_{xz}}\sum_{i\in\cI^{xz}}\phi_{xz}(X_i,Z_i) 
+ \frac{1}{n_{yz}}\sum_{i\in\cI^{yz}}\phi_{yz}(Y_i,Z_i)
+o_P(1/\sqrt{n_{xz} }) + 1/\sqrt{n_{yz} })
$, where 
$\phi_{xz}(X,Z)= X\mu_y(Z)  - \rho_{xz}\mu_{y}(Z )\mu_x(Z)- XX^\top\theta^*$, and 
$\phi_{yz}(Y,Z)= Y \mu_x(Z ) - \rho_{yz}\mu_{y}(Z )\mu_x(Z )$.
The modular estimator gets more efficient as $n_{xz}$ and $n_{yz}$ increases, while plain OLS is no longer feasible. 

\subsection{Partially pairwise observations}
\label{subsec:missing_all}

Modular prediction can also be adapted to 
settings where a limited number of 
triple observations are available in addition to  
$(X,Z)$ and $(Y,Z)$ pairs, i.e., $n_{xyz}>0$. 

In this case,  one can run OLS on $\cI^{xyz}$, 
whose asymptotic expansion is $\hat\theta_n^\ols - \theta^* = \frac{1}{n_{xyz}}\sum_{i\in\cI^{xyz}}\EE[XX^\top]^{-1}(Y_i-X_i^\top\theta^*) +o_P(1/\sqrt{n_{xyz}})$. 
The estimation error is of the scale $O_P(1/\sqrt{n_{xyz}})$. 
In contrast, by utilizing additional pairwise observations, 
Algorithm~\ref{alg:missing} achieves the  rate of
$O_P(1/\sqrt{n_{xz}+n_{xyz}}) + O_P(1/\sqrt{n_{yz}+n_{xyz}})$; this is a substantial improvement upon OLS 
if $n_{xz}$ and $n_{yz}$ are much larger than $n_{xyz}$, 
which may happen in the example at the beginning of this paper. Even if $n_{xz}$ and $n_{yz}$ are of the same order as $n_{xyz}$, i.e., $\rho_{xz}+\rho_{yz}<1$, Algorithm~\ref{alg:missing} may still achieve a smaller asymptotic variance than  OLS.

With a few $(X,Z,Y)$ observations, one may leverage them to learn the conditional independence structure. 
Following the notations in  Section~\ref{sec:robustness}, 
$Z$ will then be replaced by $Z^\full$ which 
consist of both the original features 
in $Z$ and some other features in $X$. 
Thus, in Line 7 of Algorithm~\ref{alg:missing}, 
$\hat{C}_\miss^{yz}$ has to be computed with 
data in $\cI^{xyz}$, and $\hat{C}_\miss^{zz}$ has 
to be computed with data in $\cI^{xz}\cup \cI^{xyz}$. 
That said, we do note that 
in our real data application (see Section~\ref{subsec:real}), modular regression 
performs well without structure learning. 
Theoretically, the above procedure without structure learning 
may induce substantial bias in cases where 
conditional independence does not hold. 
However, in settings where we have recorded a rich 
set of intermediate covariates in $Z$ (as often assumed  
in surrogate methods), the conditional independence 
assumption is often plausible or the bias is sufficiently
small compared with the reduced variance.

%% file: simu.tex

We evaluate 
our methods on simulated datasets  
to compare the 
bias, variance and overall estimation error 
of our methods to non-modular counterparts in both 
low and high dimensional settings. 
We also investigate the robustness to approximate conditional 
independence in the high dimensional setting.

\subsection{Low-dimensional setting}

We focus on parameter estimation in the low-dimensional setting.  
We are interested in estimating  
$\theta^* = \argmin_{\theta\in \RR^{p_x}} \EE[(Y-X^\top\theta)^2]$, 
where $Y\in \RR$ is the response, and $X\in \RR^{p_x}$ are the covariates. 
We suppose in the training data we have access to $Z\in \RR^{p_z}$ such that 
$X\indep Y\given Z$. 
We fix $p_x=4$, $p_z=6$, and 
a relatively small sample size at $n=500$. 

We design $2\times 2$ data generating processes depending on 
(i) whether the true relation is linear and (ii) whether the 
data generating process follows the graphical model $X\to Z\to Y$ 
or $X \leftarrow Z \rightarrow Y$. 
The details are summarized in Table~\ref{tab:simu}. 
The linear regression model is not necessarily well-specified 
for all settings, while the OLS parameter $\theta^*$ is always well-defined. 

\begin{table}[ht]
\centering
\renewcommand\arraystretch{1}
\begin{tabular}{c|c |c }
\toprule
    Setting  & Data generating process & Comment \\
\hline
1  & $Z\sim \textrm{Unif}[-1,1]$, $X=BZ + \epsilon_z$, $Y=Z^\top \gamma + \epsilon_y$  &  linear, $X \leftarrow Z \rightarrow Y$   \\  
    \hline
2  & $X\sim \textrm{Unif}[-1,1]$, $Z=BX + \epsilon_z$, $Y=Z^\top \gamma + \epsilon_y$ &   linear, $X\to Z\to Y$  \\  
\hline
3  &   $Z\sim \textrm{Unif}[-1,1]$, $X=f(Z) + \epsilon_z$, $Y= g(Z) + \epsilon_y$   &   nonlinear, $X \leftarrow Z \rightarrow Y$   \\  
\hline
4  &  $X\sim \textrm{Unif}[-1,1]$, $Z=f(X)+ \epsilon_z$, $Y=g(Z)  + \epsilon_y$ &   nonlinear, $X\to Z\to Y$   \\   
\bottomrule 
\end{tabular}
\vspace{0.5em}
\caption{Data generating processes in all settings, where $\epsilon_z\sim N(0,\sigma_z^2)$ and $\epsilon_y\sim N(0,\sigma_y^2)$ 
are independent noise, and $f$, $g$ are nonlinear functions defined in the text.}
\label{tab:simu}
\end{table}

In settings 1 and 2, we set $\gamma = (0.531, -0.126, 0.312,0,0,0)^\top\in \RR^{p_z}$, and 
$B\in \RR^{p_x\times p_z}$ or $B\in \RR^{p_z\times p_x}$ is a constant matrix 
where 8 out of 24 entries are randomly set to $0.5$ then fixed for all configurations, 
while the other entries are zero.   
\revise{The $f$ and $g$ functions in Settings 3 and 4 are defined as follows.} 
In setting 3, we set $X_4=[BZ]_{4}$ with the same $B$ as setting 1, 
and $X_1=0.5X_1+\ind\{Z_1>0\}$, $X_2=-0.5Z_3+\ind\{Z_4>0\}$ 
and $X_3=\ind\{Z_4>0\}$. 
In setting 4, we set $Z_{4:6} =[BX]_{4:6}$ with the same $B$ as setting 2, and 
$Z_1=0.5X_1+\ind\{X_1>0\}$, $Z_2=-0.5X_3+\ind\{X_4>0\}$ 
and $Z_3=\ind\{X_4>0\}$. 
In all settings, $\epsilon_z\sim N(0,\sigma_z^2)$ and $\epsilon_y\sim N(0,\sigma_y^2)$ 
are independent noise, where we vary 
the noise strengths $\sigma_z,\sigma_y\in \{0.1,0.5,1,2\}$,
and $X\sim \text{Unif}[-1,1]$ means 
all entries in $X$ are i.i.d~from Unif$[-1,1]$.

We compute the OLS parameter $\hat\theta_n^\ols$ 
and our modular estimator with $\hat\mu_y^{(k)}$ and $\hat\mu_x^{(k)}$ 
estimated with (i) cross-validated Lasso~\citep{glmnet}, 
(ii) cross-validated ridge regression~\citep{glmnet}, 
(iii) regression random forest from \texttt{grf} R-package, 
and (iv) linear regression. 
All procedures are repeated for $N=1000$ independent runs. 
For comprehensive illustration here, 
we aggregate all coefficients 
and evaluate the rooted mean square error (RMSE) 
$\EE[\|\hat\theta   -\theta ^*\|_2^2]^{1/2}$  
(summation of) standard deviation (SD) $\sum_{j=1}^{p_x}\textrm{sd}(\hat\theta_j -\theta_j^*)$ 
and  bias $\sum_{j=1}^{p_x}|\EE[\hat\theta_j - \theta_j^*]|$ 
for the five estimators. 
\revise{We plot the aggregated RMSE   
in settings 2 and 3 
in Figure~\ref{fig:lowd_rmse_sub}.  
Bias, SD and RMSE in all settings
(either aggregated or for each entry) are in Appendix~\ref{app:simu_lowd} which 
convey similar messages.}

\begin{figure}[ht]
\centering
\includegraphics[width=5.5in]{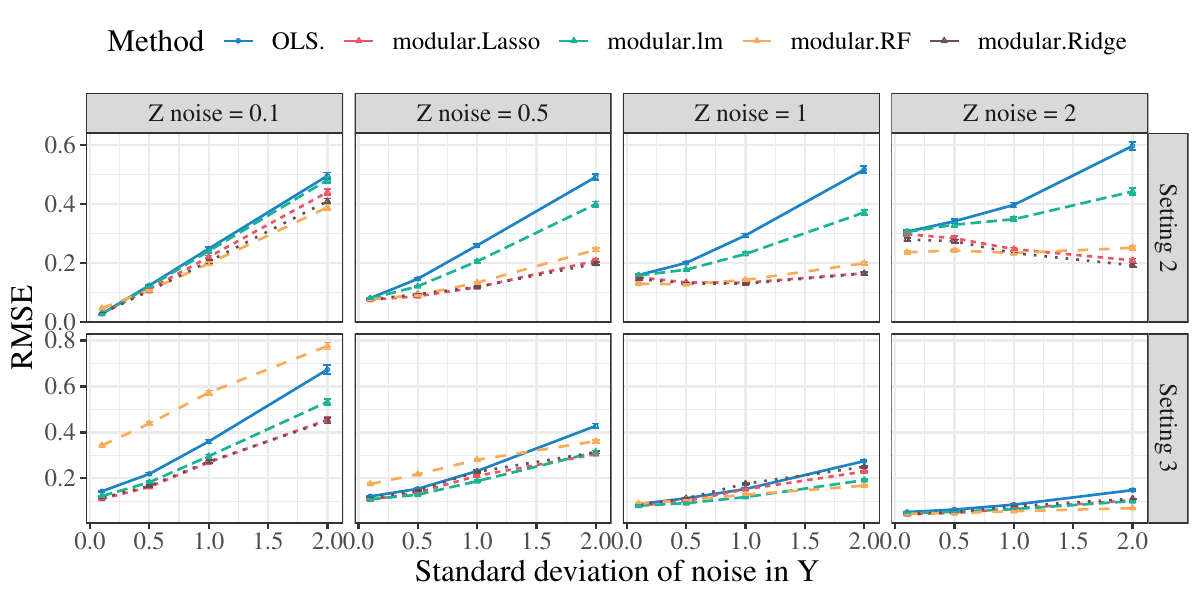}
\caption{RMSE averaged over $N=1000$ runs with $n=200$. 
Each column corresponds to a value of $\sigma_z$. In each subplot, 
the $x$-axis (noise strength) equals $\sigma_y$, the standard deviation of noise in $Y$. 
Modular regression achieves smaller estimation 
error in almost all settings.}
\label{fig:lowd_rmse_sub}
\end{figure}

Modular regression -- no matter which machine learning regressor is used 
for $\hat\mu_x$ and $\hat\mu_y$ -- achieves smaller RMSE than the plain OLS in almost all configurations.  Modular regression  has a larger bias than OLS due to the additional regression  (Figure~\ref{fig:lowd_bias_sub}, Appendix~\ref{app:simu_lowd}), 
yet much more substantial reduction in standard deviation (Figure~\ref{fig:lowd_sd_sub}, Appendix~\ref{app:simu_lowd}) 
in all settings. Also, results from 
setting 1  
confirms Remark~\ref{rem:linear}, where the asymptotic variance reduction 
grows with both $\sigma_z$ and $\sigma_y$. 

A noteworthy exception is $\sigma_z=0.1$ in setting 3, where the RMSE of 
modular regression with random forests is worse than that of the OLS estimator.
This is because with small sample size ($n=200$) and 
low noise (hence the uncertainty in $\hat\theta_n^\ols$ is very small), 
the bias introduced by the random forest regression 
is large compared to the reduction in variance. 
However, when $\sigma_z$ becomes larger, the bias  (Figure~\ref{fig:lowd_bias_sub}) has a smaller magnitude; 
this may be because we enter a signal-to-noise ratio regime that favors tree-based approaches. 
The impact of bias is also less substantial when we increase the sample size. 
Figure~\ref{fig:lowd_rmse_sub_k} in Appendix~\ref{app:simu_lowd} 
plots the RMSE for estimating $\theta_1^*$ 
when $n=2000$, where we see a much better performance of modular regression 
with random forests. Thus, we recommend using machine learning regressors 
for larger sample sizes.

We also note that modular regression with linear regression (green) 
performs  better than the plain OLS in all settings 
(even when the true relation is nonlinear),  although it is sometimes outperformed by other modular methods.  
This phenomenon is universal in our simulation (see Appendix~\ref{app:simu_lowd} 
for RMSE for all entries in all settings), 
\revise{and verifies the relaxed consistency condition we 
observe after Theorem~\ref{thm:lowd}.
Although this is not a general rule, 
we  still recommend linear regression 
in practice  especially for small sample size.}
One can also add a few transformed
regressors into the linear regression to further
adapt to nonlinearity.

\revise{Finally, to show the asymptotic behavior of modular regression, 
in Figure~\ref{fig:lowd_rmse_n} we plot the aggregated RMSE 
for varying 
sample sizes $n\in \{200,500,1000,2000\}$ with $\sigma_z=0.5$ 
and $\sigma_y=1$ in all settings. 
Modular regression outperforms OLS in nearly all settings, 
and linear models (lm, Lasso and Ridge) as base estimators 
show robust performance. The performance 
of random forest is less stable for small sample size ($n=200$) in 
setting 3, hence we  recommend using flexible machine learning 
models with larger sample sizes.}

\begin{figure}[ht]
\centering
\includegraphics[width=5.5in]{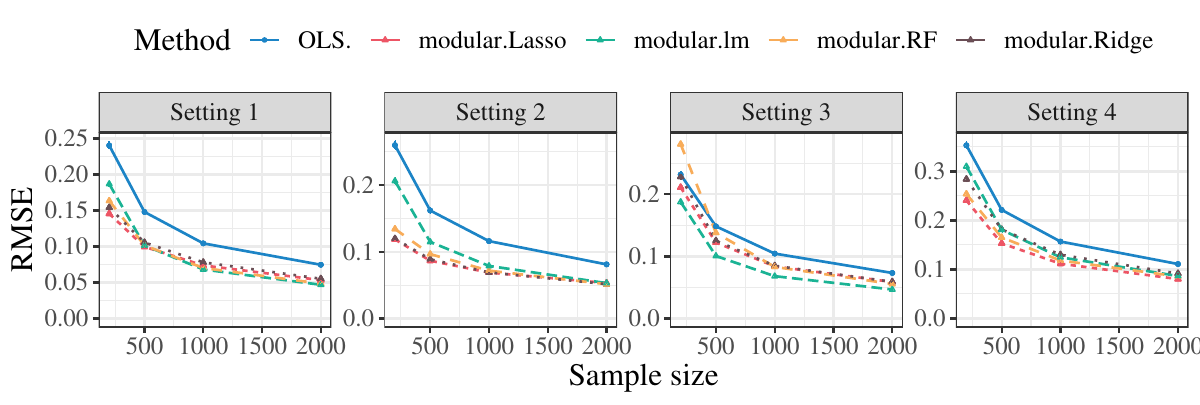}
\caption{RMSE averaged over $N=1000$ runs for various sample size, 
and $\sigma_z=0.5$, $\sigma_y$=1. Modular regression outperforms OLS in all settings, and its performance improves with $n$.}
\label{fig:lowd_rmse_n}
\end{figure}

\subsection{High-dimensional setting}


We now consider two data generating processes in the high-dimensional setting, 
where the conditional independence structure 
holds only in one of them.  
We show that our methods outperform 
the Lasso in 
Setting 1 (with conditional independence), 
and is robust against the 
violation of conditional independence 
in Setting 2.
Following the preceding notations, 
$X\in \RR^{p_x}$ is the covariates available in the prediction phase, 
while $Z\in \RR^{p_z}$ is only available for training data, 
and $Y\in \RR$ is the outcome. 
The data generating processes are visualized in 
Figure~\ref{fig:simu_dist_fig}.

\begin{figure}[ht]
\centering 
\begin{subfigure}[t]{0.45\linewidth}
    \centering
\begin{tikzpicture}[->,>=stealth', thick, main node/.style={circle,draw}]

\node[main node, text=black, circle, draw=black, fill=black!5, scale=1.1] (2) at  (0,0) {\small $Z$};
\node[main node, text=black, circle, draw=black, fill=black!5, scale=1.1] (1) at  (-1.5,-1) {\small $X $}; 
\node[main node, text=black, circle, draw=black, fill=black!5, scale=1.1] (3) at  (1.5,-1) {\small $Y$}; 

\draw[->] (1) edge [draw=black] (2);
\draw[->] (2) edge [draw=black] (3); 
 
\end{tikzpicture} 
\caption{Setting 1: conditional independence}
\end{subfigure} \hspace{0.2in}
\begin{subfigure}[t]{0.45\linewidth}
    \centering
    \begin{tikzpicture}[->,>=stealth', thick, main node/.style={circle,draw}]

        \node[main node, text=black, circle, draw=black, fill=black!5, scale=1.1] (2) at  (0,0) {\small $Z$};
        \node[main node, text=black, circle, draw=black, fill=black!5, scale=1.4] (1) at  (-1.5,-0.6) {\small $~$ }; 
        \node[main node, text=black, circle, draw=black, fill=black!5, scale=1.1] (4) at  (-1.5,-1.4) {\small $~$ }; 
        \node[main node, text=black, circle, draw=black, fill=black!5, scale=1.1] (3) at  (1.5,-1) {\small $Y$}; 
        \node[ ] (6) at (-1.85,-1.8) {$X$};
        \node[rectangle,draw, densely dashed, minimum width=1cm, minimum height=2cm] (5) at (-1.6,-1.1) {};
        \draw[->] (1) edge [draw=black] (2);
        \draw[->] (2) edge [draw=black] (3); 
        \draw[->] (4) edge [draw=black] (3); 
        \end{tikzpicture} 
        \caption{Setting 2: approximate conditional independence}
    \end{subfigure}
\caption{Illustration of the two data generating processes. }
\label{fig:simu_dist_fig}
\end{figure}
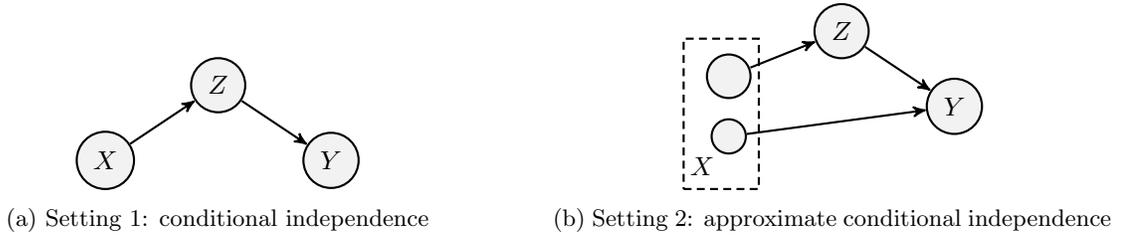

In settting 1, we generate $X\in \RR^{p_x}$ where each entry is  i.i.d.~from 
$N(0,1)$; 
then 
we generate  $Z = BX  + \epsilon_z\in \RR^{p_x}$ with 
i.i.d.~noise $\epsilon_z\sim N(0,1)$ given a parameter matrix $B\in \RR^{p_z\times p_x}$. 
Finally, we generate
$Y = Z^\top \gamma + \epsilon_y$ for i.i.d.~random noise $\epsilon_y \sim N(0,4)$ 
and some $\gamma\in \RR^{p_z}$. 
To ensure sparsity, we let $B_{ij}=0$ for all $j>s$, 
such that $X_j$'s for $j>s$ are irrelevant for the prediction. 
Then for each $j\leq s$, we randomly select $2s$ entries in 
the $j$-th column of $B$ with $B_{ij} = 0.25$, 
the remainings with $B_{ij}=0$. 
We set $\gamma_i=0.5$ for $1\leq i\leq s$, 
and $\gamma_i=0$ for $i>s$. 
In this way, $Y = X^\top \theta^* + (\gamma^\top \epsilon_z + \epsilon_y)$ 
where $\theta^* = B^\top \gamma$, and $X\indep Y\given Z$. 

In setting 2, we ensure a small subset of covariates in $X$ 
to have direct impact on $Y$ (the link from $X$ to $Y$ in Figure~\ref{fig:simu_dist_fig}(b)). 
To be specific, we generate $X$ and 
$Z=BX + \epsilon_z$, where $\{B_{ij}\colon j\leq s\}$ and $\epsilon_z$ are 
the same as setting 1, and 
generate $Y = Z^\top \gamma + X^\top \tilde\gamma + \epsilon_y$ 
for i.i.d.~random noise $\epsilon_y \sim N(0,4)$; 
the direct coefficient $\tilde\gamma$ satisfies 
$\tilde\gamma_{i}=0$ for $i\leq s$ and $\sum_{i=1}^{p_x}\ind\{\tilde\gamma_i\geq 0\} = 5$. 
In this setting, 
the high dimensional model 
$Y=X^\top \theta^* + (\gamma^\top \epsilon_z+\epsilon_y)$ holds 
with $\theta^* = B^\top \gamma + \tilde\gamma$, but 
$X\indep Y\given Z$ does not hold exactly. 

We compare our method in Section~\ref{subsec:highd_method}
to the Lasso; to ensure fair comparison, 
we run the Lasso using our modular method 
while setting $\hat{f}_y(Z_i) := Y_i$ and $\hat{f}_x(Z_i):=X_i$ for all $i$, 
so that~\eqref{eq:def_modular_highd} reduces to the Lasso. 
We also evaluate an oracle modular regression algorith, that is, 
we set $\hat\mu_y:=\EE[Y\given Z]$ and 
 $\hat \mu_x := \EE[X_j\given Z]$ as the ground truth. 
The regularization parameter $\lambda$ 
for $\ell_1$-penalty is chosen by 5-fold 
cross validation on the training data for all methods. 
In our modular regression algorithm, we fit $\hat\mu_x$ and $\hat\mu_y$ 
by ridge regression using  the \texttt{cv.glmnet} function 
from the \texttt{glmnet} R-package~\citep{glmnet}. 
The procedures are evaluated for $p_x=p_z=100$ and $s=10$. 
The training sample is $n=500$ and 
we evaluate the prediction on $n_\test=1000$ test samples 
for a relatively accurate estimate.

\subsubsection{Performance under conditional independence}

In setting 1 with a well-specified high-dimensional linear model 
and exact conditional independence $X\indep Y\given Z$, 
we evaluate 
(i) the parameter estimation error $\|\hat\theta_j - \theta_j^*\|_2 $ 
where $\hat\theta$ is the output of modular regression or the Lasso, 
as well as 
(ii) prediction performance: excess risk 
$\frac{1}{n_\test}\sum_{i=1}^{n_\test} (X_i^\top \hat\theta  - X_i^\top \theta^*)^2$ and 
mean squared error (MSE) $\frac{1}{n_\test}\sum_{i=1}^{n_\test}(X_i^\top \hat\theta  - Y_i)^2$ 
on the test samples. 

\paragraph{Parameter estimation.}
For each $j\in [p_x]$, we evaluate 
the rooted mean squared error (RMSE) $(\hat\theta_j - \theta^*_j)^2$
and bias $|\EE[\hat\theta_j]-\theta_j^*|$ over $N=100$ replicates, and visualize  the RMSE (left) and bias (right)
via boxplots in Figure~\ref{fig:simu_1}.

\begin{figure}[ht]
\centering
\begin{subfigure}[t]{0.49\linewidth}
    \centering
    \includegraphics[width=\linewidth]{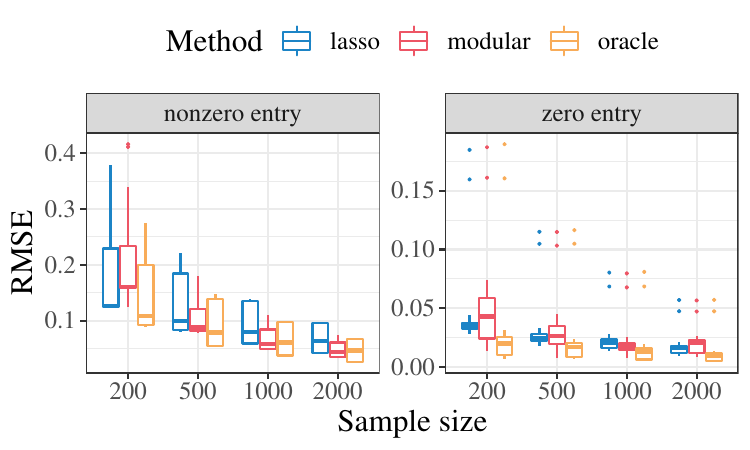}
\end{subfigure} 
\begin{subfigure}[t]{0.49\linewidth}
    \centering
    \includegraphics[width=\linewidth]{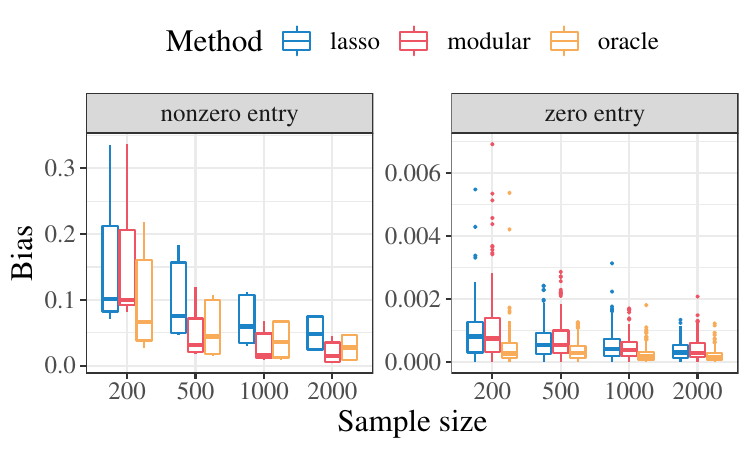}
\end{subfigure} 
\caption{Boxplot of RMSE (left) and bias (right) of $\{\hat\theta_j \colon j\in[p_x]\}$, 
averaged
over all replicates in setting 1, separately  
for $\theta_j^*\neq 0$ (nonzero entry) and $\theta_j^*=0$ (zero entry). 
The $x$-axis indicates the method to obtain $\hat\theta_j$. Modular regression achieves smaller RMSE 
and smaller bias than the Lasso
due to a different bias-variance trade-off.}
\label{fig:simu_1}
\end{figure}

For 
nonzero entries (the blue boxplots), 
we observe a significant reduction in RMSE 
compared to the Lasso \revise{under various sample sizes}; 
despite the estimation error in fitting the conditional mean, 
our method is  comparable to its oracle counterpart, 
even with a smaller overall RMSE (though less stable across entries). 
This might be due to the instability in cross-fitting with ridge regression \revise{or achieving a better bias-variance tradeoff by cross-validation. 
Also, as $n$ increases, modular regression gets more and more stable.}
For zero entries,  
the RMSE are similar across three methods: 
while the oracle yields lower RMSE 
than the Lasso, the slight inflation of RMSE in modular prediction 
might be due to the estimation error of $\hat\mu_x$ and $\hat\mu_y$.

We find that the reduction in RMSE for nonzero entries 
mainly comes from the reduction in bias, as illustrated by 
the right panel of Figure~\ref{fig:simu_1}. 
This is consistent with our theory in 
Theorem~\ref{thm:highd}: The reduced variance 
of the proxy $\hat{C}_{\textrm{lm}}$ 
allows the cross-validation step to 
choose a model with smaller bias.
Due to limited space, we defer the corresponding plot of 
standard deviations to Figure~\ref{fig:simu_sd} in Appendix~\ref{app:simu_lowd}.

\paragraph{Prediction.}
We plot the $N=1000$ excess risks 
averaged on the test samples for various training sample size $n$
in Figure~\ref{fig:simu_1_pred}. 
It shows 
significant improvement in prediction accuracy 
of our method compared to plain Lasso; 
modular regression is slightly inferior to the oracle counterpart 
but the difference is relatively moderate.  
\revise{The error of modular regression also gets smaller as the sample size increases.}

\begin{figure}[ht]
    \centering
    \begin{subfigure}[t]{0.49\linewidth}
        \centering
        \includegraphics[width=0.9\linewidth]{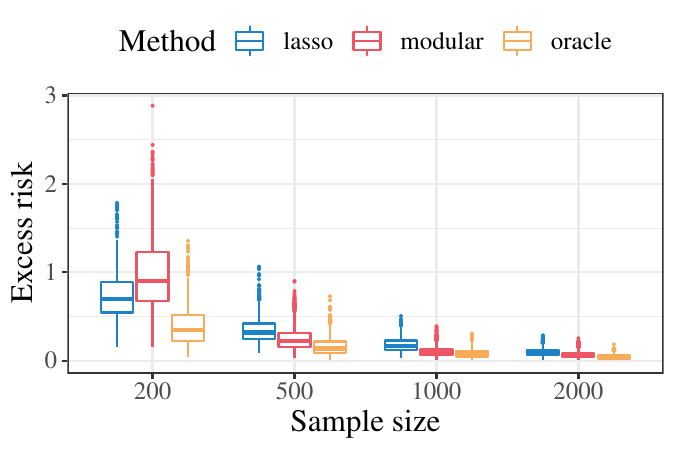}
    \end{subfigure} 
    \begin{subfigure}[t]{0.49\linewidth}
        \centering
        \includegraphics[width=0.9\linewidth]{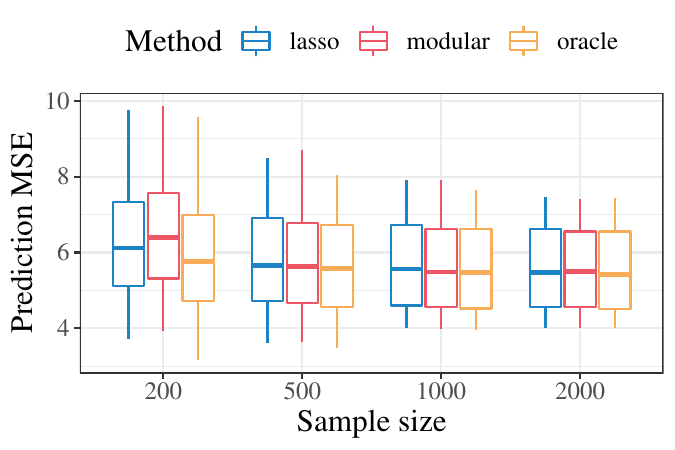}
    \end{subfigure} 
\caption{Boxplot of $N=1000$ empirical excess risks $\EE[(X_i^\top \hat\theta - X_i^\top \theta^*)^2]$ 
and MSEs $\EE[(X_i^\top\hat\theta-Y_i)^2]$ 
in all replicates for setting 1.  
Modular regression achieves smaller excess risk than the Lasso. 
The improvement in prediction MSE is less visible due to
the irreducible error. 
}
\label{fig:simu_1_pred}
\end{figure}

\subsubsection{Robustness to approximate conditional independence}

In the following, we test 
the robustness of our method 
against potential violation of the conditional independence 
assumption.
In setting 2 where the conditional independence only 
approximately holds, we additionally conduct a 
structure learning step using Lasso. 

We first run a cross-validated Lasso 
of $Y$ on $(X,Z)$
using the \texttt{cv.glmnet} function~\citep{glmnet} 
for model selection; 
all $X_j$ that are selected by this Lasso step 
is then merged into $Z$. 
For any selected $X_j$, we will 
skip the regression of $\EE[X_j\given Z]$ and directly set 
$\hat{f}_{x,j}(Z_i) = X_{i,j}$ for all training samples. 
We also evaluate an oracle counterpart which uses 
the ground truth of the structure and 
the true conditional expectations 
for those $X_j\indep Y\given Z$; we skip the regression
for those $X_j$ with a direct impact on $Y$, 
 as outlined in Section~\ref{sec:robustness}. 

Figure~\ref{fig:simu_2} plots the RMSE (left) and bias (right) 
of all coefficients $\hat\theta_j$ 
\revise{with various sample sizes}, averaged over $N=1000$ replicates. 
The plot for standard deviations of $\hat\theta_j$ is in Figure~\ref{fig:simu_sd} 
in Appendix~\ref{app:simu_lowd}. 
We again see an improved RMSE especially for nonzero entries (blue). 
While the RMSE is less stable across 
different entries, 
perhaps because of the 
additional uncertainty introduced in the structure learning step, it is 
in general better than plain Lasso and comparable to the oracle.
Similar to the previous setting, this reduction 
of RMSE mainly comes from a reduced bias as shown in the 
right panel of Figure~\ref{fig:simu_2}. 
In general,  our method is able to 
adapt to approximate conditional independence and maintain 
certain efficiency gain. 

\begin{figure}[ht]
    \centering
    \begin{subfigure}[t]{0.49\linewidth}
        \centering
        \includegraphics[width=\linewidth]{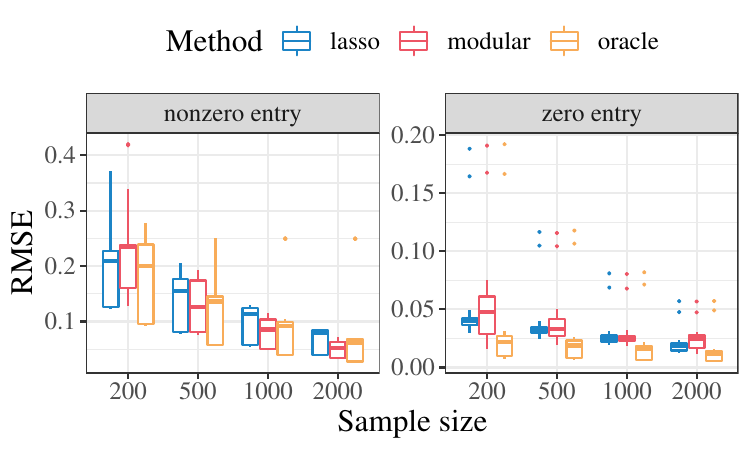}
    \end{subfigure} 
    \begin{subfigure}[t]{0.49\linewidth}
        \centering
        \includegraphics[width=\linewidth]{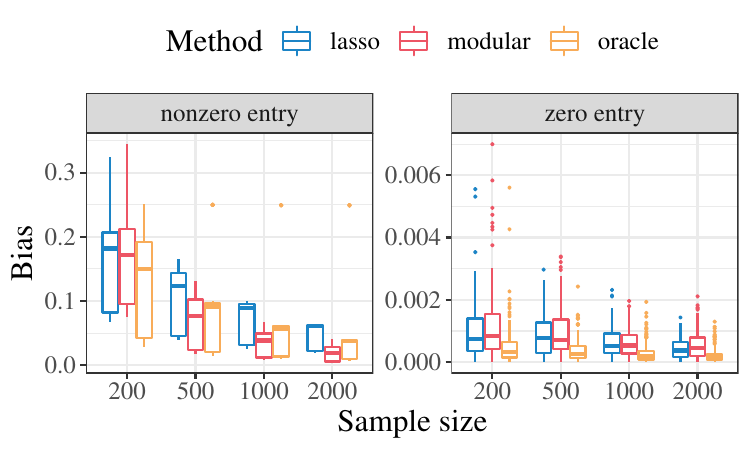}
    \end{subfigure} 
\caption{Boxplot of RMSE (left) and bias (right) of $\hat\theta_j$, $j\in[p_x]$, averaged
over all replicates, for
$\theta_j^*\neq 0$ (nonzero entry) and $\theta_j^*=0$ (zero entry) separately.  Modular regression achieves smaller estimation RMSE and bias than the Lasso, due to a different bias-variance trade-off.}
\label{fig:simu_2}
\end{figure}

We summarize 
the prediction performance in Figure~\ref{fig:simu_2_pred}. 
The excess risk of modular prediction 
 lies between that of the oracle counterpart and the 
 Lasso. Still, the relative improvement in terms of 
prediction MSE is present yet smaller due to the irreducible noise. 

\begin{figure}[htbp]
    \centering
    \begin{subfigure}[t]{0.49\linewidth}
        \centering
        \includegraphics[width=0.9\linewidth]{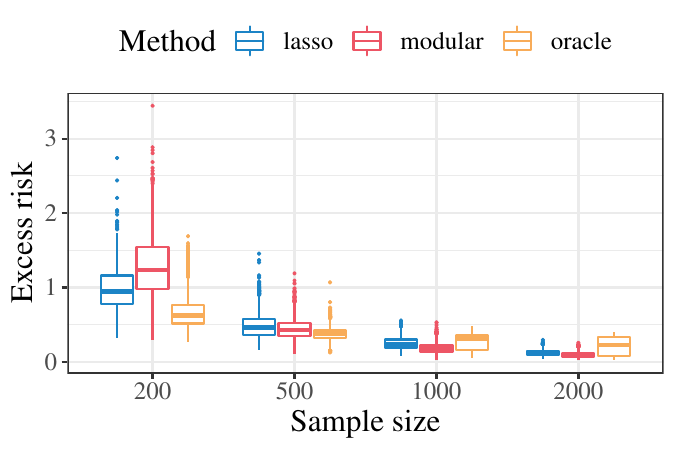}
    \end{subfigure} 
    \begin{subfigure}[t]{0.49\linewidth}
        \centering
        \includegraphics[width=0.9\linewidth]{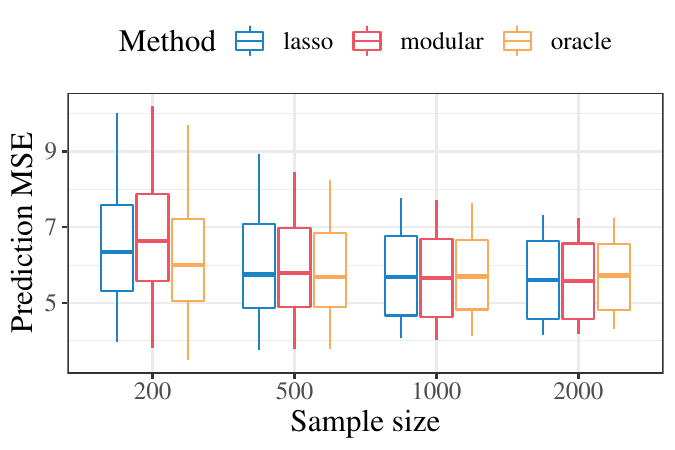}
    \end{subfigure} 
\caption{Boxplot of $N=100$ empirical excess risks $(X_i^\top \hat\theta - X_i^\top \theta^*)^2$ (left) 
and MSEs $(X_i^\top \hat\theta - Y_i)^2$ (right) in all replicates for setting 2.  
Modular regression with structure learning achieves smaller prediction excess risk than the Lasso.}
\label{fig:simu_2_pred}
\end{figure}

\vspace{-1em}

%% file: real.tex

\subsection{Dataset overview}
We apply our method to the 
English Longitudinal Study of Ageing dataset~\citep{ukageing}. 
The waves of data consistently 
measure several modules of 
features such as health trajectories, disability 
and healthy life expectancy, the  economic and financial situations, 
cognition and mental health, etc. 
We use the Wave 7 and Wave 9 data, 
collected in 2014 and 2018, respectively, 
restricted to people 
who are present in both waves. 

We consider  
predicting the future health  outcomes of people based 
on their current available features.  
The Wave 7 data is used as covariates. 
We divide all variables into two  categories: 
health (both mental and physical) and social conditions 
(including household demographics, financial,   work and social situations). 
After pre-processing through one-hot encoding for 
categorial and string-valued variables 
and filtering out some highly imbalanced variables, 
the health and social categories contain $184$ and $888$ features, respectively. 
We take the \texttt{hehelf}  variable from Wave 9 data 
as the response $Y$: 
the reported overall health situation ranging from 1 (excellent) 
to 5 (poor). We treat the outcome as  continuous.

\subsection{Real data application}
\label{subsec:real}

We take the covariates in the social category 
as $X\in \RR^{p_x}$ for $p_x=888$, and 
those in the health category as $Z\in \RR^{p_z}$ for $p_z=184$.  
 Our procedure is geared towards settings with high noise. To simulate such a setting, we smooth the discrete response $Y$ by adding i.i.d.~noise drawn from $N(0,4)$.
The task is to predict $Y$ using $X$, while $Z$ may be available during the 
training process. 
This is practical setting where 
health conditions ($Z$) may be more costly or difficult to evaluate, 
hence only available in pre-collected data. 

We consider a  scenario where 
one only has access to a limited number $n$ of $(X,Y,Z)$ triples 
as well as $n_{xz}$ observations for $(X,Z)$ and $n_{yz}$ observation for 
$(Z,Y)$ pairs 
in the training phase. This mimics a realistic scenario 
where it is difficult to obtain full observations simultaneously
but modular data are more easily accessible. 
When $n_{xz}=n_{yz}=0$, it reduces to the standard full-observation setting. 
While it is more difficult to test for conditional independence 
and the modeling assumptions with limited joint observations, 
we could still use our framework 
to merge the individual datasets and improve out-of-sample
prediction. We  focus on the prediction MSE on the test sample
{because no ground truth is available}. 
{We consider two scenarios:}

\begin{enumerate}[label=(\roman*)]
    \item Fixed $n$ and varying $n_{xz}$ and $n_{yz}$. 
    We fix $n=200$, 
    and the number of pair observations varies as $n_{xz} = n_{yz} = n\cdot \rho$ 
    for $\rho \in \{0.5, 1, 5, 10\}$. 
    \item Fixed $n+n_{xz}+n_{yz}$ and varying proportion.  
    We fix the total sample size
    $n+n_{xz}+n_{yz} = 1000$, while 
    varying the proportion of joint observations by 
    $n=1000\cdot \rho$, $n_{xz}=n_{yz}$ for 
    $\rho\in\{0.05, 0.1, 0.2, 0.5, 0.8, 1\}$. 
\end{enumerate}

We evaluate our modular regression approach outlined in Section~\ref{sec:partial} 
that utilizes the partial observations, where we use 
2-fold cross-fitting 
to obtain $\hat{\mu}_x$ and $\hat\mu_y$ 
from 1) the Lasso using \texttt{cv.glmnet} (10-fold cross-validation), 
2) ridge regression using \texttt{cv.glmnet} (10-fold cross-validation), 
and 3) random forests using \texttt{grf} R-package. 
The parameter $\lambda$ for $\ell_1$-regularization is chosen
by 10-fold cross-validation 
with the \texttt{1se} criterion, implemented 
in the same way as that in the \texttt{cv.glmnet} R function, i.e., we
selects the largest $\lambda$ within one \texttt{se} of CV error from 
the smallest CV error.
We use the modular regression without structure learning. 
These implementations are compared to the default Lasso 
using \texttt{cv.glmnet} fitted over the $(X,Y)$ joint observations, also 
using \texttt{1se} criterion 
for 10-fold validation. Here we 
omit the results for our methods and the Lasso using the \texttt{min} option (selecting minimum CV error) for cross-validation
because the Lasso performs far less stable in this case.
The average prediction MSEs over $N=100$ independent runs 
are in Figure~\ref{fig:real_pobs}.  

\begin{figure}[htbp]
\centering 
\includegraphics[width=5in]{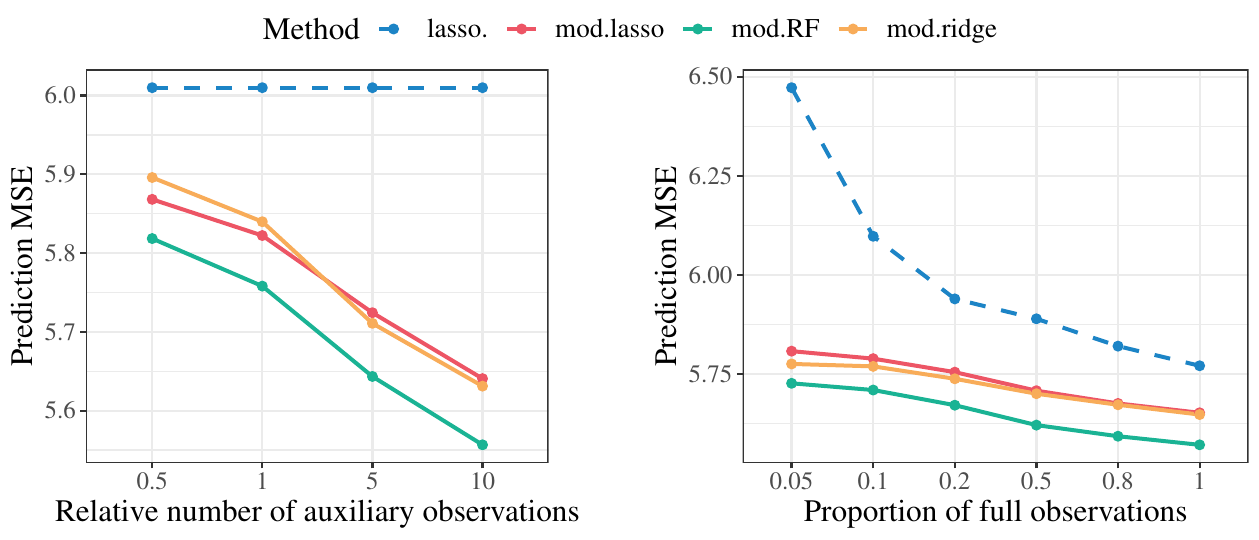}
\caption{Average prediction MSEs on the test data in settings (i) left and (ii) right. All methods 
are \texttt{1se} cross-validated. The $x$-axis represents the value of $\rho$ 
in both plots. Modular regression 
with all   three base learners
substantially reduce  the prediction MSE.}
\label{fig:real_pobs}
\end{figure}

The left panel in Figure~\ref{fig:real_pobs}
shows the results for setting (i), where Lasso uses a fixed number 
of joint observations. 
As the number of auxiliary observations increases, 
our modular regression achieves smaller prediction error, 
showing quite substantial improvement due to incorporating auxiliary observations. 

The right panel presents
those for setting (ii). 
Naturally, the performance of the Lasso (blue, dashed line) improves as $\rho$, 
the proportion of joint observations, increases. 
Our modular regression, with all of the three regressors, 
outperforms the Lasso by utilizing auxiliary observations, including  
$\rho=1$ without missing data. 
Surprisingly, keeping the total sample size fixed, 
we do not see much variation in the performance of modular regression (all solid lines) as
$\rho$ varies: 
the performance with only $5\%$ joint observations 
is comparable to that with more than $50\%$ joint observations. 
Our method achieves very similar effective sample size as 
full observations on this dataset. 
On the other hand, 
this phenomenon also indicates that we are in a regime where 
the irreducible error in $Y$ is large compared to the learnable part. 
In the next part, we are to utilize semi-synthetic data 
to evaluate the performance of our method in a setting with slighly stronger signal.

\subsection{Semi-synthetic data}

As discussed in Section~\ref{subsec:proj}, a naive implementation of our procedure is computationally prohibitive in high-dimensional settings. Thus, in the following, we evaluate the shortcut described in Section~\ref{subsec:proj} and compare it with
the standard cross-fitting implementation. 
We keep the choice of $X$ and $Z$ as before, 
and randomly subsample without replacement 
the training and test folds, where we observe $(X,Z,Y)$ 
for  $n=1000$ training data, 
but only $X$ for $n_{\test}=1794$ test sample. 
This is a scenario where 
only those easier-to-measure social-related 
covariates are available at 
the time of prediction, while the pre-collected training data 
contain both health and social covariates.

\paragraph{Data generating process.}
As the signal-to-noise ratio in the original data 
(for both $(X,Y)$ and $(Z,Y)$ regression) 
is extremely low, 
we enhance the signal with the following synthetic data generating process 
to draw a more informative comparison. 
We standardize all features using the original dataset, 
from which we subsample a set of observations aside from 
the training and test data. On this set, 
we run the Lasso for $Y$ given $Z$ which finds $13$ nonzero regression coefficients, 
and for $Y$ given $X$ which finds $12$ nonzero coefficients; 
we then reorder the features so that $Z_{1:13}$ and $X_{1:12}$ 
have with nonzero coefficients. 
For each $j\in\{1,\dots,12\}$, we run a Lasso for $X_{j}$ over $Z_{1:13}$ 
on this fold, and store all coefficients in 
the $j$-th column of a matrix $\hat B\in\RR^{13\times 12}$. 
 
We generate the training and testing data by $Z_i = Z_i^{\textrm{org}}+ \epsilon_i^z$, 
where $Z_i^{\textrm{org}}$ is the original observation, and $\epsilon_i^z\sim N(0,\sigma_z^2)$ 
is independent noise. 
We then compute $\mu_x(Z_i) = 2.5\cdot\hat{B}^\top Z_{i,1:13}$, 
and generate $X_{i,1:12} = \mu_x(Z_i) + \epsilon_i^x$, 
where $\epsilon_i^x\in \RR^{12}$ is i.i.d.~noise from $N(0,0.25\cdot \mathbf{1}_{12})$ 
to match the standardized variance in $X$, 
and $X_{13:p_x}$ are obtained by permuting each columns in the (standardized) 
original data matrices. 
Finally, we generate $Y_i=\mu_y(Z_i) + \epsilon_i^y$, where 
$\mu_y(Z_i)=Z_i^\top \gamma$ where the first $12$ entires in $\gamma\in \RR^{p_z}$ 
equals $0.5$ while the others equal to zero, and $\epsilon_i^y\sim N(0,4)$ 
is the i.i.d.~random noise.  
This setup ensures $X\indep Y \given Z$, but the 
sparse linear model $\EE[Y\given X]=X^\top\theta^*$ 
does not necessarily hold, and the true parameters $\theta^*$ are not 
available. 
We thus focus on the prediction performance. 
We vary the signal-to-noise ratio by setting $\sigma_z\in\{1,2\}$.

\paragraph{Methods.}
We evaluate two implementations of the modular regression: 
\begin{enumerate}[label=(\roman*)]
    \item Cross fitting in Section~\ref{subsec:highd_method}. 
    We use two-fold cross-fitting with cross-validated Lasso 
    and ridge regression to form $\hat\mu_x$ and $\hat\mu_y$, 
    and then use 10-fold cross-validation to 
    decide the penalty parameter $\lambda$ in~\eqref{eq:def_modular_highd} by either 
    (a) \texttt{min}: minimal CV error or (b) \texttt{1se}: the same as the 
    default implementation in \texttt{cv.glmnet} R function which 
    selects the largest $\lambda$ within one \texttt{se} of CV error for stable performance. 
    \item Projection shortcut in Section~\ref{subsec:proj}. 
    We use ridge projection with regularization parameters $(\eta_x,\eta_y)$ 
    for $\bPi_x,\bPi_y$, and then run \texttt{cv.glmnet} (with both (a) \texttt{min} and 
    (b) \texttt{1se} choice of cross-validation)
    for $X$ and $(\bPi_y+\bPi_x-\bPi_x\bPi_y)Y$, where  
    $(\eta_x,\eta_y)$ are chosen by 10-fold cross-validation 
    to minimize CV error. 
\end{enumerate}

The above two implementations are compared to two baselines: 
\begin{enumerate}[label=(\roman*)]
    \setcounter{enumi}{2}
    \item Oracle modular regression. Set  
    $\hat\mu_x$ and $\hat\mu_y$ as  ground truth,  
    and then the same as (i).
    \item Lasso: Run \texttt{cv.glmnet} on $(X,Y)$ with both (a) \texttt{min}
    and (b) \texttt{1se} cross-validation. 
\end{enumerate}

\paragraph{Results.}
The boxplots for all methods with $N=100$ independent runs 
are in Figure~\ref{fig:real_synthetic}, 
where we compare the performance under different cross-validation options. 

\begin{figure}[htbp]
    \centering 
    \includegraphics[width=6in]{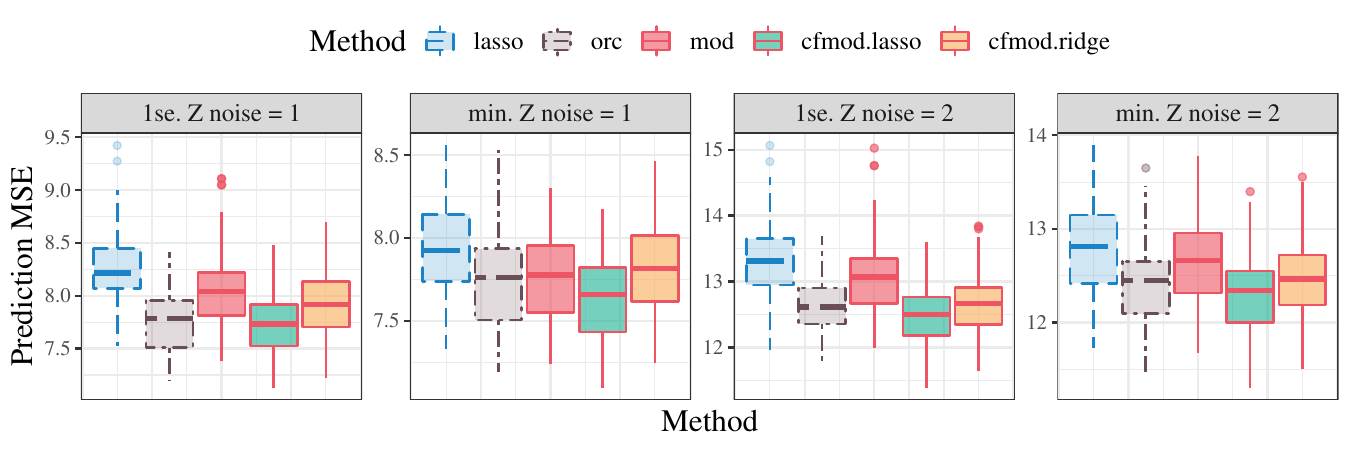}
    \caption{Boxplots of prediction MSEs on the test data. 
    Each subplot summarizes methods with one
    cross validation option (\texttt{1se} or \texttt{min}) under one value of $\sigma_z$. 
    \texttt{Method} stands for the projection shortcut (\texttt{mod}), 
    and cross-fitting with Lasso (\texttt{cfmod.Lasso})
    and ridge regression (\texttt{cfmod.ridge}).
    The Lasso (\texttt{Lasso}) 
    and oracle modular regression (\texttt{orc})
    with the same cv option are plotted for comparison. The oracle modular regression substantially reduces MSE; the cross-fitting implementation with Lasso and ridge regression is comparable to the oracle; 
    the projection shortcut is slightly inferior but still 
    improves upon the Lasso.}
    \label{fig:real_synthetic}
\end{figure}

The patterns across different values of $\sigma_z$ are similar. 
Among the baselines, the Lasso with minimal CV error 
is more accurate than \texttt{1se} (see the blue boxplots in 
the last two columns versus in
the first two columns), while the 
oracle modular regression always achieves smaller prediction MSE 
than the Lasso with the corresponding CV option (grey).

Our modular regression method performs 
reasonably well when the conditional mean functions $\mu_x(\cdot)$ and $\mu_y(\cdot)$ 
are estimated. 
When they are estimation by (i) ridge projection shortcut, 
modular regression with both \texttt{min} (see the last column) 
and \texttt{1se} (see the second column) improves upon 
the Lasso, although it is sometimes less accurate than the oracle. 
This shows the projection shortcut in Section~\ref{subsec:proj} 
is a reliable alternative to the more computationally intensive cross-fitting approach. 
When they are estimated by (ii) cross fitting (see the first and third columns), 
our method improves upon the original Lasso with both 
Lasso (green) and ridge regression (yellow)
as the regressor, and the performance is comparable to the oracle. 
The Lasso as the regressor is slightly better than the ridge regression; this 
may be due to the true sparse linear data generating process. 
In general, for the (ii) implementation, 
cross-validation with \texttt{min} CV error achieves smaller test MSE 
than \texttt{1se}, while the latter sees larger improvement upon the Lasso.

%% file: appendix.tex


\section{Deferred theoretical results} 
\label{app:glm}

This section presents the omitted theorem and proof for 
modular generalized linear regression in Section~\ref{subsec:glm}. 

\begin{assumption}\label{assump:glm}
$\theta^*\in \Theta$ for a compact set $\Theta$. 
Also, $\ell(x,y,\theta)$ is three-times-differentiable and convex in $\theta\in \Theta$. 
Let $\nabla_\theta h(x,\theta)\in \RR^{p_x}$ 
denote the first-order derivative of $h(x,\theta)$ in $\theta$, 
and similarly $\nabla_\theta^2 h(x,\theta)\in \RR^{p_x\times p_x}$ the second-order, 
and $\nabla_\theta^3 h(x,\theta)\in \RR^{p_x\times p_x\times p_x}$ the third-order ones.  
There exists some $L\colon \cX\to \RR$
such that $\EE[L(X)]<\infty$, and
$\|\nabla_\theta h(x,\theta) - \nabla_\theta h(x,\theta')\|\leq L(x) \cdot \|\theta-\theta'\|$ 
for any $\theta,\theta'\in \Theta$. 
Also, there exists some $m\colon \cX\times\Theta\to \RR$ 
such that $\EE[m(X)]<\infty$ and 
$\|\nabla_\theta^3 h(x,\theta) - \nabla_\theta^3 h(x,\theta')\|_{\normalfont \textrm{Fro}}
\leq m(x) \cdot \|\theta-\theta'\|$ for any $\theta,\theta'\in \Theta$. 
\end{assumption}

\begin{theorem}
    \label{thm:glm}
Suppose Assumptions~\ref{assump:ci} and~\ref{assump:glm} hold, 
and $\|\hat\mu_{x}^{(k)}-\mu_x\|_{L_2(\PP_Z)}\|\hat\mu_{y}^{(k)}-\mu_y\|_{L_2(\PP_Z)}=o_P(1/\sqrt{n})$ 
for $k=1,2$. 
Let $\hat\theta_n^\mod$ be the unique minimizer of~\eqref{eq:def_mod_glm}, 
and $\theta^*$ be the unique minimizer of $\EE\big[ \ell(X_i,Y_i,\theta^*) \big]$.
Let $\mu_x(\cdot) = \EE[f(X_i)\given Z_i=\cdot]$ and $\mu_y(\cdot) = \EE[Y_i\given Z_i=\cdot]$. 
Define the influence function  
$\phi(x,y,z) = \mu_x(z)g(y) + f(x) \mu_y(z) - \mu_x(z)\mu_y(z) + \nabla_\theta h(x,\theta^*)$.
Then $\sqrt{n}(\hat\theta_n^\mod - \theta^*)\stackrel{d}{\to}N(0,\Cov(\phi(X_i,Y_i,Z_i)))$ 
and $\sqrt{n}(\hat\theta_n^\mod - \theta^*) = \frac{1}{\sqrt{n}}\sum_{i=1}^n \phi(X_i,Y_i,Z_i) + o_P(1/\sqrt{n})$ as $n\to \infty$. 
Furthermore, $\phi(X_i,Y_i,Z_i)$ is the efficient influence function 
for estimating $\theta^*$ under the model space obeying Assumption~\ref{assump:ci}.
\end{theorem}

\begin{proof}[Proof of Theorem~\ref{thm:glm}]
The proof idea is similar to 
that of~\citet[Proposition E.2 and Theorem 3.11]{jin2021attribute}. 
Write $\hat{\ell}(X_i,Y_i,Z_i,\theta)= \big[\hat{\mu}_x (Z_i) g(Y_i) + f(X_i) \hat{\mu}_y (Z_i) - \hat{\mu}_x (Z_i)\hat{\mu}_y (Z_i)\big]^\top\theta + h(X_i,\theta)$. 
We first show that $\hat\theta_n^\mod\stackrel{P}{\to} \theta^*$ as $n\to \infty$. 
Since the score function $s$ is differentiable and convex in $\theta$, 
equivalently, $\hat\theta_n^\mod\in \Theta$ is the unique solution to
\$
\hat{L}_n^\mod(\theta)&:=  \frac{1}{n}\sum_{i=1}^n \big[\hat{\mu}_x (Z_i) g(Y_i) + f(X_i) \hat{\mu}_y (Z_i) - \hat{\mu}_x (Z_i)\hat{\mu}_y (Z_i)+ \nabla_\theta h(X_i,\theta)\big] 
= 0,
\$
while the population parameter $\theta^*$ is the unique solution to 
\$
L(\theta):=\EE\big[ \nabla_\theta \ell(X_i,Y_i,\theta)\big]=
\EE\big[ g(Y_i)f(X_i) +\nabla_\theta h(X_i,\theta)\big] = 0.
\$
We also define the empirical score at any $\theta\in \Theta$ as 
\$
\hat{L}_n(\theta) = \frac{1}{n}\sum_{i=1}^n \big[ g(Y_i)f(X_i) + \nabla_\theta h(X_i,\theta)\big].
\$
Then for any fixed $\theta\in \Theta$, we note that 
\$ 
&\hat{L}_n^\mod(\theta) - \hat{L}_n(\theta) 
= \frac{1}{n}\sum_{i=1}^n \big[\hat{\mu}_x (Z_i) g(Y_i) + f(X_i) \hat{\mu}_y (Z_i) - \hat{\mu}_x (Z_i)\hat{\mu}_y (Z_i) - g(Y_i)f(X_i)\big] \\
&= \frac{1}{n}\sum_{k=1}^2 \sum_{i\in \cI_k}\big[\hat{\mu}^{(k)}_x (Z_i) g(Y_i) + f(X_i) \hat{\mu}^{(k)}_y (Z_i) - \hat{\mu}^{(k)}_x (Z_i)\hat{\mu}^{(k)}_y (Z_i) - g(Y_i)f(X_i)\big].
\$
Similar to the proof of Theorem~\ref{thm:lowd} (see equation~\eqref{eq:lowd_err}), 
under the given 
consistency condition for $\hat\mu_x^{(k)}$ and $\hat\mu_y^{(k)}$, 
we know $\sup_{\theta\in\Theta} \big|\hat{L}_n^\mod(\theta) - \hat{L}_n(\theta)\big|= o_P(1)$. Furthermore, the law of large numbers implies 
$\frac{1}{n}\sum_{i=1}^n |\hat{L}_n(\theta)-L(\theta)|=o_P(1)$ 
for any fixed $\theta\in\Theta$. 
Therefore, $ \big|\hat{L}_n^\mod(\theta) -  {L} (\theta)\big| = o_P(1)$ 
for any fixed $\theta\in \Theta$. 

Since $\Theta$ is compact and $\nabla_\theta h(x,\theta)$ 
is $L(x)$-Lipschitz in $\theta\in\Theta$ for any $x\in \cX$, we know that 
for any $\theta,\theta'\in \Theta$, by the triangular inequality, 
\$
\big| \hat{L}_n(\theta) - \hat{L}_n(\theta') \big| \leq 
\frac{1}{n}\sum_{i=1}^n L(X_i) \cdot \|\theta- \theta'\|, 
\$
where $\frac{1}{n}\sum_{i=1}^n L(X_i) =O_P(1)$ since $\EE[L(X)]<\infty$. 
Thus, for any fixed $\epsilon>0$, 
there exists a finite subset $\{\theta_1,\dots,\theta_{N_\epsilon}\}$ of $\Theta$ 
such that $\sup_\theta\inf_{1\leq i\leq N_\epsilon}|\hat{L}_n(\theta)-\hat{L}_n(\theta_i)|\leq \epsilon$. 
Similar arguments applied to $L(\theta)$, together with the convergence, 
implies $\sup_{\theta\in\Theta} |\hat{L}_n(\theta) - L(\theta)| = o_P(1)$. 
Thus we have $\sup_{\theta\in\Theta}|\hat{L}_n^\mod(\theta)-L(\theta)|=o_P(1)$. 
On the other hand, the compactness of $\Theta$ and 
the uniqueness of $\theta^*$ as a solution to 
$L(\theta)=0$ implies the well-separatedness condition 
(c.f.~\citet[Theorem 5.9]{van2000asymptotic}): 
for any $\epsilon>0$, one could find some $\delta>0$ such that 
$\inf_{\theta\colon \|\theta-\theta^*\|\geq \delta} \big|L(\theta)\big|
> 2\epsilon$. 
Since $\sup_{\theta\in\Theta}|\hat{L}_n^\mod(\theta)-L(\theta)|=o_P(1)$,  
for any fixed $\epsilon>0$, one could find some $\delta>0$ 
such that for $n$ sufficiently large, 
$\PP \big( \inf_{\theta\colon \|\theta-\theta^*\|\geq \delta} \frac{1}{n}|\hat{L}_n^\mod(\theta)|
> \epsilon \big)\geq 1-\epsilon$, i.e., 
$\PP(\|\hat\theta_n^\mod - \theta^*\|\leq \delta)\geq 1-\epsilon$. 
This proves $\hat\theta_n^\mod - \theta^* = o_P(1)$. 

We now proceed to show the asymptotic normality of $\hat\theta_n^\mod$. 
Recall that $\hat{L}_n^\mod(\hat\theta_n^\mod)=0$. Taylor expansion 
around $\theta^*$ then gives  
\#\label{eq:glm_exp}
0 = \hat{L}_n^\mod(\theta^*) + \nabla_\theta \hat{L}_n^\mod(\theta^*)   (\hat\theta_n^\mod - \theta^*) + 1/2\cdot (\hat\theta_n^\mod - \theta^*)\nabla_\theta^2 \hat{L}_n^\mod(\tilde\theta_n)(\hat\theta_n^\mod - \theta^*)
\#
for some $\tilde\theta_n$ that lies on the segment between $\hat\theta_n^\mod$ 
and $\theta^*$, which implies $\tilde\theta_n=\theta^*+o_P(1)$. 
Here $\nabla_\theta \hat{L}_n^\mod \in \RR^{p_x\times p_x}$ 
and $\nabla_\theta^2 \hat{L}_n^\mod \in \RR^{p_x\times p_x\times p_x}$ 
are second and third order derivative of $\sum_{i=1}^n \hat{\ell}(X_i,Y_i,Z_i,\theta)$. 
By definition 
\$
\nabla_\theta \hat{L}_n^\mod(\theta) = \frac{1}{n}\sum_{i=1}^n \nabla_\theta^2 h(X_i,\theta),
\quad 
\nabla_\theta^2 \hat{L}_n^\mod(\theta) = \frac{1}{n}\sum_{i=1}^n \nabla_\theta^3 h(X_i,\theta). 
\$
The law of large numbers imply 
$\nabla_\theta \hat{L}_n^\mod(\theta^*) = \EE[\nabla_\theta^2 h(X_i,\theta^*)] + o_P(1)$ 
where $o_P(1)$ means a matrix whose every entry is $o_P(1)$, 
hence $\nabla_\theta \hat{L}_n^\mod(\theta^*)$ is invertible for sufficiently large $n$. 
Also, by the Lipschitz condition of $\nabla_\theta^3 h(x,\theta)$ 
and the triangular inequality, we have 
\$
\big|\nabla_\theta^2 \hat{L}_n^\mod(\theta^*) - \nabla_\theta^2 \hat{L}_n^\mod(\tilde\theta_n) \big|
\leq \frac{1}{n}\sum_{i=1}^n m(X_i) \|\theta^*-\tilde\theta_n\| = o_P(1)
\$
since $\EE[m(X_i)]<\infty$. Hence $\nabla_\theta^2\hat{L}_n^\mod(\tilde\theta_n)=O_P(1)$. 
Returning to~\eqref{eq:glm_exp}, we have 
\$
0 &= \hat{L}_n^\mod(\theta^*) + \big\{ \EE[\nabla_\theta^2 h(X_i,\theta^*)] + o_P(1)\big\} 
(\hat\theta_n^\mod - \theta^*) + O_P\big(  \|\hat\theta_n^\mod -\theta^*\|^2 \big) \\ 
&= \hat{L}_n^\mod(\theta^*) + \big\{ \EE[\nabla_\theta^2 h(X_i,\theta^*)] + o_P(1)\big\} 
(\hat\theta_n^\mod - \theta^*), 
\$
and thus 
\#\label{eq:glm_1}
\hat\theta_n^\mod - \theta^* = \EE\big[\nabla_\theta^2 h(X_i,\theta^*)\big] ^{-1} \hat{L}_n^\mod(\theta^*) + o_P\big(\hat{L}_n^\mod(\theta^*)\big).
\#
Finally, we note that under the consitency condition of 
$\hat\mu_x^{(k)}$ and $\hat\mu_y^{(k)}$, 
similar to the proof of Theorem~\ref{thm:lowd},  
one could show that 
\$
\hat{L}_n^\mod(\theta^*) 
&=  \frac{1}{n}\sum_{i=1}^n \big[ {\mu}_x (Z_i) g(Y_i) + f(X_i) {\mu}_y (Z_i) - {\mu}_x (Z_i) {\mu}_y (Z_i)+ \nabla_\theta h(X_i,\theta^*)\big] + o_P(1/\sqrt{n}).
\$
Because $X\indep Y\given Z$, we have 
$\EE\big[ {\mu}_x (Z_i) g(Y_i) + f(X_i) {\mu}_y (Z_i) - {\mu}_x (Z_i) {\mu}_y (Z_i)+ \nabla_\theta h(X_i,\theta^*)\big] = \EE[\nabla_\theta s(X_i,Y_i,\theta^*)]=0$, hence 
$\hat{L}_n^\mod(\theta^*) = O_P(1/\sqrt{n})$, which, combined with~\eqref{eq:glm_1}, implies 
\$
&\hat\theta_n^\mod - \theta^* 
=   \EE\big[\nabla_\theta^2 h(X_i,\theta^*)\big] ^{-1} \hat{L}_n^\mod(\theta^*) + o_P(1/\sqrt{n}) \\ 
&=  \frac{1}{n}\sum_{i=1}^n\EE\big[\nabla_\theta^2 h(X_i,\theta^*)\big] ^{-1} \big[ {\mu}_x (Z_i) g(Y_i) + f(X_i) {\mu}_y (Z_i) - {\mu}_x (Z_i) {\mu}_y (Z_i)+ \nabla_\theta h(X_i,\theta^*)\big] + o_P(1/\sqrt{n}).
\$
We thus complete the proof of the asymptotic expansion of $\hat\theta_n^\mod$ in Theorem~\ref{thm:glm}. 

Finally, the efficient influence function for estimating $\theta^*$ among 
all models obeying $X\indep Y\given Z$ can be obtained by, similar to Theorem~\ref{thm:lowd}, 
projecting the influence function of $\hat\theta_n^\ml$ onto the 
tangent space $\cT=\cT_1\oplus \cT_2 \oplus \cT_3$ 
where $\cT_j$ is defined as in~\eqref{eq:tj_space}. 
Standard M-estimator theory~\citep{van2000asymptotic} gives 
$\hat\theta_n^\ml - \theta^* = \frac{1}{n}\sum_{i=1}^n \phi^\ml(X_i,Y_i)+o_P(1/\sqrt{n})$ 
where 
\$
\phi^\ml (x,y) = \EE\big[\nabla_\theta^2 h(X_i,\theta^*)\big] ^{-1} \big[ g(Y_i)f(X_i) + \nabla_\theta h(X_i,\theta^*) \big].
\$
The projection follows exactly the same idea as Theorem~\ref{thm:glm} 
and one could see $\phi$ is the efficient influence function. 
We thus complete the proof of Theorem~\ref{thm:glm}. 
\end{proof}

\section{Deferred discussion}

\subsection{Deferred discussion for Remark~\ref{rem:finite_sample}}
\label{app:detail_finite_sample}
Below, we provide a (simplified) finite-sample  analysis 
of the two methods: (i) OLS, and (ii) modular OLS, where 
$Y\given Z$ and $X\given Z$ are estimated using OLS, 
under a joint Gaussian distribution. This discussion elaborates on 
the computation details for Remark~\ref{rem:finite_sample}.
%

To be more specific, the cross-term $C=\EE[XY]$ is estimated via $\hat{C}_\ols:=\frac{1}{n}\sum_{i=1}^nX_iY_i$ 
or $\hat{C}_\mod = \frac{1}{n}\sum_{i=1}^n Y_i \hat\mu_x(Z_i) + X_i \hat\mu_y(Z_i) - \hat\mu_x(Z_i)\hat\mu_y(Z_i)$. 
Hereafter, we override the notations and use $X,Y,Z$ to denote data vectors in $\RR^{n}$. 
With OLS regression, 
we have $\hat\mu_x(z)=\hat\theta_x^\top z$, 
where $\hat\theta_x = (Z^\top Z)^{-1} Z^\top X$ 
is the fitted OLS coefficient of $X$ on $Z$,  
and similarly $\hat\mu_y(z) = \hat\theta_y^\top z$, 
where $\hat\theta_y = (Z^\top Z)^{-1}  Z^\top Y$. 
Then, our modular estimator for the cross-term 
$\hat{C}_{\lm}$  is 
\$
&\hat{C}_{\mod}
= \frac{1}{n}\sum_{i=1}^n Y_i Z_i^\top \hat\theta_x + X_iZ_i^\top \hat\theta_y - Z_i^\top \hat\theta_x Z_i^\top \hat\theta_y\\ 
&= \frac{1}{n} \Big(Y^\top Z\hat\theta_x + X^\top Z \hat\theta_y - \hat\theta_x^\top Z^\top Z \hat\theta_y \Big)\\
&= \frac{1}{n} \Big(Y^\top Z(Z^\top Z)^{-1}  Z^\top X + X^\top Z (Z^\top Z)^{-1}  Z^\top Y - X^\top Z(Z^\top Z)^{-1}   Z^\top Z (Z^\top Z)^{-1}  Z^\top Y \Big) \\ 
&=  \frac{1}{n} \Big(Y^\top Z(Z^\top Z)^{-1}  Z^\top X + X^\top Z (Z^\top Z)^{-1}  Z^\top Y - X^\top Z(Z^\top Z)^{-1}  Z^\top Y \Big) \\ 
&= \frac{1}{n}  Y^\top Z(Z^\top Z)^{-1}  Z^\top X ,
\$
where the third row plugs in the definition of 
$\hat\theta_x$ and $\hat\theta_y$, and the last line 
uses the fact that $Y^\top Z(Z^\top Z)^{-1}  Z^\top X = X^\top Z(Z^\top Z)^{-1}  Z^\top Y$. 
The two estimators are 
$\hat\theta_n^\ols=(X^\top X/n)^{-1}\hat{C}_\ols$, 
and $\hat\theta_n^\mod = (X^\top X/n)^{-1}\hat{C}_\mod$, respectively. 

Suppose $X,Y,Z$ are all one-dimensional and 
jointly Gaussian. 
Then, there exists a decomposition 
$X=\alpha Z+ \epsilon_x$ and $Y=\beta Z+\epsilon_y$, 
where conditional on $Z$, 
$\epsilon_x$ and $\epsilon_y$ are normal with 
mean zero 
and mutually independent 
under the assumption 
$X\indep Y\given Z$.  
We further assume $\epsilon_x\given Z\sim N(0,\sigma_x^2)$ and $\epsilon_y\given Z\sim N(0,\sigma_y^2)$. 
Then, the population OLS parameter is 
$\theta^* = \alpha\beta/(\alpha^2+\sigma_x^2)$. Also, 
\$
\EE[\hat\theta_n^\mod] = \EE\bigg[\EE\bigg[ \frac{Y^\top Z \cdot Z^\top X}{X^\top X \cdot Z^\top Z} \bigggiven X,Z\bigg] \bigg]
= \beta \EE\bigg[  \frac{Z^\top Z \cdot Z^\top X}{X^\top X \cdot Z^\top Z}  \bigg]
= \beta \EE\bigg[ \EE\bigg[  \frac{ Z^\top X}{X^\top X } \bigggiven X\bigg] \bigg],
\$
where we repeatedly used the tower property. 
Using the joint Gaussianity of $(X,Z)$, 
we know that $\EE[Z\given X]=\frac{\alpha}{\alpha^2+\sigma_x^2}X$, which gives 
$\EE[\hat\theta_n^\mod] = \theta^*$, hence 
$\hat\theta_n^\mod$ is unbiased. 
Similarly one can show $\hat\theta_n^\ols$ is unbiased. 

We now study the variances of the two estimators. 
By the decomposition of variance, 
\$
\Var(\hat\theta_n^\mod) 
&= \EE[\Var(\hat\theta_n^\mod\given X,Z)]
+ \Var(\EE[\hat\theta_n^\mod\given X,Z]) \\ 
&= \EE\bigg[ \frac{(Z^\top X)^2 \sum_{i=1}^n Z_i^2 \sigma_y^2 }{(X^\top X \cdot Z^\top Z)^2}\bigg]
+ \Var\bigg( \frac{\beta Z^\top X}{X^\top X}  \bigg)\\
&=\sigma_y^2 \EE\bigg[ \frac{(Z^\top X) ^2 }{(X^\top X)^2 \cdot Z^\top Z }\bigg]
+ \beta^2 \Var\bigg( \frac{Z^\top X}{X^\top X}  \bigg).
\$
Similarly, using the fact that $Y\indep X\given Z$,
\$
\Var(\hat\theta_n^\ols) 
&= \EE[\Var(\hat\theta_n^\ols\given X,Z)]
+ \Var(\EE[\hat\theta_n^\ols\given X,Z]) \\ 
&= \EE\bigg[ \frac{ \sum_{i=1}^n X_i^2 \sigma_y^2 }{(X^\top X)^2}\bigg]
+ \Var\bigg( \frac{\beta Z^\top X}{X^\top X}  \bigg) =\sigma_y^2 \EE\bigg[ \frac{1}{X^\top X  }\bigg]
+ \beta^2 \Var\bigg( \frac{Z^\top X}{X^\top X}  \bigg).
\$
This gives (using $X=\alpha Z+ \epsilon_x$)
\$
\Var(\hat\theta_n^\ols) -\Var(\hat\theta_n^\mod) 
&= \sigma_y^2 \EE\bigg[ \frac{ Z^\top Z \cdot X^\top X - (Z^\top X) ^2 }{(X^\top X)^2 \cdot Z^\top Z }\bigg]\\
&= \sigma_y^2 \EE\bigg[ \frac{ Z^\top Z \cdot \epsilon_x^\top \epsilon_x - (Z^\top \epsilon_x) ^2 }{(X^\top X)^2 \cdot Z^\top Z }\bigg] \geq 0
\$
by the Cauchy-Schwarz inequality 
(it is a strict inequality as long as $\epsilon_x$ is not deterministic given $Z$). 
That is, in finite sample, 
both estimators are unbiased, while 
the modular estimator is non-inferior to OLS in terms of variance.  

\subsection{Deferred discussion on alternative estimators}
\label{app:discuss_alternative}

In this part, we elaborate our discussion on the alternative estimators in Remark~\ref{rem:alternative}. {In the following, we consider using powerful machine learning algorithms (black-boxes)  to estimate the nuisance components $\mu_x$ and $\mu_y$.
They may achieve a smaller MSE than the OLS estimator, potentially at the cost of some bias. 
The semiparametric efficiency in Theorem 2 of our manuscript does not account for such cases. In the following, we conduct additional experiments to explore this regime.}



Surprisingly, we found that  (a) outcome regression  can achieve a smaller MSE than the OLS, and that (b) orthogonal regression is often less accurate.  
Nonetheless, our modular estimator often achieves the smallest (or nearly smallest) MSE. 

In Figure~\ref{fig:simu_1}, we show results for OR (standing for outcome regression), PS (standing for orthogonal regression), OLS, and modular regression under several simple data-generating processes, where the true regression functions only involve $\leq 4$ entries. 
We vary the signal-to-noise ratio 
and the smoothness of the underlying regression functions across settings. 
We label the first row as (a-c), and the second as (d-f).  
Settings (a) and (b) has $Z\in \RR^{p_{20}}$, and the true regression functions are a combination of indicator functions and smooth exponential functions. 
Setting (c) involves a larger number of combinations of polynomial and exponential functions for $Z\in\RR^{20}$. 
Settings (d) and (e) adds a term $\sqrt{n}Z_3Z_4$ to fixed regression functions that only involve the first 4 variables in $Z\in\RR^{100}$. 
Setting (f) uses a similar model as (a-c) but kept $Z\in \RR^{100}$. 
(d-f) has  larger noise level. We omit details on the regression functions for brevity. 

\begin{figure}[h]
\centering
\begin{subfigure}[t]{0.32\linewidth}
    \centering
    \includegraphics[width=\linewidth]{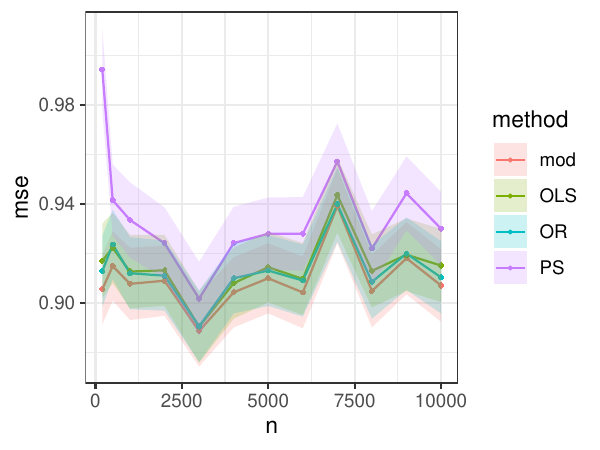} 
\end{subfigure} 
\begin{subfigure}[t]{0.32\linewidth}
    \centering
    \includegraphics[width=\linewidth]{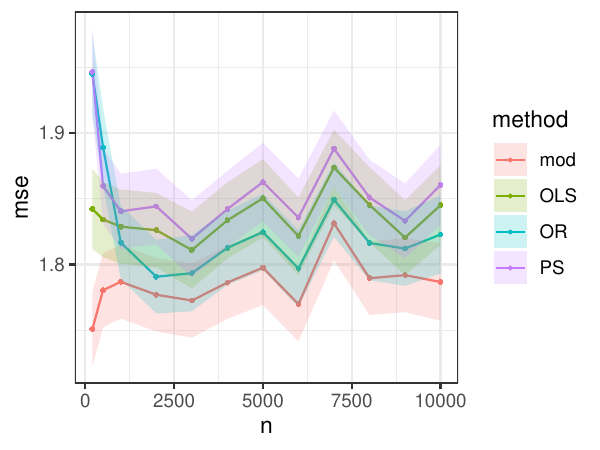}  
\end{subfigure} 
\begin{subfigure}[t]{0.32\linewidth}
    \centering
    \includegraphics[width=\linewidth]{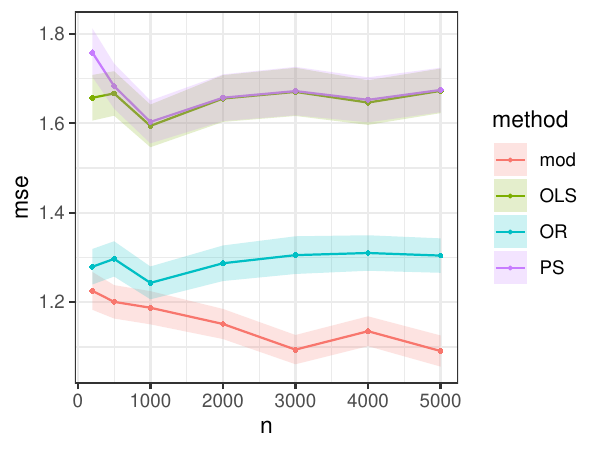}  
\end{subfigure} 
\begin{subfigure}[t]{0.32\linewidth}
    \centering
    \includegraphics[width=\linewidth]{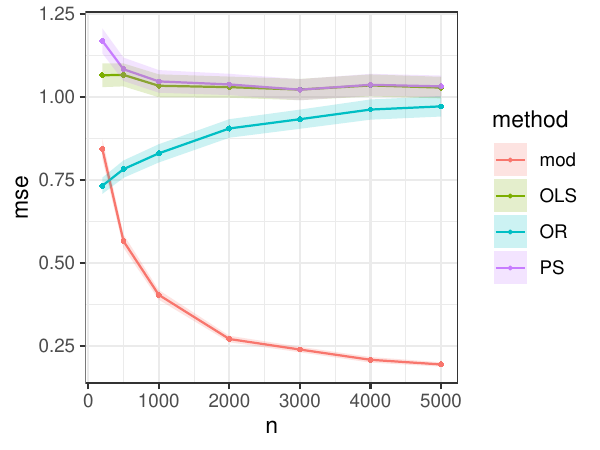} 
\end{subfigure} 
\begin{subfigure}[t]{0.32\linewidth}
    \centering
    \includegraphics[width=\linewidth]{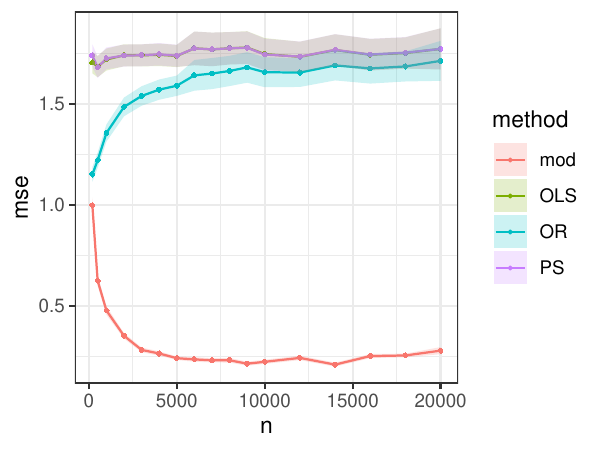}  
\end{subfigure} 
\begin{subfigure}[t]{0.32\linewidth}
    \centering
    \includegraphics[width=\linewidth]{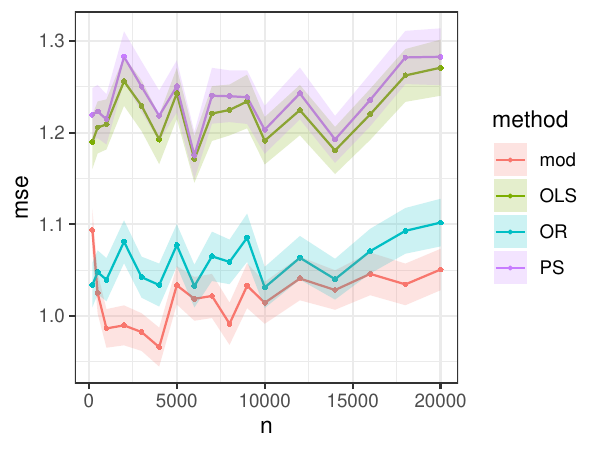}  
\end{subfigure} 
\caption{Root-mean-squared error scaled by $\sqrt{n}$, for all methods under different sample sizes.}
\label{fig:simu_1}
\end{figure}

We see that OR often performs better in terms of MSE; we conjecture that this is because the tuned ML regressor achieves a good bias-variance tradeoff. 
The PS approach is usually not very accurate, but surprisingly comparable to OLS. 
A potential explanation is that we are in the regime of ``undersmoothing'', i.e., the ML estimator has extremely small bias but larger variance, so that it is still possible to achieve good performance for both OR and PS approaches. 
Finally, however, the modular estimator always performs the best across all settings;
we conjecture that this is because it takes the best of the two worlds, powerful prediction machines and 
rigorous debiasing ideas from semiparametric statistics.

\subsection{Discussion on the estimation of $\hat{J}$}
\label{app:subsec_conditions}

Here, we discuss the conditions for estimating $\hat{J}$ and its compatibility with Assumption~\ref{assump:consistency}. 

Consistent estimation of $J$ require strong conditions such as the beta-min conditions or the irrepresentable conditions~\cite{zhao2006model}. However, it is also known that Lasso selects a superset of ``relevant'' features, i.e., entries whose coefficient is sufficently bounded away from zero (roughly, these are coefficients with absolute values larger than $O( \sqrt{s_0/n\cdot \log p }$, where $s_0$ is the number of nonzero coefficients) under less restrictive eigenvalue conditions of the covariates $X$ (see, e.g.,~\cite{van2009conditions,buhlmann2011statistics,buhlmann2010proposing}). When the entries in $J$ are sufficiently large, then a superset of $J$ can be recovered with high probability. Otherwise, entries with small coefficients would not incur too much bias in this process. 

On the other hand, consistency estimation of $J$ typically asks for $\log p_z/\sqrt{n}\to 0$, while Assumption~\ref{assump:consistency} asks for $\|\hat\mu_{x,z}^{(k)}(\cdot)-\mu_{x,z}^{(k)}(\cdot)\|_{L_2(\PP_Z)} \preceq n^{-1/4}$. When the regression function $\hat\mu_{x,z}$ is obtained from high-dimensional regression such as the Lasso, the convergence rate is typically $\sqrt{s_{x,z}\log p_z/n}$, where $s_{x,z}$ is the sparsity level for an entry in $X$ over $Z$. These conditions on $p_z$ are typically compatible for the two parts (as they are both satisfied for relatively small $p_z$). For example, if the sparsity $s_{x,z}$ is bounded, and $\log p_{z} = o(n^{1/4})$, then both conditions are satisfied. 

\section{Technical proofs}

\subsection{Proof of Theorem~\ref{thm:lowd}}
\label{app:subsec_thm_lowd}

\begin{proof}[Proof of Theorem~\ref{thm:lowd}]
Recall that $\theta^*$ is the least-squares estimator 
$\theta^* = \argmin_{\theta\in\RR^{p_x}} \EE[(Y-X^\top\theta)^2]$. The estimation equation gives 
\$
\hat\theta_n^\ols - \theta^* 
&= \Big(\sum_{i=1}^n X_iX_i^\top\Big)^{-1} \sum_{i=1}^n X_iY_i- \theta^* 
=\Big(\frac{1}{n}\sum_{i=1}^n X_iX_i^\top\Big)^{-1} \frac{1}{n}\sum_{i=1}^n X_i(Y_i-X_i^\top \theta^*).
\$
Since $X_i(Y_i-X_i^\top\theta^*)$ has 
finite second moment and we know 
$\EE[X(Y-X^\top\theta^*)]=0$ by the optimality of $\theta^*$, 
we have $\frac{1}{n}\sum_{i=1}^n X_i(Y_i-X_i^\top \theta^*)=O_P(1/\sqrt{n})$, where 
$O_P(1/\sqrt{n})$ represents a random vector 
whose each entry is $O_P(1/\sqrt{n})$. 
Since $X_iX_i^\top$ has finite expectation 
$\EE[XX^\top]\succ 0$, we have 
\$
\Big(\frac{1}{n}\sum_{i=1}^n X_iX_i^\top\Big)^{-1}
= \Big( \EE[XX^\top] + o_P(1)\Big)^{-1} = 
\big(\EE[XX^\top]\big)^{-1} + o_P(1),
\$
where $o_P(1)$ stands for a random matrix whose 
all entries converge in probability to zero. 
We thus have  
\$
\hat\theta_n^\ols - \theta^* 
&= \Big( \big(\EE[XX^\top]\big)^{-1} + o_P(1) \Big) \frac{1}{n}\sum_{i=1}^n X_i(Y_i-X_i^\top \theta^*) \\ 
&= \frac{1}{n}\sum_{i=1}^n \big(\EE[XX^\top]\big)^{-1}  X_i(Y_i-X_i^\top \theta^*) + o_P(1) \times O_P(1/\sqrt{n}) \\ 
&= \frac{1}{n}\sum_{i=1}^n\phi^\ols(X_i,Y_i)+o_P(1/\sqrt{n}),
\$
where 
$
\phi^\ols(x,y) = \EE[XX^\top]^{-1} x(y-x^\top \theta^*). 
$
For our modular estimator, we first note the closed form solution 
\$
\hat\theta_n^\mod -\theta^* = \Big(\sum_{i=1}^n X_iX_i^\top\Big)^{-1} \sum_{i=1}^n C_i -\theta^* = \Big(\sum_{i=1}^n X_iX_i^\top\Big)^{-1} \sum_{i=1}^n \big( C_i - X_iX_i^\top\theta^* \big).
\$
We now show that under the conditions in Theorem~\ref{thm:lowd}, 
one has 
\#\label{eq:lowd_err}
\frac{1}{n}\sum_{i=1}^n \big( C_i - X_iX_i^\top\theta^* \big)
= \frac{1}{n} \sum_{i=1}^n \big( C_i^* - X_iX_i^\top\theta^* \big) +o_P(1/\sqrt{n}),
\#
where $C_i^* = X_i\mu_y(Z_i)+\mu_x(Z_i)Y_i - \mu_x(Z_i)\mu_y(Z_i)$. 
To see this, by rearranging the terms, we have 
$\frac{1}{n}\sum_{i=1}^n (C_i-C_i^*)=(\textrm{i}) + (\textrm{ii})+ (\textrm{iii})$, 
where we define 
\$ 
\textrm{(i)} &=  \frac{1}{n}\sum_{k=1}^2 \sum_{i\in \cI_k} 
    \{X_i-\mu_x (Z_i)\} \{\hat\mu_y^{(k)} (Z_i) - \mu_y(Z_i)\}, \\ 
\textrm{(ii)} &=  \frac{1}{n}\sum_{k=1}^2 \sum_{i\in \cI_k}   \{\hat\mu_x^{(k)} (Z_i)-\mu_x(Z_i)\}\{Y_i - \mu_y(Z_i)\}   , \\ 
\textrm{(iii)} &=  \frac{1}{n}\sum_{k=1}^2 \sum_{i\in \cI_k}   \{\hat\mu_x^{(k)} (Z_i)-\mu_x(Z_i)\}\{\hat\mu_y^{(k)} (Z_i)  - \mu_y(Z_i)\} .
\$
We first bound the summation in (i). 
For each $k$, because $\hat\mu_y^{(k)}$ is obtained from 
the independent fold $\cI\backslash \cI_k$, 
for any $i\in \cI_k$, by the tower property, 
\$
\EE\big[  \{X_i-\mu_x (Z_i)\} \{\hat\mu_y^{(k)} (Z_i) - \mu_y(Z_i)\} \biggiven \cI\backslash \cI_k\big] = \EE\big[  \EE[ X_i-\mu_x (Z_i) \given Z_i ]   \{\hat\mu_y^{(k)} (Z_i) - \mu_y(Z_i)\}   \biggiven \cI\backslash \cI_k  \big] = 0,
\$
and they are i.i.d.~copies conditional on $\cI\backslash\cI_k$
with conditional variance 
\$
\EE\big[ \{X_i-\mu_x (Z_i)\}^2 \{\hat\mu_y^{(k)} (Z_i) - \mu_y(Z_i)\}^2   \biggiven \cI\backslash \cI_k  \big] = o_P(1).
\$
The Markov's inequality thus implies 
\$
\frac{1}{n}\sum_{i\in \cI_k}  \{X_i-\mu_x (Z_i)\} \{\hat\mu_y^{(k)} (Z_i) - \mu_y(Z_i)\} = o_P(1/\sqrt{n}),
\$
hence $\textrm{(i)}=o_P(1/\sqrt{n})$. 
Similar arguments also apply to (ii) and yield $\textrm{(ii)}=o_P(1/\sqrt{n})$. 
Finally, conditional on $\cI\backslash\cI_k$, 
for any $i\in \cI_k$, by the Cauchy-Schwarz inequality, 
\$
&\bigg| \sum_{i\in \cI_k}\{\hat\mu_x^{(k)} (Z_i)-\mu_x(Z_i)\}\{\hat\mu_y^{(k)} (Z_i)  - \mu_y(Z_i)\}  \bigg| \\
&\leq \bigg[  \sum_{i\in \cI_k}\{\hat\mu_x^{(k)} (Z_i)-\mu_x(Z_i)\}^2 \bigg]^{1/2}
\cdot \bigg[  \sum_{i\in \cI_k}\{\hat\mu_y^{(k)} (Z_i)  - \mu_y(Z_i)\}^2 \bigg]^{1/2}.
\$
By the Markov's inequality, we have  
\$
\sum_{i\in \cI_k}\{\hat\mu_x^{(k)} (Z_i)-\mu_x(Z_i)\}^2
&= O_P\big( |\cI_k| \cdot \EE\big[\{\hat\mu_x^{(k)} (Z_i)-\mu_x(Z_i)\}^2\biggiven \cI\backslash \cI_k\big] \big) \\
&= O_P\big( n \cdot \|\hat\mu_x^{(k)}-\mu_x\|_{L_2(\PP_Z)}^2 \big).
\$
With similar arguments applied to the second summation, we arrive at 
\$
\bigg| \sum_{i\in \cI_k}\{\hat\mu_x^{(k)} (Z_i)-\mu_x(Z_i)\}\{\hat\mu_y^{(k)} (Z_i)  - \mu_y(Z_i)\}  \bigg|
\leq O_P\big(n\cdot  \|\hat\mu_x^{(k)}-\mu_x\|_{L_2(\PP_Z)} \|\hat\mu_y^{(k)}-\mu_y\|_{L_2(\PP_Z)} \big).
\$
Combining $k=1,2$, we have 
\$
\textrm{(iii)} = O_P\big( \|\hat\mu_x^{(k)}-\mu_x\|_{L_2(\PP_Z)} \|\hat\mu_y^{(k)}-\mu_y\|_{L_2(\PP_Z)} \big) = o_P(1/\sqrt{n})
\$
by the conditions in Theorem~\ref{thm:lowd}.
Putting together our bounds on the three terms, we prove 
the claim in~\eqref{eq:lowd_err}. 
Note that $\EE[C_i^*-X_iX_i^\top\theta^*]=0$, hence 
$\frac{1}{n}\sum_{i=1}^n (C_i^*-X_iX_i^\top\theta^*) = O_P(1/\sqrt{n})$, thus 
\$
\hat\theta_n^\mod - \theta^* 
&= \Big[ \frac{1}{n}\sum_{i=1}^n X_iX_i^\top\Big]^{-1} \cdot \frac{1}{n}\sum_{i=1}^n \big( C_i - X_iX_i^\top\theta^* \big) \\
&= \big[  \EE[XX^\top]  ^{-1} + o_P(1)\big] 
\cdot \bigg[ \frac{1}{n}\sum_{i=1}^n \big( C_i^* - X_iX_i^\top\theta^* \big) +o_P(1/\sqrt{n})  \bigg] \\
&=\frac{1}{n}\sum_{i=1}^n  \EE[XX^\top]  ^{-1}\big( C_i^* - X_iX_i^\top\theta^* \big)
+ o_P(1/\sqrt{n}).
\$
That is, we obtain the asymptotic linear expansion 
$\hat\theta_n^\mod - \theta^* = \frac{1}{n}\sum_{i=1}^n \phi^\mod(X_i,Y_i,Z_i) + o_P(1/\sqrt{n})$, where the influence function is given by   
\$
\phi^\mod(x,y,z)  &= \EE[XX^\top]^{-1} \big(x\mu_y(z) + \mu_x(z) y - \mu_x(z)\mu_y(z) - xx^\top \theta^*\big). 
\$
The Central Limit Theorem thus gives the asymptotic distribution 
$\sqrt{n}(\hat\theta_n^\mod - \theta^*)\stackrel{d}{\to} N(0,\Sigma^\mod)$ 
where $\Sigma^\mod = \Cov(\phi^\mod(X_i,Y_i,Z_i))$. 

We now proceed to show that $\phi^\mod(X_i,Y_i,Z_i)$ 
is the efficient influence function for estimating $\theta^*$ under 
the current distribution $\PP$. 
Our argument is similar to that of~\cite[Lemma 9]{rotnitzky2019efficient} 
and~\cite[Theorem 4.5]{tsiatis2006semiparametric}. 
We let $\cP$ denote the collection of all 
distributions obeying Assumption~\ref{assump:ci}, so 
$\PP\in \cP$. 
For any $P\in \cP$, let $p$ denote the 
density of $P$ with respect to the base measure $\mu$. 
Then the joint density $p(x,y,z)$ decomposes as
$
p(x,y,z) = p(z) p(x\given z) p(y\given z), 
$
where $p(z)$ is the marginal density of $Z$, 
$p(x\given z)$ is the conditional density of $P_{X\given Z}$, and $p(y\given z)$ is the 
conditional density of $P_{Y\given Z}$. 
By~\citep[Lemma 1.6]{van2003unified}, the tangent 
space of $\cP$ at $P$ is given by 
$\cT = \cT_1 \oplus \cT_2 \oplus \cT_3$, 
where $\{\cT_j\}_{j=1}^3$ are orthogonal spaces. 
More specifically, $\cT_1$ is the closed linear 
span of socres of one-dimensional regular parametric 
submodel 
$
\gamma \mapsto p(z;\gamma) p(x\given z) p(y\given z),
$
$\cT_2$ is that of the parametric submodel 
$
\gamma \mapsto p(x\given z;\gamma) p(z) p(y\given z),
$
and $\cT_3$ is that of the parametric submodel 
$
\gamma \mapsto p(y\given z;\gamma) p(z) p(x\given z). 
$
Similar to~\citep[Theorem 4.5]{tsiatis2006semiparametric}, these 
subspaces can be equivalently represented as 
\#\label{eq:tj_space}
\cT_1 &= \big\{  f(Z)\mapsto \RR^{p_x} \colon \EE[f(Z)]=0       \big\}, \notag\\
\cT_2 &= \big\{  f(X,Z)\mapsto \RR^{p_x} \colon \EE [f(X,Z)\given Z ] = 0\big\} \notag\\
\cT_3 &= \big\{  f(Y,Z)\mapsto \RR^{p_x} \colon \EE [f(Y,Z)\given Z ] = 0\big\},
\#
where all the functions in them are additionally square integrable. 

By standard semiparametric theory~\citep{tsiatis2006semiparametric}, 
the efficient influence function, 
denoted as $\phi^*$, can be obtained by the projection of
$\phi^{\ols}$ onto the tangent space $\cT$. 
That is, $\phi^*(x,y,z) = \phi^*_1(z)+\phi^*_2(x,z)+\phi^*_3(y,z)$, 
where $\phi_j^*$ is the projection of $\phi^\ols$ 
onto $\cT_j$ defined in~\eqref{eq:tj_space}. 
Hence 
\$
\phi^*(X,Y,Z) &= 
\EE\big[ \phi^\ols(X,Y) \biggiven Z \big] -  \EE\big[ \phi^\ols(X,Y)\big]
+ \EE\big[ \phi^\ols(X,Y) \biggiven Y,Z \big] \\
&\qquad  -  \EE\big[ \phi^\ols(X,Y)\biggiven Z\big] 
+ \EE\big[ \phi^\ols(X,Y) \biggiven X,Z \big] -  \EE\big[ \phi^\ols(X,Y)\biggiven Z\big] \\ 
&= \EE[XX^\top]^{-1} \big\{  \EE[X(Y-X^\top \theta^*) \given X,Z] 
+  \EE[X(Y-X^\top \theta^*) \given Y,Z] \\ 
&\qquad \qquad \qquad \qquad  - 
\EE[X(Y-X^\top \theta^*) \given Z] - \EE[X(Y-X^\top \theta^*) ]\big\} \\ 
&= \EE[XX^\top]^{-1} \big\{ X\EE[Y\given Z] - XX^\top \theta^* 
+  \EE[X\given Z] Y \\ 
&\qquad \qquad \qquad \qquad - \EE[XX^\top \theta^* \given Z]  - \EE[XY\given Z] + \EE[XX^\top\theta^*\given Z] 
    \big\} \\ 
    &= \EE[XX^\top]^{-1} \big\{ X\EE[Y\given Z] - XX^\top \theta^* 
    +  \EE[X\given Z] Y   - \EE[X \given Z] \EE[Y\given Z]   \big\} \\ 
    &= \phi^\mod(X,Y,Z),
\$
where the first and second equalities are from the projection 
onto subsapces, the third equality uses 
the fact that $\EE[X(Y-X^\top\theta^*)]=0$ 
and the tower property, and the fourth equality uses the 
conditional independence in Assumption~\ref{assump:ci}. 
We thus conclude the proof of Theorem~\ref{thm:lowd}.
\end{proof}

\subsection{Proof of Theorem~\ref{thm:highd}}
\label{app:subsec_thm_highd}

\begin{proof}[Proof of Theorem~\ref{thm:highd}]
For notational simplicity, throughout the proof 
we write $\hat\theta_n^\mod$ as $\hat\theta$. 
Denote $\Delta = \hat\theta - \theta^*$, 
and let $X\in \RR^{n\times p_x}$ be the design matrix. 
We define the vector 
$
G = \sum_{i=1}^n C_i \in \RR^{p_x},
$
where we recall the definition of $C_i$ in~\eqref{eq:C_default}.
Since $\hat\theta$ is the minimizer of~\eqref{eq:def_modular_highd}, 
we have 
\$
\frac{1}{2n}\hat\theta^\top X^\top X \hat\theta - \frac{1}{n}G^\top \hat\theta +\lambda \|\hat\theta\|_1 
\leq 
\frac{1}{2n} (\theta^*)^\top X^\top X  \theta^* - \frac{1}{n}G^\top \theta^* +\lambda \| \theta^*\|_1, 
\$
which implies 
\$
\frac{1}{2n} \Delta^\top X^\top X\Delta \leq \frac{1}{n}\big\langle (\theta^*)^\top X^\top X - G, \Delta\big\rangle + \lambda \|\theta^*\|_1 - \lambda \|\hat\theta\|_1.
\$
Let $S$ be the set of indices for nonzero entries of $\theta^*$, 
and $\theta_S$ be the subvector containing all entries with indices in $S$. 
Then 
\$
\|\theta^*\|_1 - 
\|\hat\theta\|_1 =\|\theta^*_S\|_1 -  \|\hat\theta_S\|_1 - \|\hat\theta_{S^c}\|_1 
\leq  \|\Delta_S\|_1 - \|\Delta_{S^c}\|_1,
\$
where the last inequality uses 
triangular inequality and the fact that $\theta_{S^c}^*\equiv 0$. 
We thus have 
\#\label{eq:1}
\frac{1}{2n} \Delta^\top X^\top X\Delta 
&\leq \frac{1}{n}\big\langle (\theta^*)^\top X^\top X - G, \Delta\big\rangle + \lambda \|\Delta_S\|_1  - \lambda \|\Delta_{S^c}\|_1 \notag \\
&\leq \frac{1}{n}\|\Delta\|_1 \big\| (\theta^*)^\top X^\top X - G \big\|_\infty + \lambda \|\Delta_S\|_1  - \lambda \|\Delta_{S^c}\|_1.
\#
The first implication is that, as the left-handed side is non-negative, once $\lambda \geq \frac{2}{n} \| (\theta^*)^\top X^\top X - G  \|_\infty$, we have
$
0 \leq \frac{\lambda}{2}\|\Delta\|_1 + \lambda\|\Delta_S\|_1 - \lambda \|\Delta_{S^c}\|_1,
$
hence 
$
\Delta \in \CC_3 = \big\{x\in \RR^p\colon \|x_{S^c}\|_1 \leq 3\|x_{S}\|_1\big\}.
$
By Assumption~\ref{assump:RSC},~\eqref{eq:1} further implies 
\$
\zeta \|\Delta\|_2^2 \leq \frac{\lambda}{2} \|\Delta\|_1 + \lambda \|\Delta_S\|_1 - \lambda \|\Delta_{S^c}\|_1
\leq \frac{3\lambda}{2} \|\Delta_S\|_1 \leq \frac{3\lambda}{2}\sqrt{k}\|\Delta\|_2,
\$
where $k =|S| = \|\theta^*\|_0$ is the sparsity level. 
To summarize, under Assumption~\ref{assump:RSC}, 
the solution $\hat\theta$ satisfies
\#\label{eq:first_bd}
\|\hat\theta - \theta^*\|_2 \leq \frac{3\lambda \sqrt{k}}{2\zeta}
\#
for any regularization parameter $\lambda \geq  \frac{2}{n} \| (\theta^*)^\top X^\top X - G  \|_\infty$. 
We now write $\hat\mu_y\in \RR^{n}$ 
as the vector whose $i$-th entry is 
$\hat\mu_y^{(k)}(Z_i)$ for $i\in \cI_k$, 
and $\hat\mu_x\in \RR^{n\times p_x}$ whose $(i,j)$-th entry 
is $\hat\mu_{x,j}^{(k)}(Z_i)$ for $i\in \cI_k$. 
Similarly, $\mu_y\in \RR^{n}$ and $\mu_x\in \RR^{n\times p_x}$ 
record the ground truth of these regression functions. 
The error term is then 
\#\label{eq:err_decomp}
(\theta^*)^\top X^\top X - G
&= (\theta^*)^\top X^\top X - \hat\mu_y^\top X 
- Y^\top \hat\mu_x + \hat\mu_y^\top \hat\mu_x \notag \\
&= (Y- \mu_y)^\top (X-\mu_x) + (X\theta^*-Y)^\top X \notag \\
&\qquad + (\mu_y - \hat\mu_y)^\top (X-\mu_x) 
+ (Y-\mu_y)^\top (\mu_x-\hat\mu_x) + (\mu_y-\hat\mu_y)^\top (\mu_x-\hat\mu_x).
\#
For any constant $\delta\in(0,1)$, we define the event 
\#\label{eq:est_bound}
\cE_{\textrm{est}}^{n,\delta} = \Bigg\{\big\|\hat\mu_{x,j}^{(k)} -\mu_{x,j} \big\|_{L_2(\PP_Z)} 
,~ \big\|\hat\mu_{y}^{(k)} -\mu_{y}  \big\|_{L_2(\PP_Z)} \leq \frac{4 c_n \log(3p_x/\delta)}{n^{1/4}},~\forall 1\leq j\leq p_x,~\forall k=1,2\Bigg\}.
\#
Then under Assumption~\ref{assump:consistency}, 
for any constant $\delta\in(0,1)$, 
taking a union bound 
over $2p_x+2\leq 3p_x$ estimated functions 
for $1\leq j\leq p_x$ and $k=1,2$, 
we know that $\PP(\cE_\textrm{est}^{n,\delta})\geq 1-\delta$.  

We now proceed to analyze each entry $j\in \{1,\dots,p_x\}$ 
of the above error term. 
First, 
\$
\frac{1}{n}\big[(\mu_y - \hat\mu_y)^\top (X-\mu_x) \big]_j
= \frac{1}{n}\sum_{k=1}^2 \sum_{i\in \cI_k} \big(\mu_y(Z_i)-\hat\mu_y^{(k)}(Z_i) \big)
\big( X_{i,j} - \mu_{x,j}(Z_i)\big) 
\$
For each $k=1,2$, conditional on $\cI\backslash \cI_k$, 
the terms in the summation 
$(\mu_y(Z_i)-\hat\mu_y^{(k)}(Z_i) )
( X_{i,j} - \mu_{x,j}(Z_i))$, $i\in \cI_k$ 
are mutually independent with mean zero, since 
\$
&\EE\Big[  \big(\mu_y(Z_i)-\hat\mu_y^{(k)}(Z_i) \big)
\big( X_{i,j} - \mu_{x,j}(Z_i)\big)   \Biggiven \cI\backslash \cI_k \Big] \\
&= \EE\Big[   \big(\mu_y(Z_i)-\hat\mu_y^{(k)}(Z_i) \big)\cdot 
\EE\big[  X_{i,j} - \mu_{x,j}(Z_i)\biggiven Z_i,\cI\backslash \cI_k\big]      \Biggiven \cI\backslash \cI_k \Big] =0 
\$
by the tower property. 
Also, the absolute value of each term is 
bounded below $4c_0$  
by the boundedness in Assumption~\ref{assump:bounded}. 
We now define the event 
\$
\cE_{\textrm{cr}}^{1,k}(\delta) = \Bigg\{   
    \Big|\sum_{i\in \cI_k} \big(\mu_y(Z_i)-\hat\mu_y^{(k)}(Z_i) \big)
\big( X_{i,j} - \mu_{x,j}(Z_i)\big)\Big| 
\leq \max\Big\{ 16c_0 \log(2/\delta)/3,~ 2\hat\epsilon_{1}\sqrt{|\cI_k|\log(2/\delta)} \Big\} 
    \Bigg\},
\$
where we define 
\$
\hat\epsilon_{1}^2 =4 \sup_{k=1,2}\big\|\mu_y - \hat\mu_y^{(k)}\big\|_{L_2(\PP_Z)}^2 &\geq \EE\Big[  \big(\mu_y(Z_i)-\hat\mu_y^{(k)}(Z_i) \big)^2
\big( X_{i,j} - \mu_{x,j}(Z_i)\big)^2   \Biggiven \cI\backslash \cI_k\Big] .
\$
The Bernstein's inequality in Lemma~\ref{lem:bernstein}, 
implies $\PP\big( \cE_{\textrm{cr}}^{1,k}  \biggiven \cI\backslash \cI_k \big) \leq \delta$. 
%
Marginalizing out the conditional probability and 
taking a union bound over $k=1,2$, we know that 
$\PP(\cE_{\textrm{cr}}^1(\delta))\geq 1-\delta$ 
for any constant $\delta$, where we define 
$\cE_{\textrm{cr}}^1(\delta) = \cE_{\textrm{cr}}^{1,1} (\delta/2)\cup \cE_{\textrm{cr}}^{1,2}(\delta/2)$. 
That is, with probability at least $1-\delta$, 
\#\label{eq:err_1}
\Big|\frac{1}{n}\big[(\mu_y - \hat\mu_y)^\top (X-\mu_x) \big]_j\Big|
\leq \max\bigg\{ \frac{16c_0 \log(4/\delta)}{3n},~ \hat\epsilon_1\sqrt{\frac{2\log(4/\delta)}{n}} \bigg\}. 
\#
For each fixed $j$, applying exactly the same arguments to 
\$
\frac{1}{n}\big[(Y-\mu_y)^\top (\mu_x-\hat\mu_x)\big]_j 
= \frac{1}{n}\sum_{k=1}^2 \sum_{i\in \cI_k} \big(Y_i- \mu_y (Z_i) \big)
\big(  \mu_{x,j}(Z_i) - \hat\mu_{x,j}^{(k)}(Z_i)\big),
\$
with probability at least $1-\delta$, 
\#\label{eq:err_2}
\bigg|\frac{1}{n}\big[(Y-\mu_y)^\top (\mu_x-\hat\mu_x)\big]_j \bigg|
\leq \max\bigg\{ \frac{16c_0 \log(4/\delta)}{3n},~ c_0\hat\epsilon_{2,j}\sqrt{\frac{2\log(4/\delta)}{n}} \bigg\},
\#
where we define 
$
\hat\epsilon_{2,j} = 2 \sup_{k=1,2}\big\|\mu_{x,j} - \hat\mu_{x,j}^{(k)}\big\|_{L_2(\PP_Z)}.
$
The third term is 
\$
\frac{1}{n}\big[(\mu_y-\hat\mu_y)^\top (\mu_x-\hat\mu_x)\big]_j 
= \frac{1}{n}\sum_{k=1}^2 \sum_{i\in \cI_k} \big(\mu_y(Z_i)-\hat\mu_y^{(k)}(Z_i) \big)
\big(  \mu_{x,j}(Z_i) - \hat\mu_{x,j}^{(k)}(Z_i)\big).
\$
Here for each $k=1,2$, conditional on $\cI\backslash\cI_k$, 
each term 
$(\mu_y(Z_i)-\hat\mu_y^{(k)}(Z_i))
(  \mu_{x,j}(Z_i) - \hat\mu_{x,j}^{(k)}(Z_i))$, $i\in \cI_k$ 
in the above summation is i.i.d.~whose expectation is bounded as 
\$
&\bigg|\EE\Big[ \big(\mu_y(Z_i)-\hat\mu_y^{(k)}(Z_i) \big)
\big(  \mu_{x,j}(Z_i) - \hat\mu_{x,j}^{(k)}(Z_i)\big) \Biggiven \cI\backslash \cI_k  \Big]\bigg| \\
&\leq \big\|\mu_y - \hat\mu_y^{(k)}\big\|_{L_2(\PP_Z)} \big\|\mu_{x,j} - \hat\mu_{x,j}^{(k)}\big\|_{L_2(\PP_Z)} 
\leq \hat\epsilon_1 \cdot \hat\epsilon_{2,j}/4.
\$
Meanwhile, their second moments are bounded as 
\$
\EE\Big[ \big(\mu_y(Z_i)-\hat\mu_y^{(k)}(Z_i) \big)^2
\big(  \mu_{x,j}(Z_i) - \hat\mu_{x,j}^{(k)}(Z_i)\big)^2 \Biggiven \cI\backslash \cI_k  \Big]
\leq 4\big\|\mu_y - \hat\mu_y^{(k)}\big\|_{L_2(\PP_Z)}^2 
\leq \hat\epsilon_1^2 
\$
according to the bounededness in Assumption~\ref{assump:bounded}. 
Combining the above  bounds and invoking the 
Bernstein's inequality in Lemma~\ref{lem:bernstein}, 
it holds with probability at least $1-\delta$ that 
\$
   & \Big|\sum_{i\in \cI_k} \big(\mu_y(Z_i)-\hat\mu_y^{(k)}(Z_i) \big)
\big(  \mu_{x,j}(Z_i) - \hat\mu_{x,j}^{(k)}(Z_i)\big) \Big| \\
&\leq |\cI_k| \hat\epsilon_1  \hat\epsilon_{2,j}/4 
+ \max\Big\{ 16c_0 \log(2/\delta)/3,~ 2\hat\epsilon_{1}\sqrt{|\cI_k|\log(2/\delta)} \Big\} .
\$
Taking a union bound for $k=1,2$ implies that 
with probability at least $1-\delta$, 
\#\label{eq:err_3}
\Big|\frac{1}{n}\big[(\mu_y-\hat\mu_y)^\top (\mu_x-\hat\mu_x)\big]_j \Big| \leq \frac{\hat\epsilon_1 \hat\epsilon_{2,j}}{4}
+ \max\bigg\{ \frac{16c_0 \log(4/\delta)}{3n},~ \hat\epsilon_{1}\sqrt{\frac{2\log(4/\delta)}{n}} \bigg\}.
\#
Putting together~\eqref{eq:err_1},~\eqref{eq:err_2} and~\eqref{eq:err_3}, we know that with probability at least $1-\delta/2$, 
\$
&\big\|(\mu_y - \hat\mu_y)^\top (X-\mu_x) 
+ (Y-\mu_y)^\top (\mu_x-\hat\mu_x) + (\mu_y-\hat\mu_y)^\top (\mu_x-\hat\mu_x)\big\|_\infty \\
&\leq \frac{\hat\epsilon_1 \hat\epsilon_{2,j}}{4} + 2\max\bigg\{ \frac{16c_0 \log(24/\delta)}{3n},~ \hat\epsilon_1\sqrt{\frac{2\log(24/\delta)}{n}} \bigg\} \\ 
&\qquad 
+ \max\bigg\{ \frac{16c_0 \log(24/\delta)}{3n},~ c_0\hat\epsilon_{2,j}\sqrt{\frac{2\log(24/\delta)}{n}} \bigg\}.
\$
Further taking a union bound over the above event 
and~\eqref{eq:est_bound} for $\cE_{\textrm{est}}^{n,\delta/2}$, 
and using the fact that $\max\{a,b\}\leq a+b$ for $a,b\geq 0$, 
we know it holds with probability at least $1-\delta$ that 
\$
&\big\|(\mu_y - \hat\mu_y)^\top (X-\mu_x) 
+ (Y-\mu_y)^\top (\mu_x-\hat\mu_x) + (\mu_y-\hat\mu_y)^\top (\mu_x-\hat\mu_x)\big\|_\infty \\
&\leq\frac{4c_n^2 (\log(3p_x/\delta))^2}{\sqrt{n}} + \frac{16c_0 \log(24/\delta)}{n}+ \frac{(2+c_0)\log(3p_x/\delta)\sqrt{2\log(24/\delta)}}{n^{3/4}} .
\$
Recalling~\eqref{eq:err_decomp} and letting 
\$
\bar c_n = 4c_n^2 + 16c_0/\sqrt{n} + (2+c_0) /n^{1/4},
\$
we have $\bar c_n\to 0$ as $n\to \infty$ and 
with probability at least $1-\delta$, 
\$
\frac{1}{n}\big\|(\theta^*)^\top X^\top X - G \big\|_\infty \leq 
\frac{1}{n} \big\| (Y- \mu_y)^\top (X-\mu_x) + (X\theta^*-Y)^\top X \big\|_\infty +  \frac{\bar c_n (\log(3p_x/\delta))^2}{\sqrt{n}}.
\$
Combining this bound with~\eqref{eq:first_bd} 
and writing $D:=(Y- \mu_y)^\top (X-\mu_x) + (X\theta^*-Y)^\top X $,  
\$
&\PP\bigg( \|\hat\theta - \theta^*\|_2 \leq \frac{3\lambda \sqrt{k}}{2\zeta},~\forall ~\lambda 
\geq \frac{2 \|D\|_\infty}{n} +  \frac{2 \bar c_n (\log(3p_x/\delta))^2}{\sqrt{n}}  \bigg) \\
&\geq \PP\bigg( \frac{ \|(\theta^*)^\top X^\top X - G  \|_\infty}{n}
\leq \frac{\| D\|_\infty}{n} +  \frac{\bar c_n (\log(3p_x/\delta))^2}{\sqrt{n}}  \bigg)\geq 1-\delta,
\$
which proves~\eqref{eq:thm_highd}. 
\end{proof}

\subsection{Proof of Proposition~\ref{prop:miss}}
\label{app:subsec_proof_miss}

\begin{proof}[Proof of Proposition~\ref{prop:miss}]
Similar to the proof of Theorem~\ref{thm:lowd}, we are to show that $\hat{C}_\miss = C_{\miss}  + o_P(1/\sqrt{n_{xz}+n_{xyz}}) + o_P(1/\sqrt{n_{yz}+n_{xyz}})$, 
where 
$
C_\miss = {C}_{\miss}^{xz} +  {C}_{\miss}^{yz} -  {C}_{\miss}^{zz},
$ 
\$
 {C}_{\miss}^{xz} &=\frac{1}{n_{xz}+n_{xyz}}\sum_{i\in \cI^{xz}\cup\cI^{xyz}} X_i \mu_y (Z_i), \\
 {C}_{\miss}^{yz} &=\frac{1}{n_{yz}+n_{xyz}}\sum_{i\in \cI^{yz}\cup\cI^{xyz}} Y_i \mu_x (Z_i),\\
 {C}_{\miss}^{zz} &=\frac{1}{n_{xz}+n_{yz}+n_{xyz}}\sum_{i\in \cI^{xz}\cup \cI^{yz}\cup\cI^{xyz}} \mu_y (Z_i)\mu_x (Z_i).
\$
To see this, we note that $\hat{C}_\miss - C_{\miss} = \textrm{(i)} +\textrm{(ii)} + \textrm{(iii)}$, 
where 
\$
\textrm{(i)} &= \frac{1}{n_{xz}+n_{xyz}}\sum_{k=1}^2 \sum_{i\in \cI_k^{xz}\cup\cI_k^{xyz}}  \{X_i-\mu_x (Z_i)\} \{\hat\mu_y^{(k)}  (Z_i) - \mu_y(Z_i)\}, \\
\textrm{(ii)} &= \frac{1}{n_{yz}+n_{xyz}}\sum_{k=1}^2 \sum_{i\in \cI_k^{yz}\cup\cI_k^{xyz}}  \{Y_i-\mu_y (Z_i)\} \{\hat\mu_x^{(k)}  (Z_i) - \mu_x(Z_i)\}, \\ 
\textrm{(iii)} &= - \frac{1}{n_{xz}+n_{yz}+n_{xyz}}\sum_{i\in \cI^{xz}\cup \cI^{yz}\cup\cI^{xyz}} \{\hat\mu_x^{(k)}  (Z_i) - \mu_x(Z_i)\}\{\hat\mu_y^{(k)}  (Z_i) - \mu_y(Z_i)\}, \\ 
\textrm{(iv)} &= \frac{1}{n_{xz}+n_{xyz}}\sum_{k=1}^2 \sum_{i\in \cI_k^{xz}\cup\cI_k^{xyz}}   \mu_x (Z_i)  \{\hat\mu_y^{(k)}  (Z_i) - \mu_y(Z_i)\} \\ 
&\qquad - \frac{1}{n_{xz}+n_{yz}+n_{xyz}}\sum_{k=1}^2\sum_{i\in \cI_k^{xz}\cup \cI_k^{yz}\cup\cI_k^{xyz}}   \mu_x(Z_i) \{\hat\mu_y^{(k)}  (Z_i) - \mu_y(Z_i)\}, \\ 
\textrm{(v)} &= \frac{1}{n_{yz}+n_{xyz}}\sum_{k=1}^2 \sum_{i\in \cI_k^{yz}\cup\cI_k^{xyz}}   \mu_y (Z_i)  \{\hat\mu_x^{(k)}  (Z_i) - \mu_x(Z_i)\} \\ 
&\qquad - \frac{1}{n_{xz}+n_{yz}+n_{xyz}}\sum_{k=1}^2\sum_{i\in \cI_k^{xz}\cup \cI_k^{yz}\cup\cI_k^{xyz}}   \mu_y(Z_i) \{\hat\mu_x^{(k)}  (Z_i) - \mu_x(Z_i)\}.
\$
Note that each data in the $k$-th fold is independent of the fitted functions applied to them. Thus, each summation term for $i\in \cI_k^{xz}\cup\cI_k^{xyz}$ in (i) is i.i.d.~and mean zero conditional on $\hat\mu_y^{(k)}$, whose conditional variance is 
$\EE[(\hat\mu^{(k)}_y(Z)-\mu_y(Z))^2(X-\mu_x(Z))^2\given \hat\mu_y^{(k)}] = o_P(1)$. The Markov's inequality thus implies $\textrm{(i)}=o_P(1/\sqrt{n_{xz}+n_{xyz}})$. 
Similarly, we know $\textrm{(ii)}=o_P(1/\sqrt{n_{yz}+n_{xyz}})$. 
Also, since $\|\hat\mu_{x}^{(k)}-\mu_x\|_{L_2(\PP_Z)}\cdot\|\hat\mu_{y}^{(k)}-\mu_y\|_{L_2(\PP_Z)}=o_P(1/\sqrt{n})$, we have $\textrm{(iii)}=o_{P}(1/\sqrt{n_{xz}+n_{xyz}}+1/\sqrt{n_{yz}+n_{xyz}})$. 
In addition, let  $D_i:= \mu_x (Z_i)  \{\hat\mu_y^{(k)}  (Z_i) - \mu_y(Z_i)\} - \EE\big[ \mu_x (Z_i)  \{\hat\mu_y^{(k)}  (Z_i) - \mu_y(Z_i)\}\big]$ for $i\in\cI_k^{xz}\cup\cI_k^{xyz}$, where the expectation is conditional on data out of the $k$-th fold. 
Note that 
\$
\textrm{(iv)} &= \frac{1}{n_{xz}+n_{xyz}}\sum_{k=1}^2 \sum_{i\in \cI_k^{xz}\cup\cI_k^{xyz}} D_i   - \frac{1}{n_{xz}+n_{yz}+n_{xyz}}\sum_{k=1}^2\sum_{i\in \cI_k^{xz}\cup \cI_k^{yz}\cup\cI_k^{xyz}} D_i.
\$
For each $k$, 
conditional on $\{(Y_i,Z_i)\colon i\in (\cI^{xyz}\backslash \cI^{xyz}_k) \cup(\cI^{yz}\backslash \cI_k^{yz})\}$, 
we know that $\{D_i\}$ are i.i.d.~with mean zero. 
This implies 
\$
\big|\textrm{(iv)}\big| &\leq O_P(\|D_i\|_{L_2(\PP)}/\sqrt{n_{xz}+n_{xyz}}) + O_P(\|D_i\|_{L_2(\PP)}/\sqrt{n_{xz}+n_{yz}+n_{xyz}}) \\ 
&\leq o_P(1/\sqrt{n_{xz}+n_{xyz}}) + o_P(1/\sqrt{n_{xz}+n_{yz}+n_{xyz}}).
\$
The same bounds hold for $\textrm{(v)}$. Putting them together, we obtain 
\$
\hat{C}_\miss = C_{\miss} + o_P(1/\sqrt{n_{xz}+n_{xyz}}+1/\sqrt{n_{yz}+n_{xyz}}).
\$
Recall that $n_{xz}/(n_{xz}+n_{yz}+n_{xyz})\to \rho_{xz}$, 
and $n_{yz}/(n_{xz}+n_{yz}+n_{xyz})\to \rho_{yz}$.
Rearranging the terms in $C_\miss$ gives $C_\miss = C_\miss^*+o_P(1/\sqrt{n_{xz}+n_{xyz}}+1/\sqrt{n_{yz}+n_{xyz}})$, where
\$
C_{\miss}^* &= \sum_{i\in \cI^{xz}} \frac{X_i\mu_y(Z_i)}{n_{xz}+n_{xyz}} - \frac{\mu_y(Z_i)\mu_x(Z_i)}{n_{xz}+n_{yz}+n_{xyz}}   + \sum_{i\in \cI^{yz}} \frac{Y_i\mu_x(Z_i)}{n_{yz}+n_{xyz}} - \frac{\mu_y(Z_i)\mu_x(Z_i)}{n_{xz}+n_{yz}+n_{xyz}} \\ 
&\quad + \sum_{i\in \cI^{xyz}} \frac{X_i\mu_y(Z_i)}{n_{xz}+n_{xyz}} + \frac{Y_i\mu_x(Z_i)}{n_{yz}+n_{xyz}} - \frac{\mu_y(Z_i)\mu_x(Z_i)}{n_{xz}+n_{yz}+n_{xyz}} \\ 
&= \frac{1}{n_{xz}}\sum_{i\in \cI^{xz}} \Big( \frac{\rho_{xz}}{1-\rho_{yz}} X_i\mu_y(Z_i) - \rho_{xz}\mu_{y}(Z_i)\mu_x(Z_i)\Big) \\ 
&\quad + \frac{1}{n_{yz}}\sum_{i\in \cI^{yz}} \Big( \frac{\rho_{yz}}{1-\rho_{xz}} Y_i\mu_x(Z_i) - \rho_{yz}\mu_{y}(Z_i)\mu_x(Z_i)\Big) \\ 
&\quad + \frac{1}{n_{xyz}}\sum_{i\in \cI^{xyz}}\Big( \frac{ X_i\mu_y(Z_i)}{1-\rho_{yz}} + \frac{1}{1-\rho_{xz}} Y_i\mu_x(Z_i) -  \mu_{y}(Z_i)\mu_x(Z_i)\Big).
\$
Furthermore, we note that $C_\miss^* = \EE[XY]+O_P(1/\sqrt{n_{xz}+n_{xyz}}+1/\sqrt{n_{yz}+n_{xyz}})$. 
Similar to the arguments in the proof of Theorem~\ref{thm:lowd}, we have 
\$
\hat\theta_n^\mod -\theta^* = \EE[XX^\top]^{-1} (C_\miss^* - \hat\Sigma_{\miss}\theta^*).
\$
This concludes the proof of Proposition~\ref{prop:miss}.
\end{proof}

\section{Supporting lemmas}

\begin{lemma}[Bernstein's inequality]
\label{lem:bernstein}
Suppose $X_1,\dots,X_n$ are independent zero-mean 
random variables such that 
$|X_i|\leq M$ almost surely for some constant $M>0$.
Then for any constant $t>0$, 
\$
\PP\bigg( \Big|\sum_{i=1}^n X_i \Big|\geq t\bigg) 
\leq 2\exp\bigg(- \frac{t^2}{2\sum_{i=1}^n \EE[X_i^2] + 2Mt/3}\bigg).
\$
That is, for any $\delta\in(0,1)$, with probability at least $1-\delta$, it holds that 
\$
\big|\sum_{i=1}^n X_i \big| 
\leq \max \big\{ 2\sqrt{{\textstyle \sum_{i=1}^n} \EE[X_i^2] \log(2/\delta)} , 4M\log(2/\delta)/3 \big\}.
\$
\end{lemma}

\section{Deferred simulation results}
\label{app:simu_lowd}

\subsection{Low-dimensional simulation}

\begin{figure}[h]
\centering
\includegraphics[width=5in]{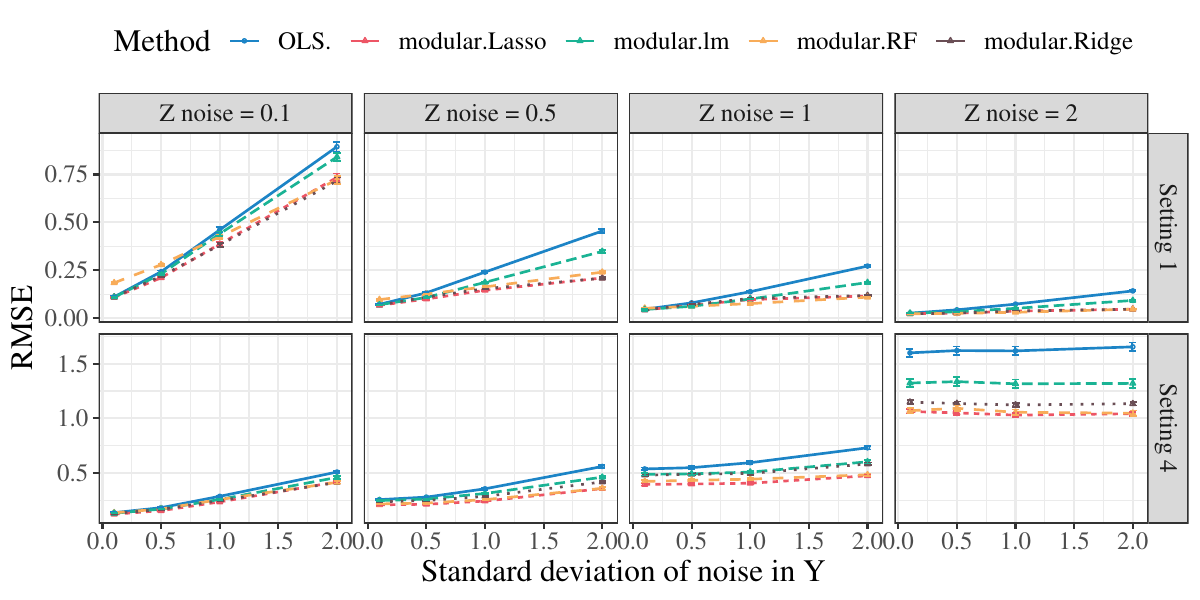}
\caption{Aggregate RMSE in settings 1 and 4. Details are
otherwise the same as Figure~\ref{fig:lowd_rmse_sub}.}
\label{fig:lowd_rmse_14_sub}
\end{figure}

\begin{figure}[h]
\centering
\includegraphics[width=5in]{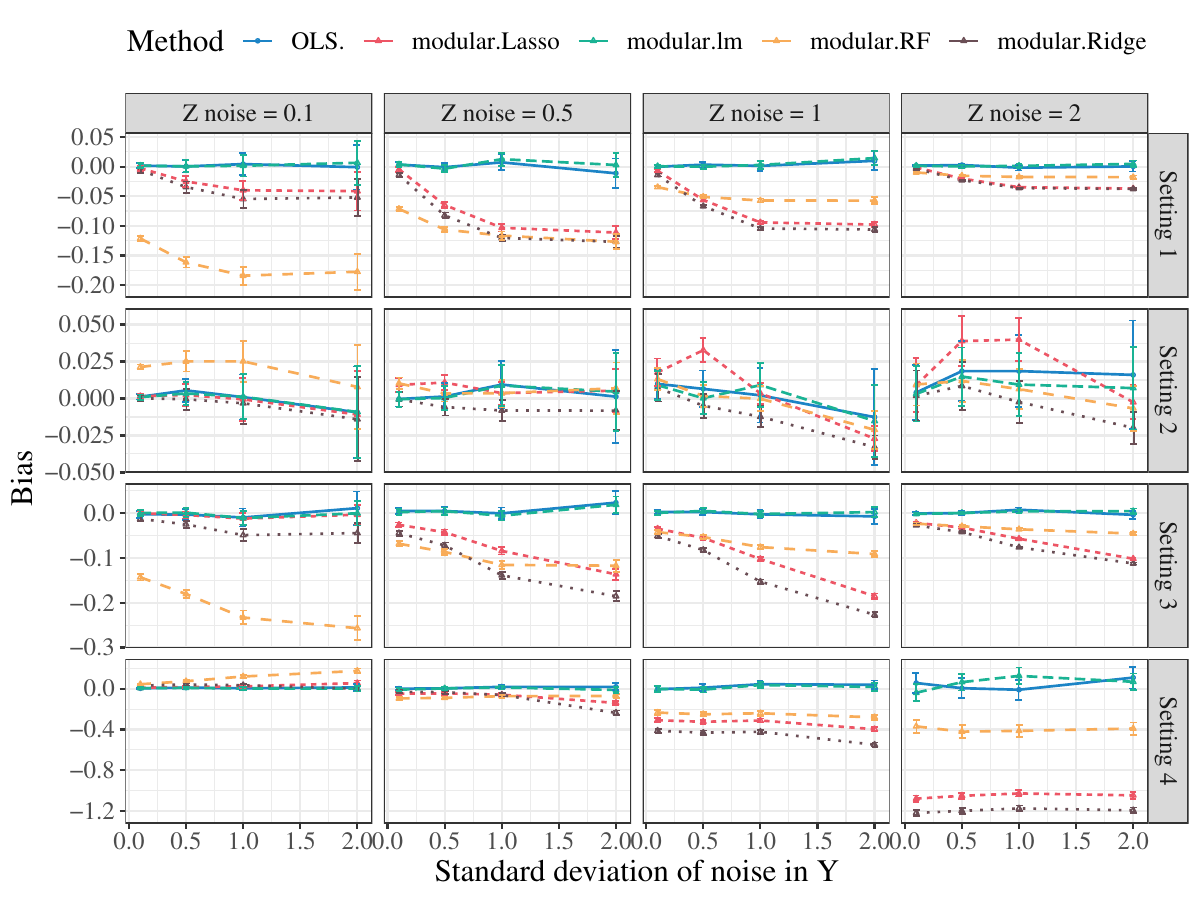}
\caption{Bias summed over all entries in all settings. Details are
otherwise the same as Figure~\ref{fig:lowd_rmse_sub}. 
Modular regression incurs a slightly larger bias 
than OLS when the sub-tasks are learned by flexible  machine learning algorithms.}
\label{fig:lowd_bias_sub}
\end{figure}
 
\begin{figure}[h]
\centering
\includegraphics[width=5in]{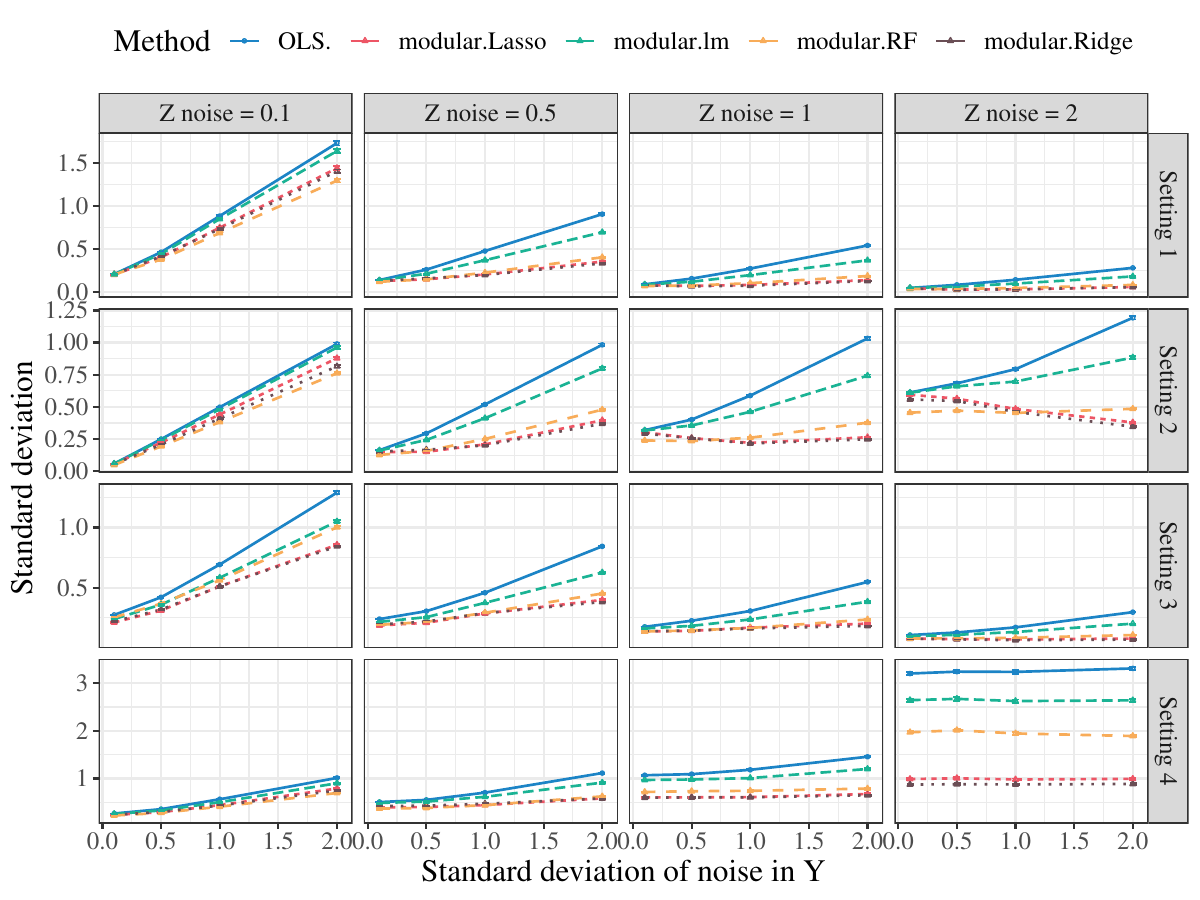}
\caption{Aggregated SD in all settings.
Details are
otherwise the same as Figure~\ref{fig:lowd_rmse_sub}.}
\label{fig:lowd_sd_sub}
\end{figure}

\begin{figure}[h]
    \centering
    \includegraphics[width=5in]{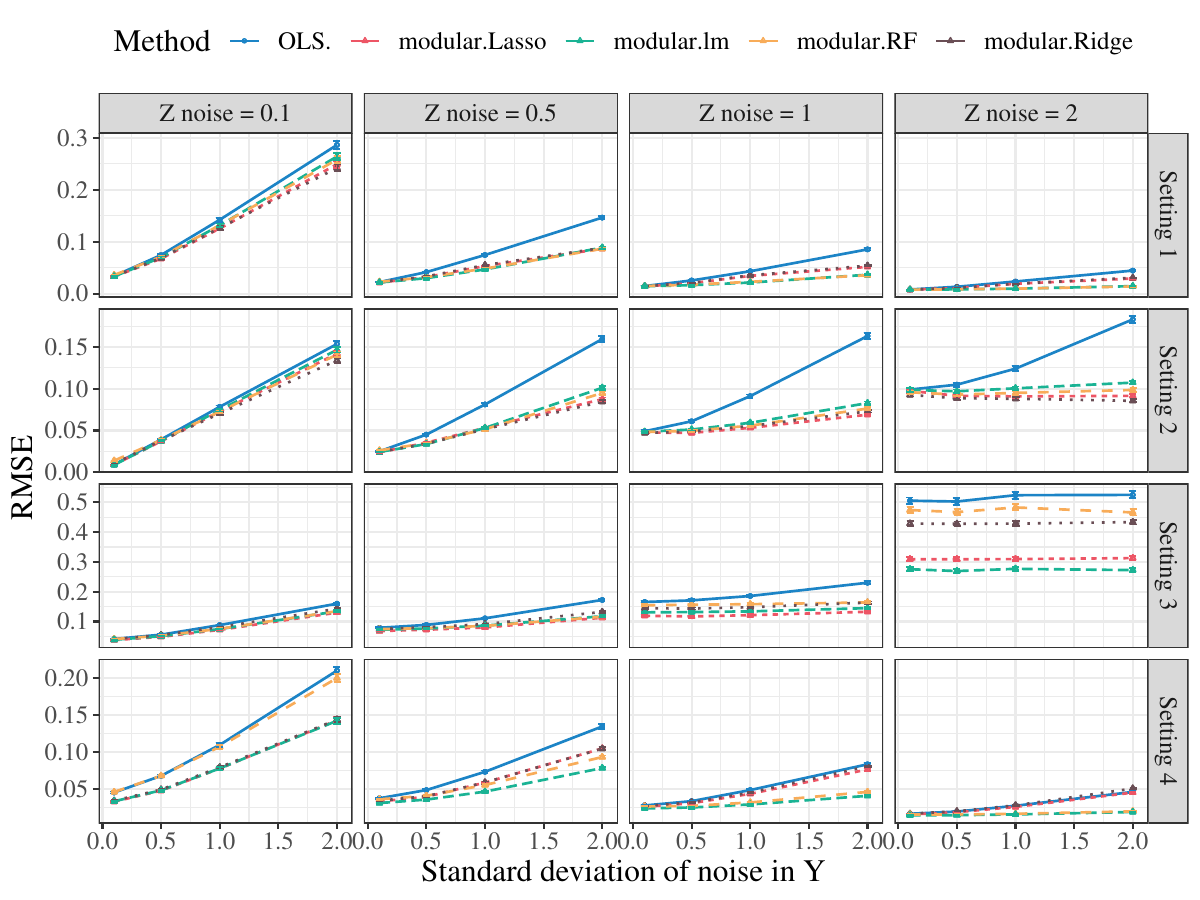}
    \caption{Aggregated RMSE when $n=2000$. Details are
    otherwise the same as Figure~\ref{fig:lowd_rmse_sub}.}
    \label{fig:lowd_rmse_sub_k}
\end{figure}

\begin{figure}[h]
\centering
\includegraphics[width=5in]{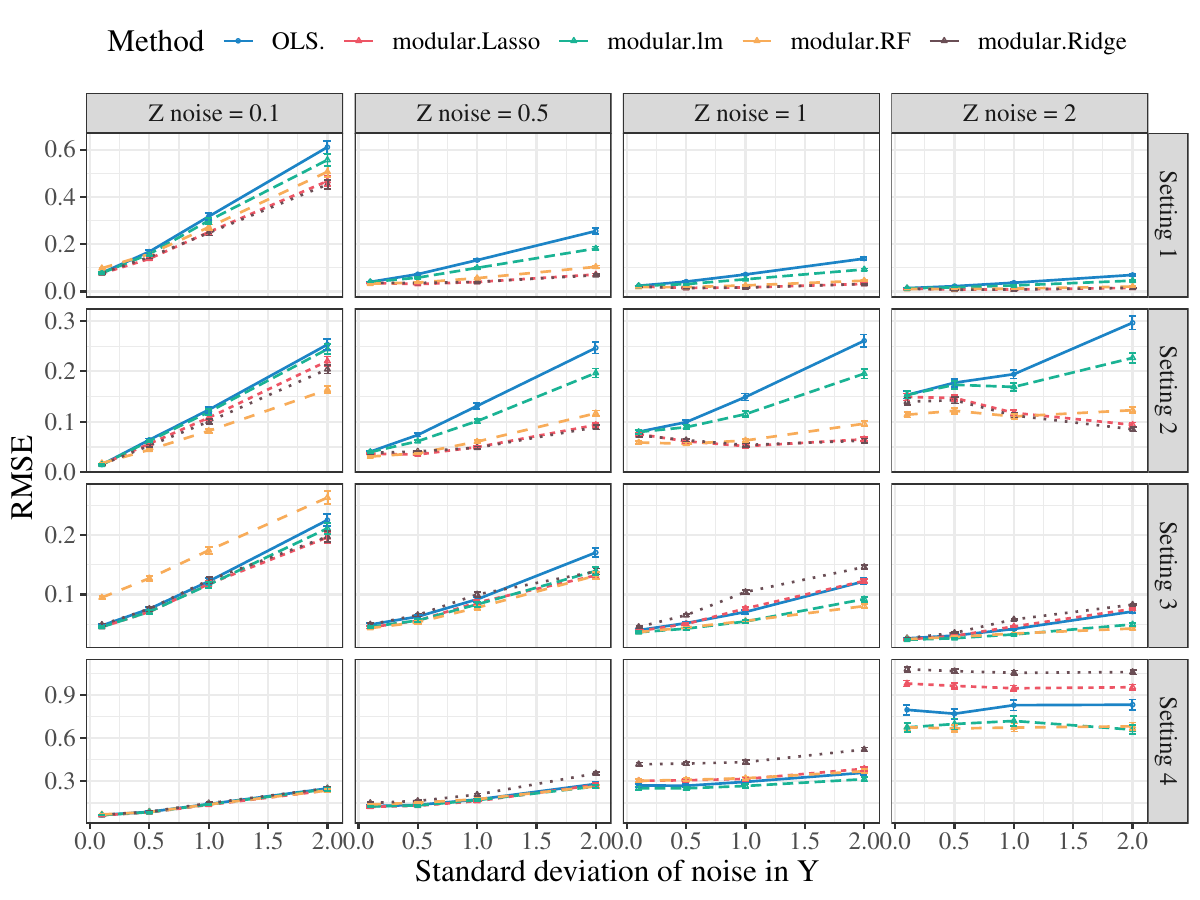}
\caption{RMSE for estimating $\theta^*_1$ in all settings. Details are the same as Figure~\ref{fig:lowd_rmse_sub}.}
\label{fig:rmse_lowd_1}
\end{figure}

\begin{figure}[h]
\centering
\includegraphics[width=5in]{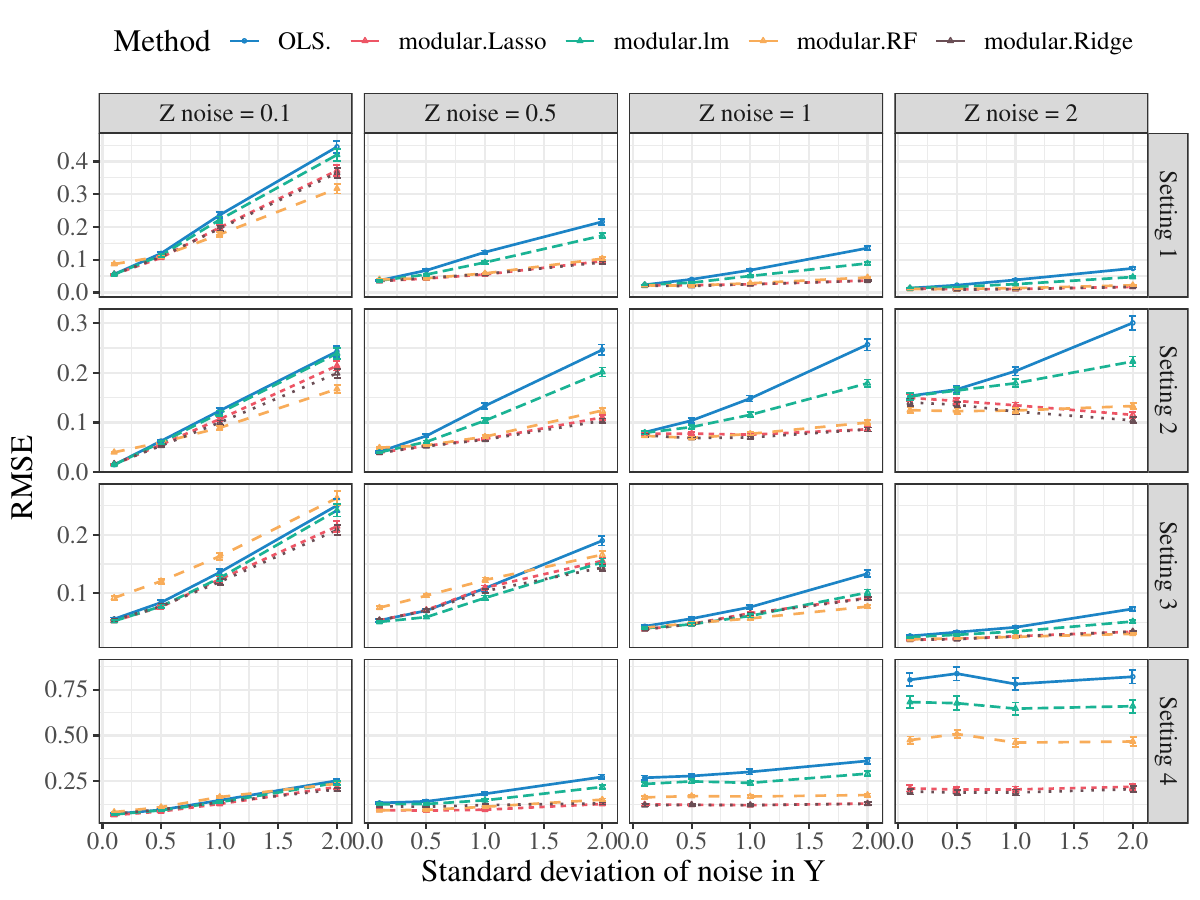}
\caption{RMSE for estimating $\theta^*_2$ in all settings. Details are the same as Figure~\ref{fig:lowd_rmse_sub}.}
\label{fig:rmse_lowd_2}
\end{figure}

\begin{figure}[h]
\centering
\includegraphics[width=5in]{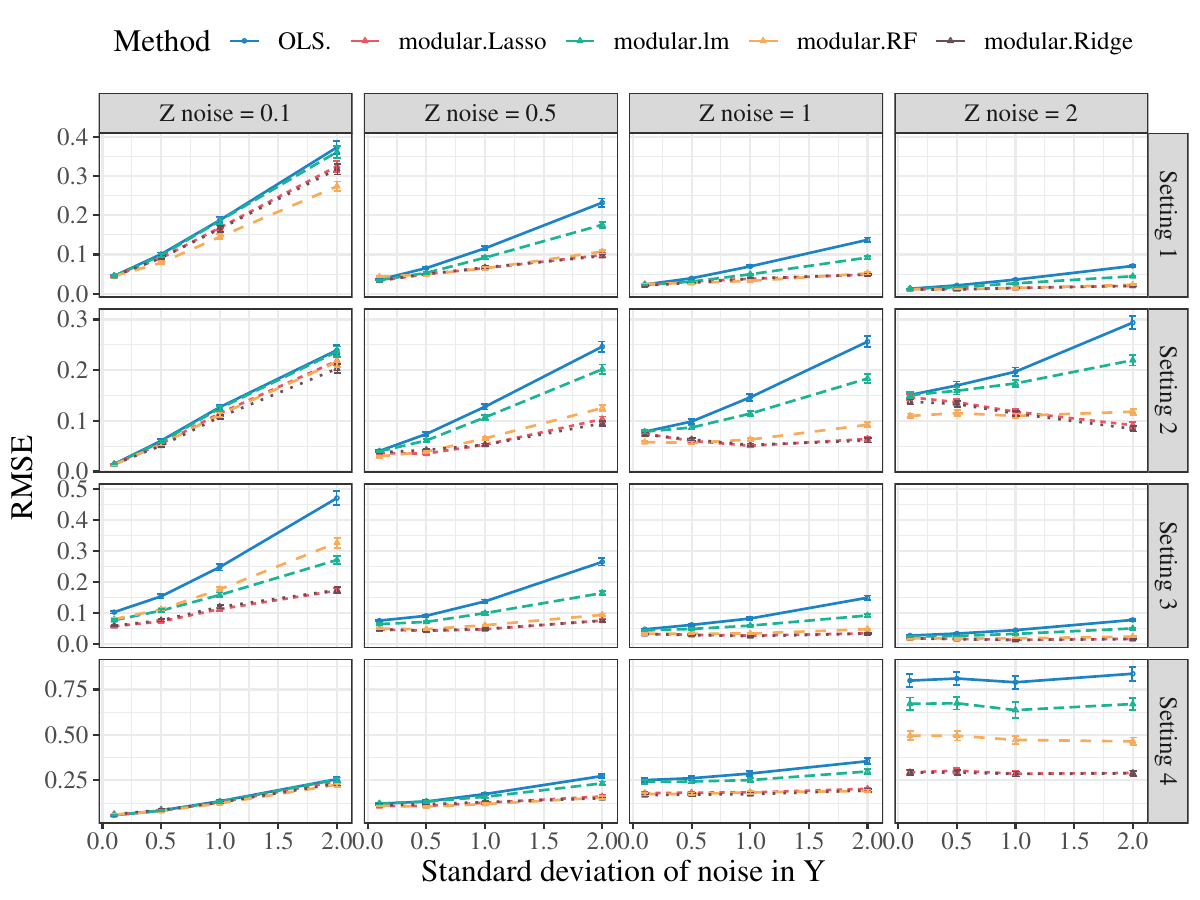}
\caption{RMSE for estimating $\theta^*_3$ in all settings. Details are the same as Figure~\ref{fig:lowd_rmse_sub}.}
\label{fig:rmse_lowd_3}
\end{figure}

\begin{figure}[h]
\centering
\includegraphics[width=5in]{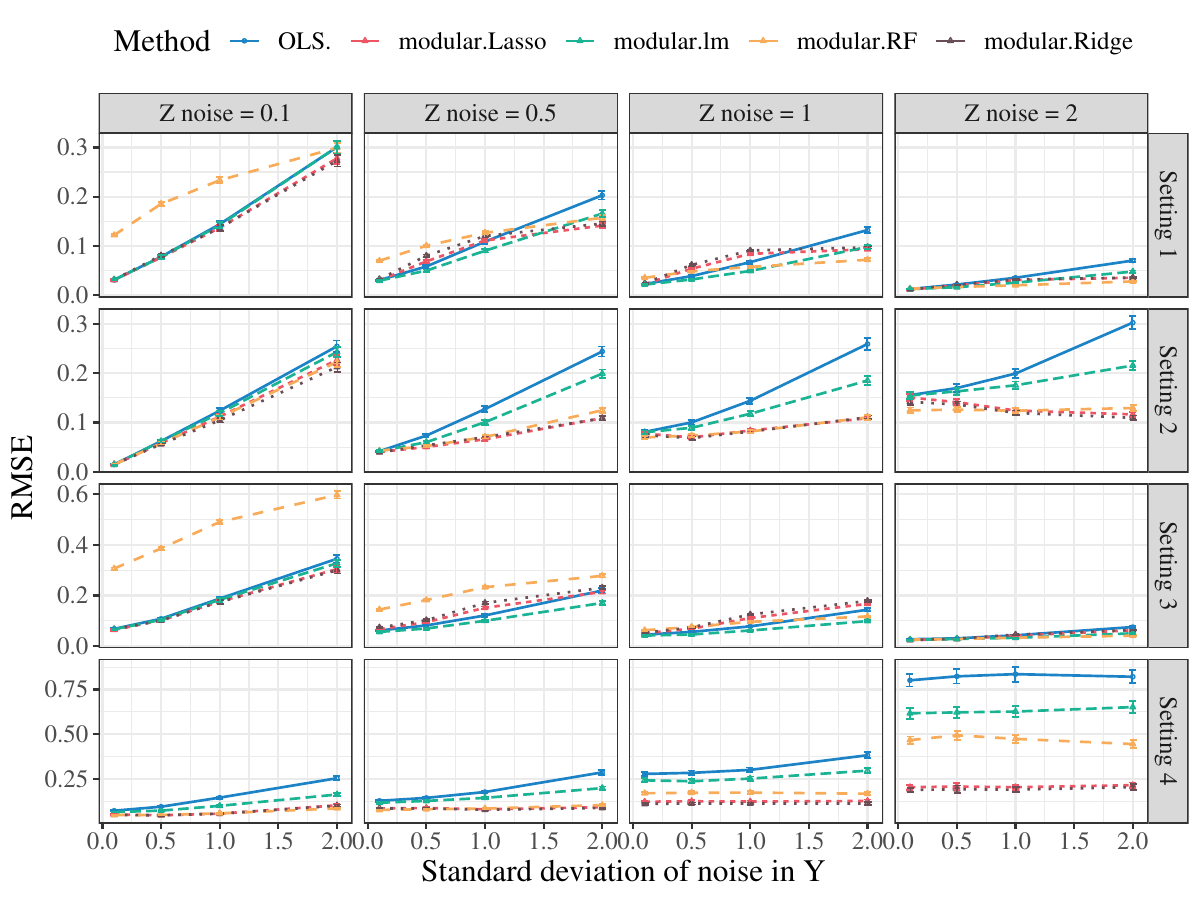}
\caption{RMSE for estimating $\theta^*_4$ in all settings. Details are the same as Figure~\ref{fig:lowd_rmse_sub}.}
\label{fig:rmse_lowd_4}
\end{figure}

\subsection{High-dimensional simulation}

\begin{figure}[h]
    \centering
    \begin{subfigure}[t]{0.48\linewidth}
        \centering
        \includegraphics[width=0.95\linewidth]{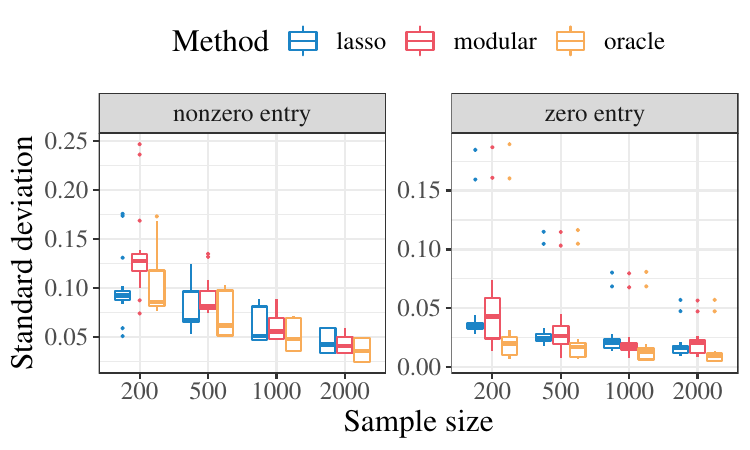}
    \end{subfigure} 
    \begin{subfigure}[t]{0.48\linewidth}
        \centering
        \includegraphics[width=0.95\linewidth]{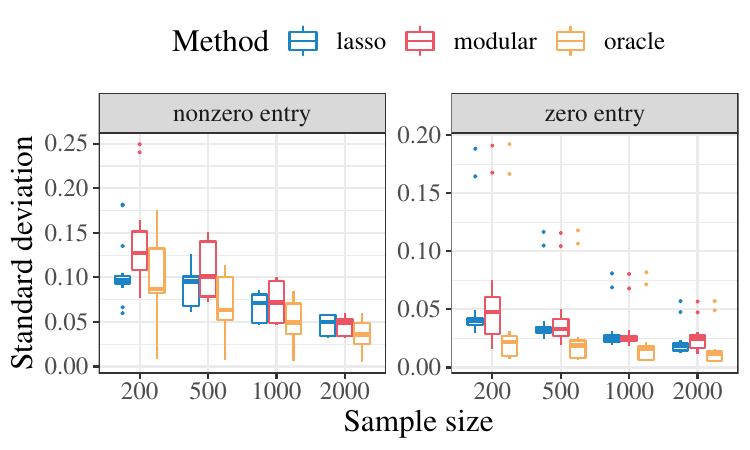}
    \end{subfigure} 
\caption{Boxplot of standard deviations for $\{\hat\theta_j\colon j\in[p_x]\}$, averaged
over $N=1000$ replicates in setting 1 (left) 
and setting 2 (right). Other details  are the same as Figure~\ref{fig:simu_1}. }
\label{fig:simu_sd}
\end{figure}